\documentclass[acmsmall,screen]{acmart}

\setcopyright{rightsretained}
\acmPrice{}
\acmDOI{10.1145/3408988}
\acmYear{2020}
\copyrightyear{2020}
\acmSubmissionID{icfp20main-p83-p}
\acmJournal{PACMPL}
\acmVolume{4}
\acmNumber{ICFP}
\acmArticle{106}
\acmMonth{8}

\makeatletter
\renewcommand\paragraph{\@startsection{paragraph}{4}{\z@}%
                                    {1.25ex \@plus1ex \@minus.2ex}%
                                    {-0.5em}%
                                    {\normalfont\normalsize\bfseries\@adddotafter}}
\makeatother

\usepackage{booktabs}   
\usepackage{subcaption} 
\usepackage[T1]{fontenc}
\usepackage{graphicx}
\usepackage{array}
\usepackage{xcolor}
\usepackage{xspace}
\usepackage{stmaryrd}
\usepackage[inference]{semantic}
\usepackage{nanoml}
\usepackage{algorithm,algorithmicx}
\usepackage[noend]{algpseudocode}
\usepackage{mathtools}
\usepackage{mathrsfs}
\usepackage{mathpartir}
\usepackage{thmtools}
\usepackage{multirow}
\usepackage[scaled=0.75]{beramono}
\usepackage{paralist}

\usepackage{calc}
\usepackage{ifthen}

\newcommand{\li}[2][XX]{%
   \ifthenelse{\equal{#1}{XX}}%
      {\ensuremath{L(#2)}}%
            {\ensuremath{L^{#1}(#2)}}}

\newlength{\rWidth}

\newcommand{\fun}[3][XX]{%
   \ifthenelse{\equal{#1}{XX}}%
      {\ensuremath{#2 {\rightarrow} #3}}%
      { \settowidth{\rWidth}{\ensuremath{#1}}%
        \ensuremath{ #2\hspace{.1em} {\xrightarrow{\hspace{\rWidth}}\hspace{-1.1\rWidth}}^%
         {#1}
         \hspace{0.3\rWidth}\hspace{.1em} #3}}}

\newcommand{\Rule}[4][]{\ensuremath{\inferrule*[lab={\footnotesize{(#2)}},#1]{#3}{#4}}}

\newcommand{\p}[1]{\ensuremath{(#1)}}
\newcommand{\sq}[1]{\ensuremath{[#1]}}
\newcommand{\set}[1]{\ensuremath{\{#1\}}}
\newcommand{\tuple}[1]{\ensuremath{\langle#1\rangle}}
\newcommand{\interp}[1]{\ensuremath{\llbracket #1 \rrbracket}}
\newcommand{\interps}[1]{\ensuremath{\llparenthesis #1 \rrparenthesis}}

\newcommand{\defeq}{\overset{\underset{\textnormal{def}}{}}{=}}

\newcommand{\bbB}{\ensuremath{\mathbb{B}}}
\newcommand{\bbN}{\ensuremath{\mathbb{N}}}
\newcommand{\bbZ}{\ensuremath{\mathbb{Z}}}
\newcommand{\bbU}{\ensuremath{\mathbb{U}}}
\newcommand{\calI}{\ensuremath{\mathcal{I}}}
\newcommand{\calJ}{\ensuremath{\mathcal{J}}}
\newcommand{\calT}{\ensuremath{\mathcal{T}}}
\newcommand{\scrB}{\ensuremath{\mathscr{B}}}

\newcommand{\subst}[3]{\ensuremath{\sq{#1/#2}#3}}
\newcommand{\sharing}{\ensuremath{\mathrel{\curlyveedownarrow}}}
\newcommand{\potv}[3]{\ensuremath{\Phi_{#1}\p{#2 : #3}}}
\newcommand{\potc}[2]{\ensuremath{\Phi_{#1}\p{#2}}}
\newcommand{\pot}[1]{\ensuremath{\Phi\p{#1}}}
\newcommand{\condv}[3]{\ensuremath{\Psi_{#1}\p{#2:#3}}}
\newcommand{\condc}[2]{\ensuremath{\Psi_{#1}\p{#2}}}

\newcommand{\etriv}{\ensuremath{\mathsf{triv}}}
\newcommand{\enat}[1]{\ensuremath{\overline{#1}}}

\newcommand{\etrue}{\ensuremath{\mathsf{true}}}
\newcommand{\efalse}{\ensuremath{\mathsf{false}}}
\newcommand{\econd}[3]{\ensuremath{\mathsf{if}\p{#1,#2,#3}}}
\newcommand{\elet}[3]{\ensuremath{\mathsf{let}\p{#1,#2.#3}}}
\newcommand{\enil}{\ensuremath{\mathsf{nil}}}
\newcommand{\econs}[2]{\ensuremath{\mathsf{cons}\p{#1,#2}}}

\newcommand{\eabs}[2]{\ensuremath{\lambda\p{#1.#2}}}
\newcommand{\efix}[3]{\ensuremath{\mathsf{fix}\p{#1.#2.#3}}}
\newcommand{\eapp}[2]{\ensuremath{\mathsf{app}\p{#1,#2}}}
\newcommand{\eimp}{\ensuremath{\mathsf{impossible}}\xspace}
\newcommand{\econsumename}{\ensuremath{\mathsf{tick}}\xspace}
\newcommand{\econsume}[2]{\ensuremath{\econsumename\p{#1,#2}}}

\newcommand{\epair}[2]{\ensuremath{\mathsf{pair}\p{#1,#2}}}
\newcommand{\ematp}[4]{\ensuremath{\mathsf{matp}\p{#1, #2.#3.#4}}}
\newcommand{\tprod}[2]{\ensuremath{#1 \times #2}}

\newcommand{\jval}[1]{\ensuremath{#1 \in \mathsf{Val}}}
\newcommand{\jstep}[4]{\ensuremath{\tuple{#1,#3} \mapsto \tuple{#2,#4}}}
\newcommand{\jsteps}[4]{\ensuremath{\tuple{#1,#3} \mapsto^* \tuple{#2,#4}}}
\newcommand{\jstepn}[5][n]{\ensuremath{\tuple{#2,#4} \mapsto^{#1} \tuple{#3,#5}}}
\newcommand{\jsort}[3]{\ensuremath{#1 \vdash #2 \in #3}}
\newcommand{\jwftype}[2]{\ensuremath{#1 \vdash #2~\mathsf{type}}}
\newcommand{\jwfctxt}[1]{\ensuremath{\vdash #1~\mathsf{context}}}
\newcommand{\jsubty}[3]{\ensuremath{#1 \vdash #2 <: #3}}
\newcommand{\jsharing}[4]{\ensuremath{#1 \vdash #2 \sharing #3 \mid #4}}
\newcommand{\jctxsharing}[3]{\ensuremath{\vdash #1 \sharing #2 \mid #3}}
\newcommand{\jprop}[2]{\ensuremath{#1 \models #2}}
\newcommand{\jatyping}[3]{\ensuremath{#1 \vdash #2 : #3}}
\newcommand{\jstyping}[3]{\ensuremath{#1 \vdash #2 \dblcolon #3}}
\newcommand{\jctxtyping}[3][\cdot]{\ensuremath{#1 \vdash #2 \dblcolon #3}}
\newcommand{\jkind}[2]{\ensuremath{ #1 \rightsquigarrow #2}}

\newcommand{\jtunfoldnil}[3]{\ensuremath{#1 \vdash #2 \triangleleft^\mathsf{nil} #3}}
\newcommand{\jtunfoldcons}[3]{\ensuremath{#1 \vdash #2 \triangleleft^\mathsf{cons} #3}}

\newcommand{\many}[1]{\overrightarrow{#1}}
\newcommand{\vbind}[2]{\ensuremath{#1:#2}}

\newcommand{\tunit}{\ensuremath{\mathsf{unit}}}
\newcommand{\tbool}{\ensuremath{\mathsf{bool}}}
\newcommand{\tnat}{\ensuremath{\mathsf{nat}}}
\newcommand{\tlist}[1]{\ensuremath{L\p{#1}}}

\newcommand{\tinduct}[5][\lhd,\pi]{\ensuremath{\mathsf{ind}_{#1}^{#5}\p{ \many{ #2 {\,:\,} (#3,#4) } }}}
\newcommand{\tpot}[2]{\ensuremath{{#1}^{#2}}}
\newcommand{\tsubset}[2]{\ensuremath{\set{ #1 \mid #2 }}}
\newcommand{\tarrow}[3]{\ensuremath{#1{\,:\,}#2 \rightarrow #3}}
\newcommand{\tarrowm}[4]{\ensuremath{{#4}\cdot\p{#1{\,:\,}#2 \rightarrow #3}}}
\newcommand{\trefined}[3]{\tpot{\tsubset{#1}{#2}}{#3}}

\newcommand{\sprod}[2]{\ensuremath{#1 \times #2}}
\newcommand{\sarrow}[2]{\ensuremath{#1 \Rightarrow #2}}
\newcommand{\rabs}[3]{\ensuremath{\lambda #1{\,:\,}#2.#3}}
\newcommand{\rpair}[2]{\ensuremath{\p{#1,#2}}}
\newcommand{\rforall}[3]{\ensuremath{\forall #1{\,:\,}#2.#3}}
\newcommand{\projl}[1]{\ensuremath{#1.\mathbf{1}}}
\newcommand{\projr}[1]{\ensuremath{#1.\mathbf{2}}}

\newcommand{\bindvar}[2]{\vbind{#1}{#2}}

\newcommand{\tscalar}[1]{\ensuremath{#1~\mathsf{scalar}}}

\algrenewcommand\algorithmicrequire{\textbf{Input:}}
\algrenewcommand\algorithmicensure{\textbf{Output:}}

\newcommand{\Subt}{<:}

\newcommand{\env}{\Gamma}

\newcommand{\eg}{\textit{e.g.}\@\xspace}
\newcommand{\ie}{\textit{i.e.}\@\xspace}

\newcommand{\tname}[1]{\textsc{#1}\xspace}
\newcommand{\tool}{\tname{LRTChecker}}
\newcommand{\resyn}{\tname{ReSyn}}
\newcommand{\raml}{\tname{RaML}}
\newcommand{\relcost}{\tname{RelCost}}
\newcommand{\synquid}{\tname{Synquid}}
\newcommand{\typesys}{$\mathrm{Re}^2$\xspace}

\newcommand{\Omit}[1]{}

\newcommand{\numBench}{12\xspace}

\newcommand{\tktype}[1]{\ensuremath{\mathsf{#1}}}
\newcommand{\tkptype}[2]{\ensuremath{\mathsf{#1}^{#2}}}
\newcommand{\tklist}[1]{\ensuremath{\mathsf{List \, #1}}}
\newcommand{\tkplist}[2]{\ensuremath{\mathsf{List \, #1}^{#2}}}
\newcommand{\tkaplist}[3]{\ensuremath{\mathsf{List \, #1}^{#2} \, \langle #3 \rangle}}
\newcommand{\tkapdt}[4]{\ensuremath{\mathsf{#4 \, #1}^{#2} \, \langle #3 \rangle}}
\newcommand{\tkadt}[3]{\ensuremath{\mathsf{#3 \, #1} \, \langle #2 \rangle}}
\newcommand{\tksym}[1]{\textcolor{blue}{#1}}

\newif\iflong
\longtrue

\begin{document}

\title{Liquid Resource Types}

\author{Tristan Knoth}
\affiliation{
  \institution{University of California, San Diego}
  \country{USA}
}
\email{tknoth@ucsd.edu}

\author{Di Wang}
\affiliation{
  \institution{Carnegie Mellon University}
  \country{USA}
}
\email{diw3@cs.cmu.edu}

\author{Adam Reynolds}
\affiliation{
  \institution{University of California, San Diego}
  \country{USA}
}
\email{acreynol@ucsd.edu}

\author{Jan Hoffmann}
\affiliation{
  \institution{Carnegie Mellon University}
  \country{USA}
}
\email{jhoffmann@cmu.edu}

\author{Nadia Polikarpova}
\affiliation{
  \institution{University of California, San Diego}
  \country{USA}
}
\email{npolikarpova@ucsd.edu}

\begin{abstract}
This article presents \emph{liquid resource types}, 
a technique for automatically verifying the resource consumption
of functional programs.
Existing resource analysis techniques trade automation for flexibility --
automated techniques are restricted to relatively constrained
families of resource bounds, while more expressive proof techniques
admitting value-dependent bounds rely on handwritten proofs.
Liquid resource types combine the best of these approaches,
using logical refinements to automatically prove precise
bounds on a program's resource consumption.
The type system augments refinement types with potential annotations
to conduct an amortized resource analysis.
Importantly, users can annotate data structure declarations to
indicate how potential is allocated within the type,
allowing the system to express bounds with polynomials and exponentials,
as well as more precise expressions depending on program values.
We prove the soundness of the type system,
provide a library of flexible and reusable data structures
for conducting resource analysis, and use our prototype
implementation to automatically verify resource bounds
that previously required a manual proof.
\end{abstract}



\maketitle

\section{Introduction}
\label{sec:intro}

Open any algorithms textbook and one will read about
a number of sorting algorithms, all functionally equivalent.
Why then, are there so many algorithms that do the same thing?
The answer is that there are subtle differences in their performance characteristics.
Consider, for example, the choice between quicksort and insertion sort.
In the worst case, both algorithms run in quadratic time.
Insertion sort, however, only needs to move the values
that are out of place, 
so it can perform much better on mostly-sorted data.

\paragraph{Resource analysis}
Choosing between implementations of seemingly simple functions like
these requires precise resource analysis.
Thus, there has been a lot of existing work in
both inferring and verifying bounds on a program's resource consumption.
In general existing approaches must trade automation
for flexibility and precision.

On one end of the spectrum,
Resource-Aware ML (\raml) \cite{RAML10} automatically infers polynomial 
bounds on recursive programs
by allocating \emph{potential} amongst data structures.
\raml{} reduces least upper bound inference to finding a minimal
solution to a system of linear constraints corresponding
to the program's resource demands.
On the other hand,
\relcost \cite{Radicek18} offers greater flexibility 
at the expense of automation.
\relcost allows users to prove precise resource bounds 
that depend on program values,
but requires hand-written proofs.

For example, consider insertion sort:
\autoref{fig:insert} shows a recursive implementation of this sorting algorithm in a functional language.
In this example we adopt a simple cost model where recursive calls
incur unit cost, and all other operations do not require resources;
we indicate this by wrapping recursive calls in a special operation $\mathsf{tick}$,
which consumes a given amount of resources. 
\raml{} can infer a tight quadratic bound on the cost
of evaluating $\mathsf{sort}:\, 0.5(n^2 + n)$, where $n$
is the length of the input list.
\relcost allows one to prove a more complex bound:
insertion sort requires resources proportional to the number
of out-of-order pairs in the input.
However, the proof must be written by hand.
\emph{Is it possible to develop a technique that admits both automation
and expressiveness and can automatically verify these kinds of
fine-grained bounds?}

\begin{figure}
\begin{minipage}{0.45\textwidth}
\begin{nanoml}
insert = \x. \xs.
  match xs with 
    Nil -> Cons x xs
    Cons hd tl -> if hd < x
      then Cons hd (tick 1 (insert x tl))
      else Cons x (Cons hd tl)
\end{nanoml} 
\end{minipage}
\begin{minipage}{0.45\textwidth}
\begin{nanoml}
sort = \xs. 
  match xs with
    Nil -> Nil
    Cons hd tl -> 
      insert hd (tick 1 (sort tl))
\end{nanoml} 
\end{minipage}

\caption{Insertion sort}
\label{fig:insert}
\end{figure}

\paragraph{Liquid Types and Resources}

Recent work on \resyn~\cite{resyn} takes a first step in this direction by
extending a \emph{liquid type system} with resource analysis.
Liquid types~\cite{RondonKaJh08} support automatic verification of nontrivial 
functional properties with an SMT solver.
\resyn augments an existing liquid type system~\cite{PolikarpovaKS16}
with a single construct: types can be annotated with
a numeric quantity called \emph{potential}.
For example, a value of type \tkptype{Int}{1} carries a single unit of potential,
which can be used to pay for an operation with unit cost.
Combined with polymorphic datatypes, this mechanism can describe
uniform assignment of potential to the elements of a data structure.
For example, instantiating a polymorphic list type \tklist{a} with $\tktype{a} \mapsto \tkptype{Int}{1}$  
yields \tkplist{Int}{1}, a type of lists where every element has a single unit of potential.

The \resyn type checker verifies that a
program has enough potential to pay for all operations
that may occur during evaluation.
For example, \resyn can check the implementation of \T{insert} in \autoref{fig:insert} against the (polymorphic) type
$\tarrow{x}{a}{\tarrow{xs}{\tkplist{a}{1}}{\tklist{a}}}$
to verify that the function makes one recursive call per element in the input \T{xs}.
Here $\tkplist{a}{1}$ stands for the type of lists where each element has one more unit of potential than prescribed by type \tktype{a}.

More interestingly, the combination of refinements and
potential annotations allows \resyn to verify \emph{value-dependent} resource bounds.
To this end, \resyn supports the use of conditional linear arithmetic (CLIA) terms
as potential annotations, as opposed to just constants.
%
For example, \resyn can also check \T{insert} against the type
$\tarrow{x}{a}{\tarrow{xs}{\tkplist{a}{\mathsf{ite}(x > \nu, 1, 0)}}{\tklist{a}}}$,
which states that \T{insert} only makes a recursive call
for each element in \T{xs} smaller than \T{x}.
The annotation on the type of the list elements conditionally assigns
potential to a value in the list only when it is smaller than \T{x}%
\footnote{Throughout the paper, the special variable $\nu$ refers to an arbitrary inhabitant of the annotated type.}.
\resyn reduces this type checking problem to a system of second-order CLIA constraints,
which can be solved relatively efficiently using existing program synthesis techniques~\cite{alur2013}. 

\paragraph{Challenge: Analyzing super-linear bounds}

A major limitation of the \resyn type system is that it only supports \emph{linear bounds}.
In particular, a type of the form $\tkplist{a}{p}$ distributes the potential $p$ \emph{uniformly} throughout the list,
%
and hence cannot express resource consumption
of a super-linear function like insertion sort, 
which traverses the end of the input list \emph{more often} than the beginning 
(recall that insertion sort recursively sorts the tail of the list 
and traverses the newly sorted tail again to insert an element).
To verify this function, we need a type that allots more potential
to elements in the tail of a list than the head.
In this paper, we propose two simple extensions to the \resyn 
type system to support the verification of super-linear resource bounds, 
while still generating only second-order CLIA constraints
to keep type checking efficiently decidable.

\paragraph{Super-linear Resource Analysis with Inductive Potentials}

Our first insight is that we can describe non-uniform allocation of potential 
in a data structure
by embedding potential annotations \emph{into datatype definitions}.
We dub this mechanism \emph{inductive potentials}.
%
%
For example, the datatype \tktype{QList} in \autoref{fig:inductive-lists} (left) 
represents lists where every element has one more unit of potential than the one before it
(the total amount of potential in the list is thus \emph{quadratic} in its length).
We express this non-uniform distribution of potential with the type of \tktype{QCons}:
the elements in the tail of the list are of type \tkptype{a}{1} instead of \T{a},
indicating that they must contain one more unit of potential than the head does.
The datatype \tktype{ISList} in \autoref{fig:inductive-lists} (right)
is similar, but only assigns extra potential to those elements of the tail that are smaller than the head. 
Using these custom datatypes we can specify a coarse-grained (with \tktype{QList}) 
and fine-grained (with \tktype{ISList}) resource bound for insertion sort.
Importantly, all potential annotations are still expressed in CLIA, 
so we can verify super-linear resource bounds while reusing \resyn's constraint-solving infrastructure. 

\begin{figure}
\begin{minipage}{0.42\textwidth}
\begin{nanoml}
data QList a where
  QNil :: QList a
  QCons :: a -> QList a$^1$ -> QList a
\end{nanoml} 
\end{minipage}%
\begin{minipage}{0.58\textwidth}
\begin{nanoml}
data ISList a where
  ISNil :: ISList a
  ISCons :: x:a -> xs:ISList a$^{\mathsf{ite}(x > \nu,1,0)}$ -> ISList a
\end{nanoml} 
\end{minipage} 

  \caption{Two list types defined with inductive potentials: 
  \tktype{QList} carries quadratic potential;
  in \tktype{ISList}, elements in the tail only have potential when they are larger than the head.}
  \label{fig:inductive-lists}
\end{figure}

\paragraph{Flexibility via Abstract Potentials}

Inductive potentials, as descried so far, are somewhat restrictive. 
One must define a custom datatype for every resource bound.
In the insertion sort example, we had to define \tktype{QList} to perform a coarse-grained analysis
and \tktype{ISList} to perform a fine-grained analysis;
moreover, both types have a fixed constant $1$ embedded in their definition, 
so if the cost of \T{tick} inside \T{insert} were to increase,
these types would no longer work.
This is clearly unwieldy:
instead, we would like to be able to write \emph{libraries}
of reusable data structures,
each able to express a broad family of resource bounds.

To address this limitation,
our second insight is to \emph{parameterize} datatypes 
by numeric logic-level functions, 
which can then be used inside the datatype definition to allocate potential.
We dub this second type system extension \emph{abstract potentials}.
With abstract potentials, the programmer can define a single datatype
that represents a family of resource bounds,
and then instantiate it with appropriate potential functions to verify different concrete bounds.
%
For example, instead of defining \T{QList} and \T{ISList} separately,
we can define a more general type
$\tkaplist{a}{}{q \dblcolon a \rightarrow a \rightarrow \mathsf{Nat}}$,
where the parameter $q$ abstracts over the potential annotation in the constructor.
We can then instantiate $q$ with different logic-level functions 
to perform different analyses.
Importantly, type checking still generates constraints in the same logic fragment as \resyn. 
%
This design enables our type checker to automate resource analyses that 
would have previously required a handwritten proof. 

\paragraph{Contributions}
In summary, this paper the following technical contributions:

\begin{enumerate}
  \item \emph{Liquid resource types} (LRT), 
  a flexible type system for automatic resource analysis. 
  With inductive and abstract potentials,
  programmers can analyze a variety of resource bounds 
  by specifying how potential is allocated within a data structure.
  \item \emph{Semantics and a soundness proof} for the type system, including
  user-defined inductive data types. 
  \item A \emph{prototype implementation}, \tool, 
  that automatically checks precise value-dependent
  resource bounds with existing constraint solving technology.
  \item A \emph{library of data types} corresponding to families of resource bounds,
  such as lists admitting polynomial or exponential bounds over their length,
  and trees admitting linear combinations of their size and height.
  \item An \emph{evaluation} on a set of challenging examples showing that \tool
  automatically performs resource analyses out of scope of prior approaches.
\end{enumerate}

\section{Overview}
\label{sec:overview}

We begin with examples to better illustrate how
liquid resource types enable the automatic verification of precise 
resource bounds.
First, we show how \resyn integrates resource analysis into a liquid type system.
Second, we show how inductive potentials enable the analysis of 
super-linear bounds.
Finally, we show how abstract potentials make this paradigm flexible
and reusable.

\subsection{Background: \resyn}
\label{sec:overview:resyn}

\paragraph{Liquid Types}
In a refinement type system~\cite{fstar,Denney99}, 
types are annotated with logical predicates that constrain the range of their values.
For instance, the type of natural numbers can be expressed as \T{type Nat = \{Int | _v >= 0\}},
where the special variable $\nu$, as before, denotes an inhabitant of the type.
Liquid types~\cite{RondonKaJh08,VazouRoJh13} are a kind of refinement types
that restrict logical refinements 
to only appear on \emph{scalar} (\ie non-function) types,
and be expressed in decidable logics.
Due to these restrictions, liquid types support fully automatic verification
of nontrivial functional properties with the help of an SMT solver.

\paragraph{Potential Annotations}
\resyn~\cite{resyn} extends liquid types with the ability 
to reason about the resource consumption of programs in addition to their functional properties.
To this end, 
a type can also be annotated with a numeric logic expression called \emph{potential},
as well as a logical refinement.
For example, the type $\tkptype{Nat}{1}$ ranges over natural numbers that carry a single unit of potential.
Intuitively, potential can be used to ``pay'' for evaluating special \T{tick} terms,
which are placed throughout the program to encode a cost model.
For example, the context $[\T{x} : \tkptype{Nat}{1}]$ has a total of 1 unit of \emph{free potential}, 
which is sufficient to type-check a term like \T{tick 1 ()}.
Because duplicating potential would lead to unsound resource analysis,
\resyn's type system is \emph{affine},
which means that creating two copies of a context---for example, to type-check both sides of an application---%
requires distributing the available potential between them.

Simple potential annotations can be combined with other features of the type system,
such as polymorphic datatypes,
to specify more complex allocation of resources.
For example, instantiating a polymorphic datatype \tklist{a} with $\T{a} \mapsto \tkptype{Nat}{1}$
yields the type \tkplist{Nat}{1} of natural-number lists that carry one unit of potential per element.
Here and throughout the paper, a missing potential annotation defaults to zero,
so the type above stands for $(\tkplist{Nat}{1})^0$.
This default annotation hints at our more general notion of type substitution,
where potential annotations are added together:
instantiating a polymorphic datatype \tkplist{a}{m} with $\T{a} \mapsto \tkptype{Nat}{n}$
yields the type \tkplist{Nat}{m+n}.

Note that only ``top-level'' potential in a type contributes to the free potential of the context:
for example, the context $[\T{xs} : \tkplist{Nat}{1}]$ has no free potential
(which makes sense, since \T{xs} could be empty).
The potential bundled inside an inductive datatype
can be freed via pattern matching:
for example, matching the \T{xs} variable above against \T{Cons hd tl}
extends the context with new bindings $\T{hd} \dblcolon \mathsf{Nat}^1$ and $\T{tl} \dblcolon \tkplist{Nat}{1}$;
this new context has a single unit of free potential attached to \T{hd}
(which also makes sense, since we now know that \T{xs} had at least one element).


Using potential annotations and \T{tick} terms,
\resyn is able to specify upper bounds on resource consumption of recursive functions.
Consider, for example, the function \T{insert} that inserts a value
into a sorted list \T{xs}, as shown in \autoref{fig:insert} (left).
We wish to check that \T{insert} traverses the list linearly:
more precisely, that it only makes a single recursive call per list element.
To this end, we wrap the recursive call in a \T{tick} with unit cost,
and annotate \T{insert} with the following type signature,
which allocates one unit of potential per element of the input list:
\[ \T{insert} \dblcolon \tarrow{x}{\T{a}}{\tarrow{xs}{\tkplist{a}{1}}{\tklist{a}}} \]
%


\begin{figure}
\begin{minipage}{.55\textwidth}
\begin{nanoml}[numbers=left]
[insert: $\tarrow{x}{\mathsf{a}}{\tarrow{xs}{\tkplist{a}{\tksym{P}}}{\tklist{a}}}$]
[insert:..., x:a, xs:$\tkplist{a}{\tksym{P}}$]
[insert:..., x:a, xs:$\tkplist{a}{\tksym{P}}$]
[insert:..., x:a, xs:$\tklist{a}$]
[insert:..., x:a, xs:$\tklist{a}$, hd:a$^{\tksym{P}}$, tl:$\tkplist{a}{\tksym{P}}$]
[insert:..., x:a, xs:$\tklist{a}$, hd:a$^{p_1}$, tl:$\tkplist{a}{q_1}$]
[insert:..., x:a, xs:$\tklist{a}$, hd:a$^{p_2}$, tl:$\tkplist{a}{q_2}$, @hd < x@]
[insert:..., x:a, xs:$\tklist{a}$, hd:a$^{p_2-1}$, tl:$\tkplist{a}{q_2}$, @hd < x@]
[insert:..., x:a, xs:$\tklist{a}$, hd:a$^{p_2}$, tl:$\tkplist{a}{q_2}$, @!(hs < x)@]
\end{nanoml}
\end{minipage}\vline%
\begin{minipage}{.38\textwidth}
  \begin{nanoml}[numbers=right]
  insert = \x. \xs.
    match 
      xs with 
        Nil -> Cons x Nil
        Cons hd tl -> 
          if hd < x
            then Cons hd (tick 1 
                           (insert x tl))
            else Cons x (Cons hd tl)
  \end{nanoml} 
\end{minipage}



  \caption{On the right, the implementation of \T{insert} alongside the contexts used for type checking.
  Each line of the program corresponds to a subexpression that generates resource constraints,
  with the typing context relevant for constraint generation alongside it to the left. The start of the \T{match}
  expression is split between two lines to separate the context used to type the entire \T{match} expression from 
  the context used to type the scrutinee. \tksym{P} is used as a symbolic resource annotation,
  as we will check this program against different bounds by providing concrete valuations for \tksym{P}.}
  \label{fig:insert-annotated} 
\end{figure}

\paragraph{Type checking}

We now describe how \resyn checks \T{insert} against this specification.
At a high level, type checking reduces to generating a system
of linear arithmetic constraints asserting that it is possible
to partition the potential available in the context amongst all expressions
that need to be evaluated.
If this system of constraints is satisfiable, 
the given resource bound is sufficient.
We generate three kinds of constraints:
\emph{sharing} constraints, which nondeterministically 
partition resources between subexpressions,
\emph{subtyping} constraints, which check that a given 
term has enough potential to be used in a given context,
and \emph{well-formedness} constraints, which assert that
potential annotations are non-negative.
%
%
%

\autoref{fig:insert-annotated} illustrates type-checking of \T{insert}:
its left-hand side shows the context in which various subexpressions are checked
(for now you can ignore the \emph{path constraints}, shown in red).
The annotations in the figure are abstract; we will use the same figure to
describe how we check both dependent and constant resource bounds.
For this first example, we set $\tksym{P} = 1$ in the top-level type annotation of \T{insert} --
we are checking that \T{insert} only makes one recursive call per element in \T{xs}.

The body of \T{insert} starts with a pattern match,
which requires distributing the resources in the context on line 2 
between the match scrutinee and the branches.
This context has no free potential, but it does have some bundled potential in \T{xs:$\tkplist{a}{1}$};
bundled potential also has to be shared between the two copies of the context,
since it could later be freed by pattern matching.
In this case, however, \T{xs} is not mentioned in either of the branches,
so for simplicity we elide the sharing constraints and assign all its potential to line 3, 
leaving $\tkplist{a}{1}$ in the context of the match scrutinee
and $\tkplist{a}{0}$ in the context of the branches.
Matching the scrutinee type $\tkplist{a}{1}$ against the type of the \T{Cons} constructor
introduces new bindings $\T{hd} \dblcolon \mathsf{a}^1$ and $\T{tl} \dblcolon \tkplist{a}{1}$ into the context:
now we have $1$ unit of free potential at our disposal,
as the input list has at least one element.

When checking the conditional,
we must again partition all available resources between the
guard and either of the two branches.
In particular, we partition the \T{hd} binding from line 5
into $\T{hd}:\T{a}^{p_1}$ and $\T{hd}:\T{a}^{p_2}$,
generating a \emph{sharing constraint} that reduces to $1 = p_1 + p_2$.
Similarly, we also partition the remaining potential in $\T{tl}$ into
$\T{tl}: \tkplist{a}{q_1}$ and $\T{tl}: \tkplist{a}{q_2}$, which 
produces a constraint $1 = q_1 + q_2$ preventing us from reusing potential
still contained in the list.
\resyn partitions resources non-deterministically and offloads
the work of finding a concrete partitioning to the constraint solver.
Neither the guard nor the \T{else} branch contains a \T{tick}
expression, so they generate only trivial constraints.
The \T{then} branch is more involved, as it does contain
a \T{tick} with a unit cost.
We must pay for this \T{tick} using the free potential $p_2$ on \T{hd}%
leaving $\T{hd}:\mathsf{a}^{p_2 - 1}$ in the context
when checking the expression inside the tick on line 8. 
Like all bindings in the context, this binding generates a \emph{well-formedness constraint} on its type,
which reduces to the arithmetic constraint $p_2 - 1 \geq 0$,
thereby implicitly checking that $p_2$ is sufficient to pay for the \T{tick}.

Finally, type-checking the application of \T{insert x} to \T{tl}
produces a \emph{subtyping constraint} between the actual and the formal argument types:
$\env \vdash \tkplist{a}{q_2} \Subt \tkplist{a}{1}$.
This in turn reduces to an arithmetic constraint $q_2 \geq 1$,
asserting that \T{tl} contains enough potential to execute the recursive call.

Now, consider the complete system of generated arithmetic constraints:
\[ \exists p_1, p_2, q_1, q_2 \in \mathbb{N}. \, 1 = p_1 + p_2 \land 1 = q_1 + q_2 \land p_2 - 1 \geq 0 \land q_2 \geq 1  \]
Though elided above, recall that all symbolic annotations are also required to be non-negative.
This system of constraints is satisfiable by setting $p_2, q_2 = 1$ and the rest of the unknowns to $0$,
which \resyn automatically infers using an SMT solver.

\paragraph{Value-dependent resource bounds}
\resyn also supports verification of dependent resource bounds.
We can use a logic-level conditional to give the following 
more precise bound for \T{insert}:
\[ \T{insert} \dblcolon \tarrow{x}{\mathsf{a}}{\tarrow{xs}{ \tkplist{a}{ \mathsf{ite}(x > \nu,1,0) }}{\tklist{a}}} \]
The dependent annotation on \T{xs} indicates that only those list elements smaller
than \T{x} carry potential, reflecting the fact that the implementation does not 
make any recursive calls once it has found the appropriate place
to insert \T{x}. 

Type checking proceeds similarly to the non-dependent case,
except that we set $\tksym{P} = \T{ite}(\T{x} > \nu, 1, 0)$ 
and treat all other symbolic potential annotations
as unknown \emph{logic-level terms} over the program variables (including the special variable $\nu$).
As a result, type checking generates 
second-order CLIA constraints,
which are universally quantified over the program variables,
and may contain assumptions on these variables,
derived from their logical refinements or from \emph{path constraints} of branching expressions.
For example, \autoref{fig:insert-annotated} shows in red the path constraints
derived from the conditional.
In particular, when checking the first branch, we can assume that $\T{hd < x}$ holds
and thus conclude that \T{hd} has potential $1$ in this branch and is able to pay the cost of \T{tick}.
When we check that an annotation is well-formed, we must also assume
that the relevant variable's logical refinements hold.
For example, to check that the annotation $p_2(x,\nu)$ on \T{hd} is non-negative
we must assert that $\nu = \T{hd}$.

More precisely, the full system of constraints (omitting irrelevant program variables) becomes:
\begin{align*}
\exists & p_1, p_2, q_1, q_2 \in \mathbb{N}\times\mathbb{N}\to\mathbb{N}. 
\forall \T{x}, \T{hd}, \nu. \\
& \mathsf{ite}(\T{x} > \nu,1,0) = p_1(\T{x}, \nu) + p_2(\T{x}, \nu)
&& \text{Sharing \T{hd} (line 5)} \\
& \land \mathsf{ite}(\T{x} > \nu,1,0) = q_1(\T{x}, \nu) + q_2(\T{x}, \nu) 
&& \text{Sharing \T{tl} (line 5)} \\
& \land (\nu = \T{hd} \land \T{hd} < x) \implies p_2(\T{x}, \nu) - 1 \geq 0 
&& \text{Well-formedness of \T{hd} (line 8)} \\
& \land \T{hd} < \T{x} \implies q_2(\T{x}, \nu) \geq \mathsf{ite}(\T{x} > \nu,1,0) 
&& \text{Subtyping of \T{tl} (from recursive call)}
\end{align*}
\resyn satisfies these constraints by setting
$p_2, q_2 = \lambda(\T{x}, \nu).\mathsf{ite}(\T{x} > \nu,1,0)$,
and the rest of the unknowns to to $\lambda(\T{x}, \nu).0$.
Synthesis of CLIA expressions is a well-studied problem \cite{alur2013,ReynoldsKTBD19},
and \resyn uses counterexample-guided inductive synthesis (CEGIS)~\cite{Solar-LezamaTBSS06}
to solve the particular form of constraints that arise.

\paragraph{Limitations}

While \resyn's type system enables the analysis
of the resource consumption of a wide variety of functions,
and can automatically check value-dependent resource bounds,
it still falls short of analyzing many useful programs.
The system only expresses linear bounds, which 
are sufficient for many data structure traversals,
but not sufficient for programs that compose several
traversals.
Thus, \resyn cannot check the resource consumption of \T{sort}.
We need a way to extend this technique to programs
with more complex recursive structure.
\resyn also formalizes the technique only for lists,
while we would like to be able to analyze 
programs that manipulate arbitrary algebraic data types.

\subsection{Our Contribution: Liquid Resource Types}
\label{sec:overview:inductive}

To address these limitations and enable verification of super-linear bounds,
this work extends the \resyn type system with two powerful mechanisms:
\emph{inductive potentials} allow the programmer to define inductively
how potential is allocated within a datatype,
while \emph{abstract potentials} support parameterizing datatype definitions by potential functions.
We dub the extended type system \emph{liquid resource types} (LRT).


\paragraph{Inductive Potentials}

Inductive potentials are expressed simply as potential annotations on constructors of a datatype.
\autoref{fig:inductive-lists} (left) shows a simple example of a datatype, \T{QList}, with inductive potentials.
Here the \T{QCons} constructor mandates that the tail of the list
\begin{inparaenum}[(a)]
\item carries at least one more unit of potential in each element than the head, and
\item is itself a \T{QList}.
\end{inparaenum}
As a result, the total potential in a value $L = [ a_1, a_2, \ldots, a_n ]$ of type $\T{QList}\ T$ is 
\emph{quadratic} in $n$ and 
given by the following expression (where $p$ is the potential of type $T$):

\[ \Phi(L) = \sum_i p + \sum_i \sum_{j > i} 1 = n p + \sum_i i = \frac{n(n + 2p -1)}{2} \]

%
We can now specify that insertion sort runs in quadratic time by giving it the type:
\[ \T{sort} \dblcolon \tarrow{\T{xs}}{\T{QList a}^1}{\T{List a}}\]
According to the formula above, this type assigns \T{xs} the total potential of $0.5(n^2 + n)$, 
which is precisely the bound inferred by \raml{},
as we mentioned in the introduction.
More interestingly, we can use \emph{value-dependent} inductive potentials
to specify a tighter bound for \T{sort},
by the replacing \T{QList} in the type signature above with \T{ISList} defined in \autoref{fig:inductive-lists} (right).
In an \T{ISList}, the elements in the tail only carry the extra potential when their value is less than the head.
Hence, the total potential stored in an $\T{ISList a}^1$ is equal to 
the number of list elements plus the number of \emph{out-of-order pairs} of list elements.
Verifying \T{sort} against this bound implies, for example,
that insertion sort behaves linearly on a fully sorted list (with no decreasing element pairs)
and takes the full $0.5(n^2 + n)$ steps on a list sorted in reverse order.


While inductive potentials are able to express non-linear bounds,
on their own, they are difficult to use:
the non-linear coefficient of a resource bound is built into the datatype definition,
and hence any slight change in the analysis or the cost model---such as changing the cost of a recursive call from 1 to 2---%
requires defining a new datatype. 
We would like to be able to reuse the \emph{structure} of these
types without relying on the precise potential annotations embedded within.



\begin{figure}
\begin{nanoml}
data List t <q::t -> t -> Nat> where
  Nil :: List t <q>
  Cons :: x:t -> xs:List t$^{\mathsf{q(x, \nu)}}$ <q> -> List t <q>
\end{nanoml}
\caption{A list datatype parameterized by a value-dependent, quadratic abstract potential.}\label{fig:abs-list}
\end{figure}

\paragraph{Abstract potentials}
To make inductive potentials reusable, 
we introduce the second new feature of LRT, which we dub \emph{abstract potentials}.
This feature is inspired by abstract refinement types~\cite{VazouRoJh13}, 
which parameterize datatypes by a refinement predicate; 
similarly, LRT allows parameterizing a datatype a potential function.
Consider the definition of the \T{List} datatype in \autoref{fig:abs-list}:
this datatype is parameterized by a numeric logic-level function $q$,
which represents the additional potential 
contained in every element of every proper suffix of the list.
This interpretation is revealed in the \T{Cons}
constructor, where the value $q(x, \nu)$ is \emph{added}
to the linear potential annotation on the tail of the list.
Note that since $q$ is a function, 
this datatype subsumes both \T{QSort} and \T{ISSort},
as well as a broad range of value-dependent ``quadratic'' potential functions.
%
More precisely, if a list element $\nu$ of type $T$ carries $p(\nu)$ units of potential,
then the total potential in a list $L = [ a_1, a_2, \ldots, a_n ]$ of type $\T{List}\ T$ is given 
by the following formula:
\[ \Phi(L) = \sum_i p(a_i) + \sum_i \sum_{j > i} q(a_i, a_j) \]
Note that we can add higher-arity abstract potentials
to extend the \T{List} datatype to support higher-degree polynomials.
Similarly, we can add a unary abstract potential $p(\nu)$
to express the linear component of the list potential more explicitly
(as opposed to relying on polymorphism in the type of the elements). 

\begin{figure}
\begin{minipage}{.72\textwidth}
\begin{nanoml}[numbers=left]
[insert: $\forall \mathsf{b}.\tarrow{x}{\mathsf{b}}{\tarrow{xs}{\tkplist{b}{1}}{\tklist{b}}}$, sort: $\forall \mathsf{c}.\tarrow{xs}{\tkaplist{c}{1}{\tksym{Q}}}{\tklist{c}}$]
[insert, sort:..., xs:$\tkaplist{a}{1}{\tksym{Q}}$]
[insert, sort:..., xs:$\tkaplist{a}{1}{\tksym{Q}}$]
[insert, sort:..., xs:$\tklist{a}$]
[insert, sort:..., xs:$\tklist{a}$, hd:a$^{1}$, tl:$\tkaplist{a}{1 + \tksym{Q}(\mathsf{hd}, \nu)}{\tksym{Q}}$]
[insert, sort:..., xs:$\tklist{a}$, hd:a$^{p_1}$, tl:$\tkaplist{a}{q_1(\mathsf{hd}, \nu)}{q_1}$]
[insert, sort:..., xs:$\tklist{a}$, hd:a$^{p_2}$, tl:$\tkaplist{a}{q_2(\mathsf{hd}, \nu)}{q_2}$]
[insert, sort:..., xs:$\tklist{a}$, hd:a$^{p_2-1}$, tl:$\tkaplist{a}{q_2(\mathsf{hd}, \nu)}{q_2}$]
\end{nanoml}
\end{minipage}\vline%
\begin{minipage}{.23\textwidth}
\begin{nanoml}[numbers=right]
  sort = \xs.
    match 
      xs with
        Nil -> Nil
        Cons hd tl -> 
          insert hd 
            (tick 1 
              (sort tl))
\end{nanoml} 
\end{minipage}
\caption{ Similar to Figure~\autoref{fig:insert-annotated}, the evolution of the typing 
context while checking different subexpressions of \T{sort}.
\tksym{Q} is used as a symbolic resource annotation,
as we will check this program against different bounds by providing concrete valuations for \tksym{Q}.}
\label{fig:sort-annotated} 
\end{figure}

\paragraph{Type checking}

With abstract potentials, we can use the same  \T{List} datatype from \autoref{fig:abs-list} 
to verify both coarse- and fine-grained bounds for insertion sort.
For the coarse-grained case, we can give this function the following type signature:
\[\T{sort} \dblcolon \tarrow{\mathsf{xs}}{\tkaplist{a}{1}{\lambda (\_, \_). 1} }{\tklist{a}} \]
As before, omitted potential annotations are zero by default, 
so the return type $\tklist{a}$ is short for $(\tkaplist{a}{0}{\lambda (\_, \_). 0})^0$
The type checking process is illustrated in \autoref{fig:sort-annotated},
where we set $\tksym{Q} = \lambda (\_, \_). 1$.
The initial context contains bindings for both the helper function \T{insert}
and the function \T{sort} itself, which can be used to make a recursive call.
More precisely, the binding for \T{sort} is added to the context as a result of type-checking 
the implicit fixpoint construct that wraps the lambda abstraction.
Importantly for this example, LRT supports \emph{polymorphic recursion}:
the type \T{c} of list elements in the recursive call can be different from the type \T{a}
of list elements in the body.


The top-level term in the body of \T{sort} is a pattern-match,
so, as before, we have to split the context between the scrutinee and the branches.
Since neither of the branches mentions \T{xs}, for simplicity we omit the sharing 
constraints and leave all of its potential with line 3, thus
inferring the type $\tkaplist{a}{1}{1}$ for the scrutinee.
Matching this type against the return type of the \T{Cons} constructor in \autoref{fig:abs-list},
yields the substitution $\T{t} \mapsto \T{a}^1, q\mapsto 1$,
adding the following two new bindings to the context of the \T{Cons} branch:
$\T{hd} : \T{a}^{1}$ and $\T{tl} : \tkaplist{a}{2}{\lambda(\_,\_).1}$.
Importantly, the tail list \T{tl} ends up with more linear potential than the original list \T{xs},
which is precisely the purpose of the inductive potential annotations in \autoref{fig:abs-list},
and is necessary to afford \emph{both} the recursive call and the call to \T{insert}.

Proceeding with type-checking the \T{Cons} branch,
note that there are three terms that consume resources:
the application of \T{insert hs}, the \T{tick} expression, and the recursive call.
We can use the free unit of potential attached to \T{hd} to pay for \T{tick}.
As for \T{tl}, recall that it has twice the potential that the recursive call to \T{sort} consumes,
and we would like to ``save up'' this extra potential to pay for the application of \T{insert hs} to the result of the recursive call.
This is where polymorphic recursion comes in:
the type checker is free to instantiate \T{c} in the type of the recursive call with $\T{a}^s$,
essentially giving every list element some amount of extra potential $s$ which is simply ``piped through'' the call; 
LRT leaves the exact value of $s$ for the solver to find.


All together, type checking leaves us the following
system of arithmetic constraints:
\begin{align*}
\exists p_1, p_2, q_1, q_2, s \in \mathbb{N}. & \, p_1 + p_2 = 1 \land p_2 - 1 \geq 0  \\
& \land q_1 + q_2 = 2 \land q_2 \geq s + 1  \land s \geq 1
\end{align*}
which is satisfiable with $p_2, q_2, s = 1$ and the rest of unknowns set to $0$.
Note that while the annotations in \autoref{fig:sort-annotated} involve applications of abstract potentials,
all potential functions involved in the coarse-grained version of the example are constants, 
so we can treat these as simple first-order numerical constraints. 



\paragraph{Value-dependent resource bounds}
Instantiating the abstract potentials with non-constant functions 
allows us to use the exact same \T{List} datatype to verify a fine-grained 
bound for insertion sort.
To this end, we give it the type signature:
\[\T{sort} \dblcolon \tarrow{xs}{\tkaplist{a}{1}{\lambda (x_1, x_2). \, \mathsf{ite}(x_1 > x_2, 1, 0)}}{\tklist{a}} \]
Type checking still proceeds as illustrated in \autoref{fig:sort-annotated},
except we set $\tksym{Q} = \lambda (x_1, x_2). \, \mathsf{ite}(x_1 > x_2, 1, 0)$.
%
One key difference is that matching the type of the scrutinee \T{xs} against
the return type of \T{Cons} 
requires applying the abstract potential function to yield
$\T{tl} : \tkaplist{a}{1 + \mathsf{ite}(\mathsf{x} > \nu, 1, 0)}{\lambda (x_1, x_2). \, \mathsf{ite}(x_1 > x_2, 1, 0)}$,
in the context. 
The generated arithmetic constraints are similar to the coarse-grained case,
but now symbolic potentials can be functions, 
so the constraints are second-order and must quantify over the program variables $\T{hd}, \nu$ 
and parameters $x_1, x_2$ of abstract potentials:
\begin{align*}
  \exists & p_1, p_2, q_1, q_2, s \in \mathbb{N}\times\mathbb{N}\to\mathbb{N} . \, \forall \T{hd}, \nu, x_1, x_2 \in \mathbb{N} .\\ 
  &p_1(\T{hd}, \nu) + p_2(\T{hd}, \nu) = 1 
  && \text{Sharing \T{hd} (line 5)} \\
  &\land p_2(\T{hd}, \nu) - 1 \geq 0 
  && \text{Well-formedness of \T{hd} (line 8)} \\
  &\land q_1(\T{hd}, \nu) + q_2(\T{hd}, \nu) = 1 + \mathsf{ite}(\T{hd} > \nu, 1, 0)
  && \text{Sharing \T{tl} (line 5)} \\
  &\land q_2(\T{hd}, \nu) \geq s(\T{hd}, \nu) + 1 
  && \text{Subtyping from the call to \T{sort}} \\
  &\land s(\T{hd}, \nu) \geq \mathsf{ite}(\T{hd} > \nu, 1, 0)
  && \text{Subtyping from the call to \T{insert}}
\end{align*}
The solver can validate these constraints by setting 
$p_2, \lambda(x_1,x_2).1$, $q_2, s = \lambda(x_1,x_2). \mathsf{ite}(x_1 > x_2, 1, 0)$, and the rest of the unknowns to $\lambda(x_1,x_2).0$.
Importantly, even though inductive and abstract potentials significantly increase the expressiveness of the type system,
the generated constraints still belong to the same logic fragment (second-order CLIA),
as constraints generated by \resyn, and hence are efficiently decidable.
This is a consequence of the core design principle that differentiates LRT 
from other fine-grained resource analysis techniques~\cite{RadicekBG0Z18,OOPSLA:WWC17,HandleyVH20}:
to encode complex resource consumption,
rather than increasing the complexity of the resource annotations,
we embed \emph{simple annotations} into \emph{complex types}.

Although in this section we focused solely on the resource consumption of insertion sort,
LRT is also able to specify and verify its functional properties---%
that the output list is sorted and contains the same number and/or set of elements as the input list.
To this end, LRT relies on existing liquid type checking techniques~\cite{VazouRoJh13,PolikarpovaKS16}.
Additionally, while this section only shows the use of inductive and abstract potentials
for expressing quadratic potentials on lists,
\autoref{sec:eval} further demonstrates the flexibility of this specification style.
In particular, we show how to use abstract potentials to analyze exponential-time algorithms,
as well as reason about the resource consumption of tree-manipulating
programs in terms of their height and size.

\newcommand{\diw}[1]{{\color{ACMOrange}DW: #1}}

\section{Technical Details}
\label{sec:technical}

In this section, we formulate a substantial subset of our type system as a core calculus and prove type soundness.
This subset features natural numbers and Booleans that are refined by their values, as well as user-defined inductive datatypes that can be refined by user-defined measures.
The gap from the core calculus to our full type system involves abstract refinements and polymorphic datatypes.
The restriction to this subset in the technical development is only for brevity and proofs carry over to all the features of our tool.

\subsection{Setting the Stage: A Resource-Aware Core Language}
\label{sec:core-language}

\paragraph{Syntax}
\autoref{fig:syntax} presents the grammar of terms in the core calculus via abstract binding trees~\cite{book:PFPL16}.
We extend the core language of $\mathrm{Re}^2$~\cite{resyn} with natural numbers, null tuples, ordered pairs, and replace lists with general inductive data structures. 
Expressions are in \emph{a-normal-form}~\cite{LFP:SF92}, which means that syntactic forms for non-tail positions allow only \emph{atoms} $\hat{a} \in \mathsf{Atom}$, which are irreducible terms, \eg, variables and values, without loss of expressivity.
The restriction simplifies typing rules in our system, as we will explain in \autoref{sec:typing-rules}.
We further identify a subset $\mathsf{SimpAtom}$ of $\mathsf{Atom}$ that contains \emph{interpretable} atoms in the refinement logic.
Intuitively, the type of an interpretable atom $a \in \mathsf{SimpAtom}$ admits a well-defined \emph{interpretation} that maps the value of $a$ to its logical refinements, \eg, lists can be refined by their lengths.
A \emph{value} $\jval{v}$ is an atom without reference to any program variable.
An inductive data structure $C(v_0,\tuple{v_1,\cdots,v_m})$ is represented by the constructor name $C$, the stored data $v_0$ in this constructor, and a sequence of child nodes $\tuple{v_1,\cdots,v_m}$.
Note that the core language has two kinds of match expressions: $\mathsf{matp}$ for pairs and $\mathsf{matd}$ for inductive data structures.

The syntactic form $\eimp$ is used as a placeholder for unreachable code, \eg, the then-branch of a conditional expression whose predicate is always false.
The syntactic form $\econsume{c}{e_0}$ is introduced to define the cost model, and it is intended to cost $c \in \bbZ$ units of resource and then reduce to $e_0$.
A negative $c$ means that ${-c}$ units of resource will become available.
The $\mathsf{tick}$ expressions support flexible user-defined resource metrics.
For example, the programmers can wrap every recursive call in $\econsume{1}{\cdot}$ to count those function calls;
alternatively, they may wrap every data constructor in $\econsume{c}{\cdot}$ to keep track of memory consumption, where $c$ is the amount of memory allocated by the constructor.

\begin{figure}
  \[
  \begin{array}{rcl}
    a \in \mathsf{SimpAtom} & \Coloneqq & x \mid \enat{n} \mid \etrue \mid \efalse \mid \etriv \mid \epair{a_1}{a_2} \mid C(a_0,\tuple{a_1,\cdots,a_m}) \\
    \hat{a} \in \mathsf{Atom} & \Coloneqq & a \mid \eabs{x}{e_0} \mid \efix{f}{x}{e_0} \\
    e \in \mathsf{Exp} & \Coloneqq & a \mid \econd{a_0}{e_1}{e_2} \mid \ematp{a_0}{x_1}{x_2}{e_1} \mid \mathsf{matd}(a_0,\many{C_j(x_0,\tuple{x_1,\cdots,x_{m_j}}).e_j}) \\
    & \mid & \eapp{\hat{a}_1}{\hat{a}_2} \mid \elet{e_1}{x}{e_2} \mid \eimp \mid \econsume{c}{e_0} \\
    v \in \mathsf{Val} & \Coloneqq & \enat{n} \mid \etrue \mid \efalse \mid \etriv \mid \epair{v_1}{v_2} \mid C(v_0,\tuple{v_1,\cdots,v_m}) \mid \eabs{x}{e_0} \mid \efix{f}{x}{e_0}
  \end{array}
  \]
  \caption{Syntax of the core calculus}
  \label{fig:syntax}
\end{figure}

\paragraph{Semantics}
The resource consumption of a program is determined by a small-step operational cost semantics.
The semantics is a standard structural semantics augmented with a \emph{resource parameter}, which indicates the amount of available resources.
The \emph{single-step} reduction judgments have the form $\jstep{e}{e'}{q}{q'}$, where $e$ and $e'$ are expressions, and $q,q' \in \bbZ^+_0$ are nonnegative integers.
The intuitive meaning of such a judgment is that with $q$ units of available resources, $e$ reduces to $e'$ without running out of resources, and $q'$ resources are left.
\autoref{fig:semantics} shows some of the reduction rules of the small-step cost semantics.
Note that all the judgments $\jstep{e}{e'}{q}{q'}$ implicitly constrain that $q,q' \ge 0$, so in the rule \textsc{(E-Tick)} for resource consumption, we do not need to distinguish whether the cost $c$ is nonnegative or not.

\begin{figure}
  \begin{flushleft}
  \small\fbox{$\jstep{e}{e'}{q}{q'}$}
  \end{flushleft}
	\begin{mathpar}\footnotesize
		\Rule{E-Cond-True}{ }{ \jstep{\econd{\etrue}{e_1}{e_2}}{e_1}{q}{q} }
		\and
		\Rule{E-Cond-False}{ }{ \jstep{\econd{\efalse}{e_1}{e_2}}{e_2}{q}{q} }
		\and
		\Rule{E-Let-Val}{ \jval{v_1} }{ \jstep{\elet{v_1}{x}{e_2}}{\subst{v_1}{x}{e_2}}{q}{q} }
		\and
		\Rule{E-Tick}{ }{ \jstep{\econsume{c}{e_0}}{e_0}{q}{q-c} }
		\and
		\Rule{E-MatP-Val}{ \jval{v_1} \\ \jval{v_2} }{ \jstep{\ematp{\epair{v_1}{v_2}}{x_1}{x_2}{e_1}}{ \subst{v_1,v_2}{x_1,x_2}{e_1} }{q}{q} }
		\and
		\Rule{E-MatD-Val}{ \jval{v_0} \\ \jval{v_1} \\ \cdots \\ \jval{v_{m_j}} }{ \jstep{\mathsf{matd}(C_j(v_0,\tuple{v_1,\cdots,v_{m_j}}), \many{C_j(x_0,\tuple{x_1,\cdots,x_{m_j}}).e_j}) }{ \subst{v_0,v_1,\cdots,v_{m_j}}{x_0,x_1,\cdots,x_{m_j}}{e_j}  }{q}{q} }
		\and
		\Rule{E-App-Abs}{ \jval{v_2} }{ \jstep{\eapp{\eabs{x}{e_0}}{v_2}}{\subst{v_2}{x}{e_0}}{q}{q} }
		\and
		\Rule{E-App-Fix}{ \jval{v_2} }{ \jstep{\eapp{\efix{f}{x}{e_0}}{v_2}}{\subst{\efix{f}{x}{e_0},v_2}{f,x}{e_0}}{q}{q} }
	\end{mathpar}
	\caption{Selected rules of the small-step operational cost semantics}
	\label{fig:semantics}
\end{figure}

The \emph{multi-step} reduction relation $\mapsto^{*}$ is defined as the reflexive transitive closure of $\mapsto$.
Multi-step reduction can be used to reason about \emph{high-water mark} resource usage of a reduction from $e$ to $e'$, by finding the minimal $q$ such that $\jsteps{e}{e'}{q}{q'}$ for some $q'$.
For monotone resources such as time, the high-water mark cost coincides with the \emph{net cost}, \ie, the sum of costs specified by $\mathsf{tick}$ expressions in the reduction.  
In general, net costs are invariant, \ie, $p-p'=q-q'$ if $\jstepn[m]{e}{e'}{p}{p'}$ and $\jstepn[m]{e}{e'}{q}{q'}$, where $\mapsto^m$ is the $m$-element composition of $\mapsto$.

\subsection{Types and Refinements}
\label{sec:type-definition}

\paragraph{Refinements}
We follow the approach of liquid types~\cite{RondonBKJ12,PolikarpovaKS16,resyn} and develop a refinement language that is distinct from the term language.
\autoref{fig:type-system} formulates the syntax of the core type system.
The refinement language is essentially a simply-typed lambda calculus augmented with logical connectives and linear arithmetic.
As terms are classified by types, refinements $\psi,\phi$ are classified by \emph{sorts} $\Delta$.
The core type system's sorts include Booleans $\bbB$, natural numbers $\bbN$, nullary $\bbU$ and binary products $\sprod{\Delta_1}{\Delta_2}$, arrows $\sarrow{\Delta_1}{\Delta_2}$, and \emph{uninterpreted symbols} $\delta_\alpha$ parametrized by type variables $\alpha$.
In our system, logical constraints $\psi$ have sort $\bbB$, potential annotations $\phi$ have sort $\bbN$, and refinement-level functions have arrow sorts.
Refinements can reference program variables.
Our system interprets a program variable of Boolean, natural-number, or product type as its value, type variable $\alpha$ as an uninterpreted symbol of sort $\delta_\alpha$, and inductive datatype as its \emph{measurement}, which is computed by a total function $\calI_D : (\text{values of datatype $D$}) \to (\text{refinements of sort $\Delta_D$})$.
The function $\calI_D$ is derived by user-defined \emph{measures} for datatypes, which we omit from the formal presentation;
Although measures play an important role in specifying functional properties (\eg, in~\cite{PolikarpovaKS16}), they are orthogonal to resource analysis.
We include the full development with measures in 
\iflong
\autoref{sec:appendixre2}
\else
the technical report~\cite{techreport}
\fi

Formally, we define the following \emph{interpretation} $\calI(\cdot)$ to reflect interpretable atoms $a \in \mathsf{SimpAtom}$ as their logical refinements:
\begin{align*}
  \calI(x) & = x \\
  \calI(\enat{n}) & = n & \calI(\etriv) & = \star \\
  \calI(\etrue) & = \top & \calI(\efalse) & = \bot \\
  \calI(\epair{a_1}{a_2}) & = \rpair{\calI(a_1)}{\calI(a_2)} & \calI(C(a_0,\tuple{a_1,\cdots,a_m})) & = \calI_D(C(a_0,\tuple{a_1,\cdots,a_m}))  
\end{align*}

\begin{figure}
  \[
  \begin{array}{rclrcl}
    \multicolumn{3}{l}{\fbox{\text{Refinement}}} \\
    \psi,\phi & \Coloneqq & \multicolumn{4}{l}{\nu \mid x \mid n \mid \star \mid \top \mid \neg\psi \mid \psi_1\wedge \psi_2 \mid \phi_1 \le \phi_2 \mid \phi_1 + \phi_2 \mid \psi_1 = \psi_2 \mid \rforall{a}{\Delta}{\psi}} \\
    & \mid & \multicolumn{4}{l}{a \mid \rabs{a}{\Delta}{\psi} \mid \psi_1~\psi_2 \mid \rpair{\psi_1}{\psi_2} \mid \projl{\psi} \mid \projr{\psi} } \\
    \multicolumn{3}{l}{\fbox{\text{Sort}}}  \\
    \Delta & \Coloneqq & \bbB \mid  \bbN \mid  \bbU \mid \delta_\alpha \mid \sprod{\Delta_1}{\Delta_2} \mid \sarrow{\Delta_1}{\Delta_2} \\
    \multicolumn{3}{l}{\fbox{\text{Base Type}}} & \multicolumn{3}{l}{\fbox{\text{Resource-Annotated Type}}}  \\
    B & \Coloneqq & \tnat \mid \tbool \mid \tunit \mid \tprod{B_1}{B_2} \mid \tinduct{C}{T}{m}{\theta} \mid m \cdot \alpha & \quad T & \Coloneqq & \tpot{R}{\phi}   \\
    \multicolumn{3}{l}{\fbox{\text{Refinement Type}}} & \multicolumn{3}{l}{\fbox{\text{Type Schema}}} \\
    R & \Coloneqq & \tsubset{B}{\psi} \mid \tarrowm{x}{T_x}{T}{m} & \quad S & \Coloneqq & T \mid \forall\alpha. S 
  \end{array}
  \]
  \caption{Syntax of the core type system}
  \label{fig:type-system}
\end{figure}

\begin{example}[Interpretations of datatypes]\label{exa:inductive-measures}
  Consider a natural-number list type $\mathsf{NatList}$ with constructors $\mathsf{Nil}$ and $\mathsf{Cons}$.
  In the core language, an empty list is encoded as $\mathsf{Nil}(\etriv,\tuple{})$ and a singleton list containing a zero is represented as $\mathsf{Cons}(\enat{0},\tuple{\mathsf{Nil}(\etriv,\tuple{})})$.
  Below defines an interpretation $\calI_\mathsf{NatList}:(\text{values of $\mathsf{NatList}$}) \to (\text{refinements of sort $\bbN$})$ that computes the length of a list:
  \begin{align*}
    \calI_\mathsf{NatList}(\mathsf{Nil}(\etriv,\tuple{})) & \defeq 0, & \calI_\mathsf{NatList}(\mathsf{Cons}(v_h,\tuple{v_t})) &   \defeq \calI_\mathsf{NatList}(v_t) + 1.
  \end{align*}
  %
  %
  In the rest of this section, we will assume that the type $\mathsf{NatList}$ admits a length interpretation.
\end{example}

We will use the abbreviations $\bot,{\vee},{\implies},{\ge},{<},{>},\mathbf{ite}$ with obvious semantics;
\eg, $\psi_1 \vee \psi_2 \defeq \neg(\neg\psi_1 \wedge \neg\psi_2)$ and $\mathbf{ite}(\psi_0,\psi_1,\psi_2) \defeq (\psi_0 \implies \psi_1) \wedge (\neg\psi_0 \implies \psi_2)$.
We will also abbreviate the $m$-element sum $\psi+\psi+\cdots+\psi$ as $m \times \psi$.
We will use finite-product sorts $\Delta_1 \times \Delta_2 \times \cdots \times \Delta_m$, or $\prod_{i=1}^m \Delta_i$ for short, with an obvious encoding with nullary and binary products.
We will also write $\psi.\mathbf{i}$ as the $i$-th projection from a refinement of a finite-product sort.

\paragraph{Types}
We adapt the methodology of $\mathrm{Re}^2$~\cite{resyn} and classify types into four categories.
Base types $B$ are natural numbers, Booleans, nullary and binary products, inductive datatypes, and type variables.
%
%
An inductive datatype $\tinduct{C}{T}{m}{\theta}$ consists of a sequence of constructors, each of which has a name $C$, a content type $T$ (which must be a scalar type), and a finite number $m \in \bbZ_0^+$ of child nodes.
In terms of recursive types,  $(\many{C{\,:\,}(T,m)})$ compactly represents $\mathsf{rec}( X. \many{C{\,:\,} \tprod{T}{X^m}})$, where $X^m$ is the $m$-element product type $X \times X \times \cdots \times X$,
\eg, the type $\mathsf{NatList}$ in \autoref{exa:inductive-measures} can be seen as an abbreviation of $\mathsf{ind}(\mathsf{Nil}{\,:\,} (\tunit,0), \mathsf{Cons}{\,:\,}(\tnat,1))$.
We will explain the resource-related parameters $\theta$,$\lhd$, and $\pi$ later in \autoref{sec:potentials}.
Type variables $\alpha$ are annotated with a \emph{multiplicity}~$m \in \bbZ^+_0 \cup \{\infty\}$, which specifies an upper bound on the number of references for a program variable of such a type.
For example, $\mathsf{ind}(\mathsf{Nil}{\,:\,}(\tunit,0),\mathsf{Cons}{\,:\,}(2\cdot\alpha,1))$ denotes a universal list, each of whose elements can be used at most twice.

Refinement types $R$ are \emph{subset types} and \emph{dependent arrow types}.
Inhabitants of a subset type $\tsubset{B}{\psi}$ are values of type $B$ that satisfy the refinement $\psi$.
The refinement $\psi$ is a logical formula over program variables and a special \emph{value variable} $\nu$, which is distinct from program variables and represents the inhabitant itself.
For example, $\tsubset{\tbool}{\neg \nu}$ is a type of $\efalse$, $\tsubset{\tnat}{\nu > 0}$ is a type of positive integers, and $\tsubset{\mathsf{NatList}}{\nu = 1}$ stands for singleton lists of natural numbers.
A dependent arrow type $\tarrow{x}{T_x}{T}$ is a function type whose return type may reference its formal argument $x$.
Similar to type variables, these arrow types are also annotated with a multiplicity $m \in \bbZ_0^+ \cup \{\infty\}$ bounding from above the number of times a function of such a type can be applied.

Resource-annotated types $\tpot{R}{\phi}$ are refinement types $R$ augmented with potential annotations $\phi$.
The resource annotations are used to carry out the potential method of amortized analysis~\cite{kn:Tarjan85}; intuitively, $\tpot{R}{\phi}$ assigns $\phi$ units of potential to values of the refinement type $R$.
The potential annotation $\phi$ can also reference the value variable $\nu$.
For example, $\tpot{\mathsf{NatList}}{2 \times \nu}$ describes natural-number lists $\ell$ with $2 \cdot\calI_\mathsf{NatList}(\ell) = 2\cdot |\ell|$ units of potential where $|\ell|$ is the length of $\ell$.
As we will show in \autoref{sec:potentials}, the same potential can also be expressed by assigning 2 units of potential to each element in the list.

Type schemas represent possibly polymorphic types, where the type quantifier $\forall$ is only allowed to appear outermost in a type.
Similar to $\mathrm{Re}^2$~\cite{resyn}, we only permit polymorphic types to be instantiated with \emph{scalar} types, which are resource-annotated base types (possibly with subset constraints).
Intuitively, the restriction derives from the fact that our refinement-level logic is first-order, which renders our type system decidable.

We will abbreviate $1 \cdot \alpha$ as $\alpha$, $\tsubset{B}{\top}$ as $B$, $\tarrowm{x}{T_x}{T}{\infty}$ as $\tarrow{x}{T_x}{T}$, and $\tpot{R}{0}$ as $R$.

\subsection{Potentials of Inductive Data Structures}
\label{sec:potentials}


Resource-annotated types $\tpot{R}{\phi}$ provide a mechanism to specify potential functions of inductive data structures in terms of their interpretations.
However, this mechanism is not so expressive because it can only describe potential functions that are \emph{linear} with respect to the interpretations of data structures, since our refinement logic only has linear arithmetic.
One way to support non-linear potentials is to extend the refinement logic with non-linear arithmetic, which would come at the expense of decidability of the type system.
In contrast, our type system adapts the idea of \emph{univariate polynomial potentials}~\cite{RAML10} to a refinement-type setting.
This combination allows us to not only reason about polynomial resource bounds with linear arithmetic in the refinement logic, but also derive fine-grained resource bounds that go beyond the scope of prior work on typed-based amortized resource analysis~\cite{RAML10,RAML11,resyn}.

\paragraph{Simple numeric annotations}
We start by adding numeric annotations to datatypes, following the approach of univariate polynomial potentials~\cite{RAML10}.
Recall the type $\mathsf{NatList}$ introduced in \autoref{exa:inductive-measures}.
We now annotate it with a vector $\vec{q} = (q_1,\cdots,q_k) \in (\bbZ^+_0)^k$ and denote the annotated type by $\mathsf{NatList}^{\vec{q}}$.
The annotation is intended to assign $q_1$ units of potential to every element of the list, $q_2$ units of potential to every element of every suffix of the list (\ie, to every ordered pair of elements), $q_3$ units of potential to the elements of the suffixes of the suffixes (\ie, to every ordered triple of elements), etc.
Let $\ell$ be a list of type $\mathsf{NatList}$ and $\pot{\ell : \mathsf{NatList}^{\vec{q}}}$ be its potential with respect to the annotated type.
Then the potential function $\pot{\cdot}$ can be expressed as a linear combination of binomial coefficients, where $|\ell|$ is the length of $\ell$:
\begin{equation}\label{eq:natlist-potential}
  \pot{\ell : \mathsf{NatList}^{\vec{q}}} = \sum_{i=1}^k \sum_{1 \le j_1 < \cdots < j_i \le |\ell|} q_i = \sum_{i=1}^k q_i \cdot \binom{|\ell|}{i}.
\end{equation}
For example, $\mathsf{NatList}^{(2)}$ assigns $2$ units of potential to each list element, so it describes lists $\ell$ with $2 \cdot |\ell|$ units of potential.

As shown by the proposition below,
one benefit of the binomial representation in \eqref{eq:natlist-potential} is that the potential function $\pot{\cdot}$ can be defined \emph{inductively} on the data structure, and be expressed using only linear arithmetic.

\begin{proposition}\label{prop:natlist-potential}
  Define the potential function $\pot{\cdot}$ for type $\mathsf{NatList}^{\vec{q}}$ as follows:
  \begin{align*}
    \pot{\mathsf{Nil}(\etriv,\tuple{}) : \mathsf{NatList}^{\vec{q}}} & \defeq 0, & \pot{\mathsf{Cons}(v_h,\tuple{v_t}) : \mathsf{NatList}^{\vec{q}}} & \defeq q_1 + \pot{v_t : \mathsf{NatList}^{\lhd(\vec{q})}},
  \end{align*}
  where a potential \emph{shift} operator $\lhd$ is defined as $\lhd(\vec{q}) \defeq (q_1+q_2,q_2+q_3,\cdots,q_{k-1}+q_k,q_k)$.
  Then \eqref{eq:natlist-potential} gives a closed-form solution to the inductive definition above.
\end{proposition}

Based on the observation presented above, prior work~\cite{RAML10,RAML11} builds an automatic resource analysis that infers polynomial resource bounds via efficient \emph{linear programming} (LP).
In this work, our main goal is not to develop an automatic inference algorithm, but rather to extend the expressivity of the potential annotations.

\paragraph{Dependent annotations}
Our first step is to generalize numeric potential annotations to dependent ones.
The idea is to express the potential annotations in the refinement language of our type system.
For example, we can annotate the type $\mathsf{NatList}$ with a vector $\theta=(\theta_1,\cdots,\theta_k)$, where $\theta_i$ is a refinement-level abstraction of sort $\sarrow{\bbN^i}{\bbN}$, for every $i=1,\cdots,k$.
Intuitively, $\theta_i$ denotes the amount of potential assigned to ordered $i$-tuple of elements in a list, depending on the actual values of the elements, \ie, let $\ell=[v_1,\cdots,v_{|\ell|}]$ be a list of natural numbers, then the potential function $\pot{\cdot}$ with respect to the dependently annotated type $\mathsf{NatList}^\theta$ can be expressed as
\begin{equation}\label{eq:natlist-potential-dependent}
  \pot{\ell : \mathsf{NatList}^\theta} = \sum_{i=1}^k \sum_{1 \le j_1 < \cdots < j_i \le |\ell|} \theta_i(v_{j_1},\cdots,v_{j_i}).
\end{equation}

\begin{example}[Dependent potential annotations]\label{exa:dependent-annotation}
  Suppose we want to assign the number of ordered pairs $(a,b)$ satisfying $a>b$ in a list $\ell$ of type $\mathsf{NatList}^\theta$ as the potential of $\ell$.
  Then the desired potential function is $\pot{\ell : \mathsf{NatList}^\theta} = \sum_{1 \le j_1 < j_2 \le |\ell|}  \mathbf{ite}(v_{j_1} >v_{j_2},1,0)$.
  Compared with \eqref{eq:natlist-potential-dependent}, a feasible $\theta = (\theta_1,\theta_2)$ can be defined as follows:
  \begin{align*}
    \theta_1 & \defeq \rabs{x}{\bbN}{0}, & \theta_2 & \defeq \lambda(x_1{\,:\,}\bbN,x_2{\,:\,}\bbN). \mathbf{ite}(x_1>x_2,1,0).
  \end{align*}
  Later we will show the dependent annotation given here can be used to derive a fine-grained resource bound for insertion sort at the end of \autoref{sec:typing-rules}.
\end{example}

Although dependent annotations seem to complicate the representation of potential functions, they \emph{do} retain the benefit of numeric annotations.
The key observation is that we can still express the potential \emph{shift} operator $\lhd$ in our refinement language, which only permits linear arithmetic.
Below presents a generalization of \autoref{prop:natlist-potential}.

\begin{proposition}\label{prop:natlist-potential-dependent}
  Define the potential function $\pot{\cdot}$ for type $\mathsf{NatList}^\theta$ as follows:
  \begin{align*}
    \pot{\mathsf{Nil}(\etriv,\tuple{}) : \mathsf{NatList}^\theta } & \defeq  0, & \pot{\mathsf{Cons}(v_h,\tuple{v_t}) : \mathsf{NatList}^\theta } & \defeq \theta_1(v_h) + \pot{v_t : \mathsf{NatList}^{\lhd(v_h)(\theta)}},
  \end{align*}
  where a dependent potential \emph{shift} operator $\lhd$ is defined in the refinement-level language as 
  \[
  \lhd \defeq \rabs{y}{\bbN}{ \lambda(\theta_1{\,:\,}\sarrow{\bbN}{\bbN}, \cdots, \theta_k{\,:\,} \sarrow{\bbN^k}{\bbN}). (\theta_1',\cdots,\theta_k')  },
  \]
  where $\theta_1' \defeq \rabs{x}{\bbN}{(\theta_1(x) + \theta_2(y,x))}$, $\theta_2' \defeq \rabs{x}{\bbN^2}{(\theta_2(x)+\theta_3(y,x))}$, \ldots, $\theta_{k-1}' \defeq \rabs{x}{\bbN^{k-1}}{(\theta_{k-1}(x) + \theta_k(y,x))}$, and $\theta_k' \defeq \theta_k$.
  Then \eqref{eq:natlist-potential-dependent} gives a closed-form solution to the inductive definition above.
\end{proposition}

\paragraph{Generic annotations}
In general, the potential annotation $\theta$ does not need to have the form of vectors of refinement-level functions; it can be an arbitrary well-sorted refinement, as long as we know how to \emph{extract} potentials from it (\eg, a projection from $\theta=(\theta_1,\cdots,\theta_k)$ to $\theta_1$), and how to \emph{shift} potential annotations to get annotations for child nodes (\eg, \autoref{prop:natlist-potential-dependent}).
This form of generic annotations formulates the notion of \emph{abstract potentials} (introduced in \autoref{sec:overview:inductive}), which is one major contribution of this paper.

In our type system, we parametrize inductive datatypes with not only a potential annotation $\theta$, but also a shift operator $\lhd$ and an extraction operator $\pi$.
For natural-number lists of type $\mathsf{NatList}^\theta$, the potential function $\pot{\cdot}$ is defined inductively in terms of $\lhd$ and $\pi$ as follows:
\begin{align*}
  \pot{\mathsf{Nil}(\etriv,\tuple{}):\mathsf{NatList}^\theta} & \defeq 0, \\
  \pot{\mathsf{Cons}(v_h,\tuple{v_t}) : \mathsf{NatList}^\theta} & \defeq \pi(v_h)(\theta) + \pot{v_t : \mathsf{NatList}^{\lhd(v_h)(\theta)}}.
\end{align*}
Recall that in our type system, an inductive datatype is represented as $\tinduct{C}{T}{m}{\theta}$, where $C$'s are constructor names, $T$'s are content types of data stored at constructors, and $m$'s are numbers of child nodes of constructors.
Let the potential annotation $\theta$ be sorted $\Delta_\theta$, and values of content type $T_j$ be sorted as $\Delta_{T_j}$ for each constructor $C_j{\,:\,}(T_j,m_j)$.
Then the extraction operator $\pi$ is supposed to be a tuple, the $j$-th component of which is a refinement-level function with sort $\sarrow{\Delta_{T_j}}{\sarrow{\Delta_\theta}{\bbN}}$, \ie, extracts potential for the $j$-th constructor from the annotation $\theta$.
Similarly, the shift operator $\lhd$ is also a tuple whose $j$-th component is a refinement-level function with sort $\sarrow{\Delta_{T_j}}{\sarrow{\Delta_\theta}{\Delta_\theta^{m_j}}}$, \ie, shifts potential annotations for the child nodes of the $j$-th constructor.
With the two operators $\lhd,\pi$ and the potential annotation $\theta$, we can now define the potential function $\pot{\cdot}$ for general inductive datatypes as an inductive function:
\begin{equation}\label{eq:generic-potential-function}
\begin{split}
  \pot{C_j(v_0,\tuple{v_1,\cdots,v_{m_j}}) : \tinduct{C}{T}{m}{\theta}} & \defeq \pot{v_0 : T_j} \\
  &  + \pi.\mathbf{j}(\calI(v_0))(\theta)  \\
  & + \sum_{i=1}^{m_j} \pot{ v_i : \tinduct{C}{T}{m}{ \lhd.\mathbf{j}( \calI(v_0))(\theta).\mathbf{i} } }.
\end{split}
\end{equation}
Note that
(i) the definition above includes the potential of the value $v_0$ stored at the constructor with respect to its type $T_j$, because the elements in the data structure may also carry potentials, and
(ii) we use the interpretation $\calI(\cdot)$ defined in \autoref{sec:type-definition} to interpret values as their logical refinements.

\begin{example}[Generic potential annotations]\label{exa:generic-annotation}
  Recall the dependently annotated list type $\mathsf{NatList}^{(\theta_1,\theta_2)}$ in \autoref{exa:dependent-annotation}.
  We can now formalize it in the core type system.
  Let
  \[
  \mathsf{NatList}^{(\theta_1,\theta_2)} \defeq \mathsf{ind}^{(\theta_1,\theta_2)}_{\lhd,\pi}(\mathsf{Nil}{\,:\,}(\tunit,0), \mathsf{Cons}{\,:\,}(\tnat,1)),
  \]
  where 
  $\lhd=(\lhd_\mathsf{Nil},\lhd_{\mathsf{Cons}})$ and $\pi = (\pi_{\mathsf{Nil}},\pi_{\mathsf{Cons}})$ are defined as follows:
  \begin{align*}
  \pi_{\mathsf{Nil}} & \defeq \rabs{\_}{\bbU}{  \lambda (\theta_1{\,:\,}\sarrow{\bbN}{\bbN}, \theta_2{\,:\,} \sarrow{\sprod{\bbN}{\bbN}}{\bbN}) . 0 }, \\
  \pi_{\mathsf{Cons}} & \defeq \rabs{y}{\bbN}{ \lambda (\theta_1{\,:\,}\sarrow{\bbN}{\bbN}, \theta_2{\,:\,} \sarrow{\sprod{\bbN}{\bbN}}{\bbN}). \theta_1(y) }, \\
  \lhd_{\mathsf{Nil}} & \defeq \rabs{\_}{\bbU}{ \lambda (\theta_1{\,:\,}\sarrow{\bbN}{\bbN}, \theta_2{\,;\,} \sarrow{\sprod{\bbN}{\bbN}}{\bbN} ). \star  }, \\
  \lhd_{\mathsf{Cons}} & \defeq \rabs{y}{\bbN}{ \lambda (\theta_1{\,:\,}\sarrow{\bbN}{\bbN}, \theta_2{\,:\,} \sarrow{\sprod{\bbN}{\bbN}}{\bbN}). (\rabs{x}{\bbN}{\theta_1(x)+\theta_2(y,x)}, \theta_2)  }.
  \end{align*}
  
  Different instantiations of $\theta_1,\theta_2$ lead to different potential functions.
  \autoref{exa:dependent-annotation} presents an instantiation to count the out-of-order pairs in a natural-number list.
  Meanwhile, one can implement the simple numeric annotations $(q_1,q_2)$ by setting $\theta_1 \defeq \rabs{x}{\bbN}{q_1}$ and $\theta_2 \defeq \rabs{x}{\sprod{\bbN}{\bbN}}{q_2}$ as constant functions.
\end{example}


\subsection{Typing Rules}
\label{sec:typing-rules}

In this section, we formulate our type system as a set of derivation rules.
The \emph{typing context} $\Gamma$ is a sequence of bindings for program variables $x$, bindings for refinement variables $a$, type variables $\alpha$, path constraints $\psi$, and free potentials $\phi$:
\[
\Gamma \Coloneqq \cdot \mid \Gamma, \bindvar{x}{S} \mid \Gamma,\bindvar{a}{\Delta} \mid \Gamma,\alpha \mid \Gamma, \psi \mid \Gamma, \phi.
\]
%
Our type system consists of five kinds of judgments: sorting, well-formedness, subtyping, sharing, and typing.
We omit sorting and well-formedness rules and include them in
\iflong
\autoref{sec:appendixre2} 
\else
the technical report~\cite{techreport}
\fi
The sorting judgment $\jsort{\Gamma}{\psi}{\Delta}$ states that a term $\psi$ has a sort $\Delta$ under the context $\Gamma$ in the refinement language.
A type $S$ is said to be well-defined under a context $\Gamma$, denoted by $\jwftype{\Gamma}{S}$, if every referenced variable in $S$ is in the proper scope.

\paragraph{Typing with refinements}
\autoref{fig:typing} presents the typing rules of the core type system.
The typing judgment $\jstyping{\Gamma}{e}{S}$ states that the expression $e$ has type $S$ under context $\Gamma$.
Its intuitive meaning is that if all path constraints in $\Gamma$ are satisfied, and there is \emph{at least} the amount resources as indicated by the potential in $\Gamma$ then this suffices to evaluate $e$ to a value $v$ that satisfies logical constraints indicated by $S$, and after the evaluation there are \emph{at least} as many resources available as indicated by the potential in $S$.
The rules can be organized into syntax-directed and structural rules.
Structural rules \textsc{(S-*)} can be applied to every expression; in the implementation, we apply these rules strategically to avoid redundant proof search.

The auxiliary \emph{atomic-typing} judgment $\jatyping{\Gamma}{a}{B}$ assigns base types to interpretable atoms $a \in \mathsf{SimpAtom}$.
Atomic typing is useful in the rule \textsc{(T-SimpAtom)}, which uses the interpretation $\calI(\cdot)$ to derive a most precise refinement type for interpretable atoms, \eg, $\etrue$ is typed $\tsubset{\tbool}{\nu = \top}$, $\enat{5}$ is typed $\tsubset{\tnat}{\nu = 5}$, and a singleton list $\mathsf{Cons}(\enat{5},\tuple{\mathsf{Nil}(\etriv,\tuple{})})$ is typed $\tsubset{ \mathsf{NatList}^{\theta} }{ \nu = 1 }$ with some appropriate $\theta$ (recall that $\mathsf{NatList}$ admits a length interpretation).

The \emph{subtyping} judgment $\jsubty{\Gamma}{T_1}{T_2}$ is defined via a common approach for refinement types, with the extra requirement that the potential in $T_1$ should be not less than that in $T_2$.
\autoref{fig:sharing-and-subtyping} shows the subtyping rules.
A canonical use of subtyping is to ``forget'' locally introduced program variables in the result type of an expression, \eg, to ``forget'' $x$ in the type of $e_2$ when typing $\elet{e_1}{x}{e_2}$.
In rule \textsc{(Sub-Dtype)}, we introduce a partial order $\sqsubseteq_{\Delta_\theta}$ over potential annotations $\theta$ of sort $\Delta_\theta$.
For example, if $\theta_1$ and $\theta_2$ are sorted $\bbN$, then $\theta_1 \sqsubseteq_\bbN \theta_2$ is encoded as $\theta_1 \le \theta_2$ in the refinement language.
We carefully define the partial order, in a way that the partial-order relation can be encoded as a first-order fragment of the refinement language.
Notable is that we introduce \emph{validity-checking} judgments $\jprop{\Gamma}{\psi}$ to reason about logical constraints, \ie, to state that the Boolean-sorted refinement $\psi$ is always true under any instance of the context $\Gamma$.
We formalize the validity-checking relation via a set-based denotational semantics for the refinement language.
Validity checking is then reduced to Presburger arithmetic, making it decidable.
The full development of validity checking is included in
\iflong
\autoref{sec:appendixvalidity} 
\else
the technical report~\cite{techreport}
\fi

The rule \textsc{(T-MatD)} reasons about \emph{invariants} for inductive datatypes.
These invariants come from the associated interpretation of inductive data structures, \eg, the length of a list $\mathsf{Cons}(a_h,\tuple{a_t})$ is one plus the length of its tail $a_t$.
Intuitively, if the data structure $a_0$ can be deconstructed as $C_j(x_0,\tuple{x_1,\cdots,x_{m_j}})$ of a datatype $D$ with the form $\tinduct{C}{T}{m}{\theta}$,
then by the definition of the interpretation $\calI(\cdot)$, we can derive
\begin{align*}
    \calI(a_0) & = \calI(C_j(x_0,\tuple{x_1,\cdots,x_{m_j}})) = \calI_D(C_j(x_0,\tuple{x_1,\cdots,x_{m_j}})), 
\end{align*}
which is exactly the path constraint required by the rule \textsc{(T-MatD)} to type the $j$-th branch $e_j$.
For example, if $a_0$ has type $\mathsf{NatList}^\theta$, then the path constraints for the $\mathsf{Nil}(\_,\tuple{})$ and $\mathsf{Cons}(x_h,\tuple{x_t})$ constructors become $\calI(a_0)=0$ and $\calI(a_0)=x_t+1$, respectively.

The type system has two rules for function applications: \textsc{(T-App)} and \textsc{(T-App-SimpAtom)}. 
In the former case, the function return type $T$ does not mention $x$, and thus can be directly used as the type of the application.
This rule deals with cases \eg for all applications with higher-order arguments, since our sorting rules prevent functions from showing up in the refinements language.
In the latter case, the function return type $T$ mentions $x$, but the argument has a scalar type, and thus must be an interpretable atom $a \in \mathsf{SimpAtom}$, so we can substitute $x$ in $T$ with its interpretation $\calI(a)$.
Note that it is the use of a-normal-form that brings us the ability to derive precise types for dependent function applications.

\begin{figure}
  \begin{flushleft}
  \small\fbox{$\jatyping{\Gamma}{a}{B}$}
  \end{flushleft}
  \begin{mathpar}\footnotesize
    \Rule{SimpAtom-Var}
    { \Omit{\jwfctxt{\Gamma} \\} \Gamma(x) = \trefined{B}{\psi}{\phi} }
    { \jatyping{\Gamma}{x}{B} }
    \and
    \Rule{SimpAtom-Bool}
    { \Omit{\jwfctxt{\Gamma} \\} b \in \set{ \etrue, \efalse } }
    { \jatyping{\Gamma}{b}{\tbool} }
    \and
    \Rule{SimpAtom-Nat}
    { \Omit{\jwfctxt{\Gamma}} }
    { \jatyping{\Gamma}{\enat{n}}{\tnat} }
    \and
    \Rule{SimpAtom-Unit}
    { \Omit{\jwfctxt{\Gamma}} }
    { \jatyping{\Gamma}{\etriv}{\tunit} }
    \and
    \Rule{SimpAtom-Pair}
    { \jctxsharing{\Gamma}{\Gamma_1}{\Gamma_2} \\\\ \jatyping{\Gamma_1}{a_1}{B_1} \\ \jatyping{\Gamma_2}{a_2}{B_2} }
    { \jatyping{\Gamma}{\epair{a_1}{a_2}}{\tprod{B_1}{B_2}} }
    \and
    \Rule{SimpAtom-ConsD}
    { \jctxsharing{\Gamma}{\Gamma_1}{\Gamma_2} \\
    \jstyping{\Gamma_1}{a_0}{T_j} \\\\
    \jatyping{\Gamma_2}{\tuple{a_1,\cdots,a_{m_j}}}{\textstyle\prod_{i=1}^{m_j} \tinduct{C}{T}{m}{ \lhd.\mathbf{j}(\calI(a_0))(\theta).\mathbf{i} } } }
    { \jatyping{\Gamma,\pi.\mathbf{j}(\calI(a_0))(\theta)}{C_j(a_0, \tuple{a_1,\cdots,a_{m_j}} )}{ \tinduct{C}{T}{m}{\theta}  }}
  \end{mathpar}
  \begin{flushleft}
  \small\fbox{$\jstyping{\Gamma}{e}{S}$}
  \end{flushleft}
  \begin{mathpar}\footnotesize
    \Rule{T-SimpAtom}
	{ \jatyping{\Gamma}{a}{B} }
	{ \jstyping{\Gamma}{a}{\tsubset{B}{\nu = \calI(a)}} }
	\and
    \Rule{T-Var}
    { \Gamma(x) = S }
    { \jstyping{\Gamma}{x}{S} }
    \and
    \Rule{T-Imp}
    { \jprop{\Gamma}{\bot} \\ \jwftype{\Gamma}{T} }
    { \jstyping{\Gamma}{\eimp}{T} }
    \and
    \Rule{T-Tick-P}
    {  c \ge 0 \\ \jstyping{\Gamma}{e_0}{T} }
    { \jstyping{\Gamma,c}{\econsume{c}{e_0}}{T} }
    \and
    \Rule{T-Tick-N}
    { c < 0 \\ \jstyping{\Gamma,-c}{e_0}{T} }
    { \jstyping{\Gamma}{\econsume{c}{e_0}}{T} }
    \and
    \Rule{T-Cond}
    { \jatyping{\Gamma}{a_0}{\tbool} \\\\
    \jstyping{\Gamma,\calI(a_0)}{e_1}{T} \\\\
    \jstyping{\Gamma,\neg\calI(a_0)}{e_2}{T} }
    { \jstyping{\Gamma}{\econd{a_0}{e_1}{e_2}}{T} }
    \and
    \Rule{T-MatP}
    { \jctxsharing{\Gamma}{\Gamma_1}{\Gamma_2} \\
    \jatyping{\Gamma_1}{a_0}{\tprod{B_1}{B_2}} \\
    \jwftype{\Gamma}{T} \\\\
    \jstyping{\Gamma_2,x_1:B_1,x_2:B_2,\calI(a_0)=(x_1,x_2)}{e_1}{T} }
    { \jstyping{\Gamma}{\ematp{a_0}{x_1}{x_2}{e_1}}{T} }
    \and
    \Rule{T-MatD}
		{ \jctxsharing{\Gamma}{\Gamma_1}{\Gamma_2} \\
		\jatyping{\Gamma_1}{a_0}{ \tinduct{C}{T}{m}{\theta} } \\
		\jwftype{\Gamma}{T'}  \\\\
		\text{for each $j$,} \enskip \Gamma_{2,j} \defeq {} \Big(\Gamma_2,x_0\!:\!T_j, \\\\
		x_1\!:\!\tinduct{C}{T}{m}{\lhd.\mathbf{j}(x_0)(\theta).\mathbf{1}},\makebox[1em][c]{.\hfil.\hfil.},x_{m_j}\!:\!\tinduct{C}{T}{m}{\lhd.\mathbf{j}(x_0)(\theta).\mathbf{m_j}}, \calI(a_0)=\calI_D(C_j(x_0,\tuple{x_1,\makebox[1em][c]{.\hfil.\hfil.},x_{m_j}})), \pi.\mathbf{j}(x_0)(\theta) \Big) , \\\\
		\jstyping{\Gamma_{2,j} }{e_j}{T'}	 }
		{ \jstyping{\Gamma}{ \mathsf{matd}(a_0,\many{ C_j(x_0,\tuple{x_1,\cdots,x_{m_j}}).e_j }) }{ T' }}
    \and
    \Rule{T-Let}
    { \jctxsharing{\Gamma}{\Gamma_1}{\Gamma_2} \\
    \jwftype{\Gamma}{T_2} \\\\
    \jstyping{\Gamma_1}{e_1}{S_1} \\
    \jstyping{\Gamma_2,\bindvar{x}{S_1}}{e_2}{T_2}  }
    { \jstyping{\Gamma}{\elet{e_1}{x}{e_2}}{T_2} }
    \and
    \Rule{T-App-SimpAtom}
    { \jctxsharing{\Gamma}{\Gamma_1}{\Gamma_2} \\
    \jstyping{\Gamma_1}{\hat{a}_1}{\tarrowm{x}{\trefined{B}{\psi}{\phi}}{T}{1}} \\
    \jstyping{\Gamma_2}{a_2}{\trefined{B}{\psi}{\phi}} }
    { \jstyping{\Gamma}{\eapp{\hat{a}_1}{a_2}}{\subst{\calI(a_2)}{x}{T}} }
    \and
    \Rule{T-App}
    { \jctxsharing{\Gamma}{\Gamma_1}{\Gamma_2} \\
    \jstyping{\Gamma_1}{\hat{a}_1}{ \tarrowm{x}{T_x}{T}{1} } \\
    \jstyping{\Gamma_2}{\hat{a}_2}{T_x} \\
    \jwftype{\Gamma}{T} }
    { \jstyping{\Gamma}{\eapp{\hat{a}_1}{\hat{a}_2}}{T} }
    \and
    \Rule{T-Abs}
    { \jwftype{\Gamma}{T_x} \\
    \jstyping{\Gamma,\bindvar{x}{T_x} }{e_0}{T} \\
    \jctxsharing{\Gamma}{\Gamma}{\Gamma} }
    { \jstyping{\Gamma}{\eabs{x}{e_0}}{ \tarrow{x}{T_x}{T}} }
    \and
    \Rule{T-Abs-Lin}
    { \jwftype{\Gamma}{T_x} \\
    \jstyping{\Gamma,\bindvar{x}{T_x}}{e_0}{T} }
    { \jstyping{m \times \Gamma}{\eabs{x}{e_0}}{\tarrowm{x}{T_x}{T}{m}} }
    \and
    \Rule{T-Fix}
    { S = \forall \many{\alpha}. \tarrow{x}{T_x}{T} \\ \jwftype{\Gamma}{S} \\\\
    \jstyping{\Gamma,\bindvar{f}{S}, \many{\alpha}, \bindvar{x}{T_x}}{e_0}{T} \\
    \jctxsharing{\Gamma}{\Gamma}{\Gamma} }
    { \jstyping{\Gamma}{\efix{f}{x}{e_0}}{S} }
    \and
    \Rule{S-Gen}
    { \jval{v} \\ \jstyping{\Gamma,\alpha}{v}{S} \\\\ \jsharing{\Gamma,\alpha}{S}{S}{S} }
    { \jstyping{\Gamma}{v}{\forall\alpha.S} }
    \and
    \Rule{S-Inst}
    { \jstyping{\Gamma}{e}{\forall\alpha.S} \\ \jwftype{\Gamma}{\tpot{\tsubset{B}{\psi}}{\phi}} }
    { \jstyping{\Gamma}{e}{\subst{\tpot{\tsubset{B}{\psi}}{\phi}}{\alpha}{S}} }
    \and
    \Rule{S-Subtype}
    { \jstyping{\Gamma}{e}{T_1} \\ \jsubty{\Gamma}{T_1}{T_2} }
    { \jstyping{\Gamma}{e}{T_2} }
    \and
    \Rule{S-Transfer}
    { \jstyping{\Gamma'}{e}{S} \\\\
    \jprop{\Gamma}{\pot{\Gamma} = \pot{\Gamma'}} \\ |\Gamma| = |\Gamma'| }
    { \jstyping{\Gamma}{e}{S} }
    \and
    \Rule{S-Relax}
    { \jstyping{\Gamma}{e}{\tpot{R}{\phi}} \\ \jsort{\Gamma}{\phi'}{\bbN} }
    { \jstyping{\Gamma,\phi'}{e}{\tpot{R}{\phi+\phi'}} }
  \end{mathpar}
  \caption{Typing rules}
  \label{fig:typing}
\end{figure}

\paragraph{Resources}

There are two typing rules for the syntactic form $\econsume{c}{e_0}$, one for nonnegative costs and the other for negative costs.
The rule \textsc{(T-Tick-N)} assumes $c<0$ and adds $-c$ units of free potential to the context for typing $e_0$.
The rule \textsc{(T-Tick-P)} behaves differently; it states that $\econsume{c}{e_0}$ is only typable in a context containing a free-potential term $c$.
Nevertheless, we can use the rule \textsc{(S-Transfer)} to rearrange free potentials within the context into this form, as long as the total amount of free potential stays unchanged.
In the rule \textsc{(S-Transfer)}, $\pot{\Gamma}$ extracts all the free potentials in the context $\Gamma$, while $|\Gamma|$ removes all the free potentials, \ie, $|\Gamma|$ keeps the functional specifications of $\Gamma$.

To carry out amortized resource analysis~\cite{kn:Tarjan85}, our type system is supposed to properly reason about potentials, that is, potentials cannot be generated from nothing.
This \emph{linear} nature of potentials motivates us to develop an \emph{affine} type system~\cite{kn:Walker02}.
As in $\mathrm{Re}^2$~\cite{resyn}, we have to introduce explicit \emph{sharing} to use a program variable multiple times.
The sharing judgment takes the form $\jsharing{\Gamma}{S}{S_1}{S_2}$ and is intended to state that under the context $\Gamma$, the potential associated with type $S$ is apportioned into two parts to be associated with type $S_1$ and type $S_2$.
\autoref{fig:sharing-and-subtyping} also presents the sharing rules.
In rule \textsc{(Share-Dtype)}, we introduce a notation $\theta = \theta_1 \oplus_{\Delta_\theta} \theta_2$, which means that the annotation $\theta$ is the ``sum'' of two annotations $\theta_1,\theta_2$ that have sort $\Delta_\theta$.
For example, we define $\theta_1 \oplus_\bbN \theta_2$ by $\theta_1 + \theta_2$ in the refinement language.
Similar to the partial order $\sqsubseteq_{\Delta_\theta}$, which is used in the subtyping rules, we encode the ``sum`` operator $\oplus_{\Delta_\theta}$ using a first-order fragment of the refinement language.
The sharing relation is further extended to \emph{context sharing}, written $\jctxsharing{\Gamma}{\Gamma_1}{\Gamma_2}$, which means that $\Gamma_1$ and $\Gamma_2$ have the same sequence of bindings as $\Gamma$, but the free potentials in $\Gamma$ are split into two parts to be associated with $\Gamma_1$ and $\Gamma_2$.
Context sharing is used extensively in the typing rules where the expression has at least two sub-expressions to evaluate, \eg, in the rule \textsc{(T-Let)} for an expression $\elet{e_1}{x}{e_2}$, we apprortion $\Gamma$ into $\Gamma_1$ and $\Gamma_2$, use $\Gamma_1$ for typing $e_1$ and $\Gamma_2$ for typing $e_2$.
Note that the rule \textsc{(T-Abs)} and \textsc{(T-Fix)} has self-sharing $\jctxsharing{\Gamma}{\Gamma}{\Gamma}$ as a premise, which means that the function can only use free variables with zero potential in the context.
This restriction ensures that the program cannot gain potential through free variables by repeatedly applying a function of type $\tarrowm{x}{T_x}{T}{\infty}$ with an infinite multiplicity.

The rule \textsc{(T-Abs-Lin)} is introduced for typing functions with upper bounds on the number of applications.
The rule associates a multiplicity $m \in \bbZ^+_0$ with the function type as the upper bound.
We use a finer-grained premise than context self-sharing to state that the potential of the free variables in the function is enough to pay for $m$ function applications.
This rule is useful for deriving types of curried functions \eg a function of type $\tarrow{x}{T_x}{\tarrow{y}{T_y}{T}}$ that require nonzero units of potential in its first argument $x$.
In that case, a function $f$ can be assigned a type $\tarrow{x}{T_x}{\tarrowm{y}{T_y}{T}{m}}$, which means that the potential stored in the first argument $x$ is enough for the partially applied function $\eapp{f}{x}$ to be invoked for $m$ times.

The elimination rule \textsc{(T-MatD)} realizes the inductively defined potential function in \eqref{eq:generic-potential-function}:
for typing the $j$-th branch $e_j$, one has to add bindings of the content type $x_0:T_j$ and properly shifted types for child nodes $x_i:\tinduct{C}{T}{m}{\lhd.\mathbf{j}(x_0)(\theta).\mathbf{i}}$, as well as a free-potential term $\pi.\mathbf{j}(x_0)(\theta)$ indicated by the potential-extraction operator $\pi.\mathbf{j}$, to the context.
The introduction rule \textsc{(SimpAtom-ConsD)} stores the amount of potentials required for deconstructing data structures.
For typing $C_j(a_0,\tuple{a_1,\cdots,a_{m_j}})$ with type $\tinduct{C}{T}{m}{\theta}$, the rule requires $\pi.\mathbf{j}(\calI(a_0))(\theta)$ as free potential in the context, which is used to pay for potential extraction $\pi.\mathbf{j}$, and a premise stating that each child node $a_i$ has a corresponding properly-shifted annotated datatype $\tinduct{C}{T}{m}{\lhd.\mathbf{j}(\calI(a_0))(\theta).\mathbf{i}}$.

Finally, the structural rule \textsc{(S-Relax)} is usually used when we are analyzing function applications.
Both the rule \textsc{(T-App)} and the rule \textsc{(T-App-SimpAtom)} use up all the potential in the context, but in practice it is necessary to pass some potential through the function call to analyze non-tail-recursive programs.
This is achieved by using the rule \textsc{(S-Relax)} at a function application with $\phi'$ as the potential threaded to the computation that continues after the function returns.

\begin{figure}
  \begin{flushleft}
  \small\fbox{$\jsharing{\Gamma}{S}{S_1}{S_2}$}
  \end{flushleft}
  \begin{mathpar}\footnotesize
    \Rule{Share-Nat}
    {  }
    { \jsharing{\Gamma}{\tnat}{\tnat}{\tnat} }
    \and
    \Rule{Share-Bool}
    { }
    { \jsharing{\Gamma}{\tbool}{\tbool}{\tbool} }
    \and
    \Rule{Share-Unit}
    {  }
    { \jsharing{\Gamma}{\tunit}{\tunit}{\tunit} }
    \and
    \Rule{Share-Poly}
    { \jsharing{\Gamma,\alpha}{S}{S}{S} }
    { \jsharing{\Gamma}{\forall\alpha.S}{\forall\alpha.S}{\forall\alpha.S} }
    \and
    \Rule{Share-Prod}
    { \jsharing{\Gamma}{B_1}{B_{11}}{B_{12}} \\
    \jsharing{\Gamma}{B_2}{B_{21}}{B_{22}} }
    { \jsharing{\Gamma}{\tprod{B_1}{B_2}}{\tprod{B_{11}}{B_{21}}}{\tprod{B_{12}}{B_{22}}} }
    \and
    \Rule{Share-Dtype}
    { 
    \jsharing{\Gamma}{\many{T}}{\many{T_1}}{\many{T_2}} \\
    \jsort{\Gamma}{\theta,\theta_1,\theta_2}{\Delta_\theta} \\   \jprop{\Gamma}{\theta = \theta_1 \oplus_{\Delta_\theta} \theta_2}  }
    { \jsharing{\Gamma}{ \tinduct{C}{T}{m}{\theta}  }{ \tinduct{C}{T_1}{m}{\theta_1} }{ \tinduct{C}{T_2}{m}{\theta_2} } }
    \and
    \Rule{Share-Tvar}
    { \alpha \in \Gamma \\ {m = m_1 + m_2} }
    { \jsharing{\Gamma}{m \cdot \alpha}{m_1 \cdot \alpha}{m_2 \cdot \alpha} }
    \and
    \Rule{Share-Subset}
    { \jsharing{\Gamma}{B}{B_1}{B_2} \\ \jwftype{\Gamma}{\tsubset{B}{\psi}} }
    { \jsharing{\Gamma}{\tsubset{B}{\psi}}{\tsubset{B_1}{\psi}}{\tsubset{B_2}{\psi}} }
    \and
    \Rule{Share-Arrow}
    {  \jwftype{\Gamma}{\p{\tarrow{x}{T_x}{T}}} \\ m = m_1+m_2 }
    { \jsharing{\Gamma}{\p{\tarrowm{x}{T_x}{T}{m}}}{\p{\tarrowm{x}{T_x}{T}{m_1}}}{\p{\tarrowm{x}{T_x}{T}{m_2}}} }
    \and
    \Rule{Share-Pot}
    { \jsharing{\Gamma}{R}{R_1}{R_2} \\ \jprop{\Gamma,\bindvar{\nu}{R}}{\phi=\phi_1+\phi_2} }
    { \jsharing{\Gamma}{\tpot{R}{\phi}}{\tpot{R_1}{\phi_1}}{\tpot{R_2}{\phi_2}} } 
  \end{mathpar}
  \begin{flushleft}
  \small\fbox{$\jsubty{\Gamma}{T_1}{T_2}$}
  \end{flushleft}
  \begin{mathpar}\footnotesize
    \Rule{Sub-Nat}
    {  }
    { \jsubty{\Gamma}{\tnat}{\tnat} }
    \and
    \Rule{Sub-Unit}
    {  }
    { \jsubty{\Gamma}{\tunit}{\tunit} }
    \and
    \Rule{Sub-Bool}
    {  }
    { \jsubty{\Gamma}{\tbool}{\tbool} }
    \and
    \Rule{Sub-Prod}
    { \jsubty{\Gamma}{B_1}{B_1'} \\ \jsubty{\Gamma}{B_2}{B_2'} }
    { \jsubty{\Gamma}{\tprod{B_1}{B_2}}{\tprod{B_1'}{B_2'}} }
    \and
    \Rule{Sub-Dtype}
    { 
    \jsubty{\Gamma}{\many{T}}{\many{T'}} \\
    \jsort{\Gamma}{\theta,\theta'}{\Delta_\theta} \\
    \jprop{\Gamma}{\theta' \sqsubseteq_{\Delta_\theta} \theta} }
    { \jsubty{\Gamma}{ \tinduct{C}{T}{m}{\theta} }{ \tinduct{C}{T'}{m}{\theta'} } }
    \and
    \Rule{Sub-TVar}
    { \alpha \in \Gamma \\ {m_1 \ge m_2} }
    { \jsubty{\Gamma}{m_1 \cdot \alpha}{m_2 \cdot \alpha} }
    \and
    \Rule{Sub-Subset}
    { \jsubty{\Gamma}{B_1}{B_2} \\\\ \jprop{\Gamma,\bindvar{\nu}{B_1} }{\psi_1 \implies \psi_2} }
    { \jsubty{\Gamma}{\tsubset{B_1}{\psi_1}}{\tsubset{B_2}{\psi_2}} }
    \and
    \Rule{Sub-Arrow}
    { \jsubty{\Gamma}{T_x'}{T_x} \\ \jsubty{\Gamma,\bindvar{x}{T_x'}}{T}{T'} \\ m \ge m' }
    { \jsubty{\Gamma}{\tarrowm{x}{T_x}{T}{m}}{\tarrowm{x}{T_x'}{T'}{m'}} }
    \and
    \Rule{Sub-Pot}
    { \jsubty{\Gamma}{R_1}{R_2} \\ \jprop{\Gamma,\bindvar{\nu}{R_1}}{\phi_1 \ge \phi_2} }
    { \jsubty{\Gamma}{\tpot{R_1}{\phi_1}}{\tpot{R_2}{\phi_2}} }
  \end{mathpar}
  \caption{Sharing and subtyping}
  \label{fig:sharing-and-subtyping}
\end{figure}

\begin{example}[Insertion sort]\label{exa:insertion-sort}
  As shown in \autoref{sec:overview:inductive}, our type system is able to verify that an implementation of insertion sort 
  performs exactly the same amount of insertions as the number of out-of-order pairs in the input list.
  %
  %
  We rewrite the function \T{insert} as follows in the core calculus, using the dependently annotated list type $\mathsf{NatList}^{(\theta_1,\theta_2)}$ from \autoref{exa:dependent-annotation}:
\begin{align*}
  \T{insert} & \dblcolon \tarrow{y}{\tnat}{\tarrow{\ell}{\mathsf{NatList}^{(\rabs{x}{\bbN}{\mathbf{ite}(y>x,1,0)}, \rabs{x}{\sprod{\bbN}{\bbN}}{0}  ) }}{ \mathsf{NatList}^{(\rabs{x}{\bbN}{0}, \rabs{x}{\sprod{\bbN}{\bbN}}{0} )} } } \\
  \T{insert} & = \lambda(y. \mathsf{fix}(f.\ell.\mathsf{matd}(\ell, \\  
  & \qquad  \mathsf{Nil}(\_,\tuple{}). \mathsf{Cons}(y,\tuple{\mathsf{Nil}(\etriv,\tuple{})}), \\
  & \qquad \mathsf{Cons}(h,\tuple{t}). \mathsf{let}(y>h, b . \\
  & \qquad\qquad \econd{b}{ \econsume{1}{\elet{{\eapp{f}{t}}}{t'}{ \mathsf{Cons}(h, \tuple{t'}) }} }{ \mathsf{Cons}(y, \tuple{\mathsf{Cons}(h,\tuple{t})}) }  )
\end{align*}
  We assume that a comparison function $>$ with signature $\tarrow{a}{\tnat}{\tarrow{b}{\tnat}{\tsubset{\tbool}{\nu = (a > b)}}}$ is provided in the typing context.
  Next, we illustrate how our type system justifies the number of recursive calls in \T{insert} is bounded by the number of elements in $\ell$ that are less than the element $y$ that is being inserted to $\ell$.
  Suppose $\Gamma$ is a typing context that contains the signature of $>$, as well as type bindings for $y$, $f$, and $\ell$.
  To reason about the pattern match on the list $\ell$, we apply the \textsc{(T-MatD)} rule, where $T \defeq  \mathsf{NatList}^{(\rabs{x}{\bbN}{0},\rabs{x}{\sprod{\bbN}{\bbN}}{0} ) }$:
  \begin{mathpar}\footnotesize
  \inferrule{  \jctxsharing{\Gamma}{\Gamma_1}{\Gamma_2 } \\ \jatyping{\Gamma_1}{\ell}{\mathsf{NatList}^{(\rabs{x}{\bbN}{\mathbf{ite}(y>x,1,0)}, \rabs{x}{\sprod{\bbN}{\bbN}}{0}  ) }}  \\ \jstyping{\Gamma_2,\ell = 0}{e_1}{T} \\ \jstyping{\Gamma_2,h:\tnat,t:\mathsf{NatList}^{(\rabs{x}{\bbN}{\mathbf{ite}(y>x,1,0)}, \rabs{x}{\sprod{\bbN}{\bbN}}{0}  ) },\ell=t+1, \mathbf{ite}(y>h,1,0) }{e_2}{T}  }{ \jstyping{\Gamma}{\mathsf{matd}(\ell,\mathsf{Nil}(\_,\tuple{}).e_1, \mathsf{Cons}(h,\tuple{t}).e_2)}{T } }
  \end{mathpar}
  For the context sharing, we apportion all the potential of $\ell$ to $\Gamma_1$ and the rest of potential of $\Gamma$ to $\Gamma_2$.
  In fact, since $y$ and $f$ do not carry potentials, the context $\Gamma_2$ is potential-free \ie $\jctxsharing{\Gamma_2}{\Gamma_2}{\Gamma_2}$.
  For the $\mathsf{Nil}$-branch, $e_1$ is a value that describes a singleton list containing $y$, thus we can easily conclude this case by rule \textsc{(SimpAtom-ConsD)} and the fact that the return type $T$ is potential-free.
  For the $\mathsf{Cons}$-branch, we first apply the \textsc{(T-Let)} rule with \textsc{(T-App-SimpAtom)} rule to derive a precise refinement type for the comparison result $b$:
  \begin{mathpar}\footnotesize
    \inferrule{ \jctxsharing{\Gamma_2}{\Gamma_2}{\Gamma_2} \\ \jstyping{\Gamma_2,h:\cdots,t:\cdots,\ell=t+1,0}{y>h}{\tsubset{\tbool}{\nu = (y > h)}} \\ \jstyping{\Gamma_2,h:\cdots,t:\cdots,\ell=t+1,\mathbf{ite}(y>h,1,0),b:\tsubset{\tbool}{\nu = (y > h)}}{e_3}{T} }{ \jstyping{\Gamma_2,h:\cdots,t:\cdots,\ell=t+1,\mathbf{ite}(y>h,1,0) }{ \elet{y>h}{b}{e_3} }{ T }}
  \end{mathpar}
  Then we use the rule \textsc{(T-Cond)} to reason about the conditional expression $e_3$:
  \begin{mathpar}\footnotesize
    \inferrule{ \jstyping{\Gamma_2,h:\cdots,t:\cdots,\ell=t+1,\mathbf{ite}(y>h,1,0),b:\tsubset{\tbool}{\nu = (y > h)}, b}{e_4}{T} \\ \jstyping{\Gamma_2,h:\cdots,t:\cdots,\ell=t+1,\mathbf{ite}(y>h,1,0),b:\tsubset{\tbool}{\nu = (y > h)}, \neg b}{e_5}{T} }{ \jstyping{\Gamma_2,h:\cdots,t:\cdots,\ell=t+1,\mathbf{ite}(y>h,1,0),b:\tsubset{\tbool}{\nu = (y > h)}}{\econd{b}{e_4}{e_5}}{T} }  
  \end{mathpar}
  By validity checking, we can show that $\jprop{y:\tnat,h:\tnat,b:\tsubset{\tbool}{\nu = (y>h)},b}{y > h}$, thus $\jprop{y:\tnat,h:\tnat,b:\tsubset{\tbool}{\nu = (y>h)},b}{ \mathbf{ite}(y>h,1,0) = 1 }$.
  Then, by the \textsc{(S-Transfer)} rule on the goal involving the then-branch $e_4$, it suffices to show that $\jstyping{\Gamma_2,h:\cdots,t:\cdots,\ell=t+1,b:\tsubset{\tbool}{\nu = (y > h)},b, 1}{e_4}{T}$.
  Note that we now have one unit of free potential in the context, so we can use it for typing the $\mathsf{tick}$ expression by \textsc{(T-Tick-P)}:
  \begin{mathpar}\footnotesize
    \inferrule{ {  \jstyping{\Gamma_2,h:\cdots,t:\cdots,\ell=t+1,b:\tsubset{\tbool}{\nu = (y > h)},b}{{\elet{\eapp{f}{t}}{t'}{\mathsf{Cons}(h,\tuple{t'})}} }{T} } }{ \jstyping{\Gamma_2,h:\cdots,t:\cdots,\ell=t+1,b:\tsubset{\tbool}{\nu = (y > h)},b, 1}{\econsume{1}{\elet{\eapp{f}{t}}{t'}{\mathsf{Cons}(h,\tuple{t'})}} }{T} }
  \end{mathpar}
  It remains to derive the type of the recursive function application $\eapp{f}{t}$, and the list construction $\mathsf{Cons}(h,\tuple{t'})$ where $t'$ is the return of the application.
  The derivation is straightforward as $f$ has type $\tarrow{\ell}{\mathsf{NatList}^{(\rabs{x}{\bbN}{\mathbf{ite}(y>x,1,0)}, \rabs{x}{\sprod{\bbN}{\bbN}}{0}  ) }}{T}$, $t$ has type  $\mathsf{NatList}^{(\rabs{x}{\bbN}{\mathbf{ite}(y>x,1,0)}, \rabs{x}{\sprod{\bbN}{\bbN}}{0}  ) }$, thus the returned list $t'$ has type $T$ and so does $\mathsf{Cons}(h,\tuple{t'})$.


We now turn to the function \T{sort} that makes use of \T{insert}:
\begin{align*}
  \T{sort} & \dblcolon  \tarrow{\ell}{ \mathsf{NatList}^{ (\rabs{x}{\bbN}{1}, \lambda(x_1{\,:\,}\bbN,x_2{\,:\,}\bbN). \mathbf{ite}(x_1>x_2,1,0)  ) } }{ \mathsf{NatList}^{(\rabs{x}{\bbN}{0}, \rabs{x}{\sprod{\bbN}{\bbN}}{0} )} } \\
  \T{sort} & =  \mathsf{fix}(f.\ell. \mathsf{matd}(\ell, \\
  & \qquad  \mathsf{Nil}(\_,\tuple{}). \mathsf{Nil}(\etriv,\tuple{}) , \\
  & \qquad \mathsf{Cons}(h,\tuple{t}). \econsume{1}{ \elet{{\eapp{f}{t}}}{t'}{\elet{\eapp{\T{insert}}{h}}{{ins}}{ \eapp{ins}{t'}}  }} )
\end{align*}
Recall that in \autoref{exa:dependent-annotation}, we explain that the type of the argument list $\ell$ defines a potential function in terms of the number of out-of-order pairs in $\ell$.
Let $\Gamma'$ be a typing context that contains the signature of \T{insert}, as well as potential-free type bindings for $f$ and $\ell$.
Using the shift operation $\lhd$ for $\mathsf{NatList}$, we are supposed to derive the following judgment for the $\mathsf{Cons}$-branch of the pattern match:
\[
\jstyping{\Gamma',h:\tnat,t: \mathsf{NatList}^{ (\rabs{x}{\bbN}{1+\mathbf{ite}(h>x,1,0)}, \lambda(x_1{\,:\,}\bbN,x_2{\,:\,}\bbN). \mathbf{ite}(x_1>x_2,1,0) )}  , \ell=t+1}{ \elet{\eapp{f}{t}}{t'}{\cdots} }{ T }.
\]
However, we get stuck here, because there is a mismatch between the argument type of $f$ \ie \T{sort}, and the shifted type of the tail list $t$ in the context.
\end{example}

\paragraph{Polymorphic recursion}
In general, it is often necessary to type recursive function calls with a type that has different potential annotations from the declared types of the recursive functions.
We achieve this using polymorphic recursion that allows recursive calls to be instantiated with types that have different potential annotations.
Although we get stuck when typing \T{sort} in \autoref{exa:insertion-sort}, we will show how our system is able to type a polymorphic version of \T{sort}, which has been informally demonstrated in \autoref{sec:overview:inductive}.

\begin{example}[Insertion sort with polymorphic recursion]\label{exa:insertion-sort-cont}
  We start with a polymorphic list type, which is supported by our implementation but not formulated in the core calculus:
  \[
  \mathsf{List}^\theta(\alpha) \equiv \mathsf{ind}^\theta_{\lhd,\pi}( \mathsf{Nil}:\tunit, \mathsf{Cons}: \tprod{(x{\,:\,} \alpha)}{ \mathsf{List}^{ \lhd_\mathsf{Cons}(x)(\theta) }(\tpot{\alpha}{ \theta(x,\nu) }) }  ),
  \]
  where 
  $\lhd=(\lhd_\mathsf{Nil},\lhd_\mathsf{Cons})$, $\pi=(\pi_\mathsf{Nil},\pi_\mathsf{Cons})$ are defined as follows:
  \begin{align*}
    \pi_\mathsf{Nil} & \defeq \lambda\_.\lambda\theta.0, & 
    \pi_\mathsf{Cons} & \defeq \lambda y. \lambda \theta. 0, \\
    \lhd_\mathsf{Nil} & \defeq \lambda\_.\lambda\theta.\star, &
    \lhd_\mathsf{Cons} & \defeq \lambda y. \lambda \theta. \theta.
  \end{align*}
  We then generalize the type signatures of \T{insert} and \T{sort} with the polymorphic list type:
  \begin{align}
    \T{insert} & \dblcolon \forall \alpha. \tarrow{y}{\alpha}{\tarrow{\ell}{ \mathsf{List}^{\lambda(x_1,x_2).0}(\tpot{\alpha}{\mathbf{ite}(y>\nu,1,0)}) }{\mathsf{List}^{\lambda(x_1,x_2).0}(\alpha)}} ,\\
    \T{sort} & \dblcolon \forall\alpha. \tarrow{\ell}{\mathsf{List}^{\lambda(x_1,x_2).\mathbf{ite}(x_1>x_2,1,0) }(\tpot{\alpha}{1})}{ \mathsf{List}^{\lambda(x_1,x_2).0  }(\alpha) } \label{eq:decl-resource-polymorphism}
  \end{align}
  Similar to the type derivation in \autoref{exa:insertion-sort}, we are supposed to derive the following judgment for the $\mathsf{Cons}$-branch of the pattern match in the implementation of \T{sort}:
\[
\jstyping{\Gamma',h:\alpha,t: \mathsf{List}^{ \lambda(x_1,x_2). \mathbf{ite}(x_1>x_2,1,0) }(\tpot{\alpha}{ 1+\mathbf{ite}(h>\nu,1,0) })  , \ell=t+1}{ \elet{\eapp{f}{t}}{t'}{\cdots} }{ T }.
\]
  Now the function $f$ is bound to the polymorphic type in \eqref{eq:decl-resource-polymorphism}.
  To type the function call $\eapp{f}{t}$, we instantiate $f$ with $\tpot{\alpha}{\mathbf{ite}(h>\nu,1,0)}$,
  \ie, $f$ has type $ \tarrow{\ell}{\mathsf{List}^{\lambda(x_1,x_2).\mathbf{ite}(x_1>x_2,1,0)}(\tpot{\alpha}{1+\mathbf{ite}(h>\nu,1,0)})}{\mathsf{List}^{\lambda(x_1,x_2).0}( \tpot{\alpha}{\mathbf{ite}(h>\nu,1,0)} )} $.
  Thus, the type of the return value $t'$ of $\eapp{f}{t}$ matches the argument type of \T{insert}, and we can derive the function application $\elet{\eapp{\T{insert}}{h}}{{ins}}{ \eapp{ins}{t'}}$ has the desired return type $\mathsf{List}^{\lambda(x_1,x_2).0}(\alpha)$.
\end{example}

\subsection{Soundness}
\label{sec:soundness}

We now extend $\mathrm{Re}^2$'s type soundness~\cite{resyn} to new features we introduced in previous sections, including refinement-level computation and user-defined inductive datatypes.
The soundness of the type system is based on progress and preservation, and takes resources into account.
The progress theorem states that if $\jstyping{q}{e}{S}$, then either $e$ is already a value, or we can make a step from $e$ with at least $q$ units of available resource.
Intuitively, progress indicates that our type system derives bounds that are indeed upper bounds on the high-water mark of resource usage.

\begin{lemma}[Progress]
	If $\jstyping{q}{e}{S}$ and $p \ge q$, then either $\jval{e}$ or there exist $e'$ and $p'$ such that $\jstep{e}{e'}{p}{p'}$.
\end{lemma}
\begin{proof}
	By strengthening the assumption to $\jstyping{\Gamma}{e}{S}$ where $\Gamma$ is a sequence of type variables and free potentials, and then induction on $\jstyping{\Gamma}{e}{S}$.
\end{proof}

The preservation theorem then relates leftover resources after a step in computation and the typing judgment for the new term to reason about resource consumption.

\begin{lemma}[Preservation]
	If $\jstyping{q}{e}{S}$, $p \ge q$ and $\jstep{e}{e'}{p}{p'}$, then $\jstyping{p'}{e'}{S}$.
\end{lemma}
\begin{proof}
	By strengthening the assumption to $\jstyping{\Gamma}{e}{S}$ where $\Gamma$ is a sequence of free potentials, and then induction on $\jstyping{\Gamma}{e}{S}$, followed by inversion on the evaluation judgment $\jstep{e}{e'}{p}{p'}$.
\end{proof}

As in other refinement type systems, purely syntactic soundness statement about results
of computations (\ie, they are well-typed values) is unsatisfactory.
Thus, we also formulate a denotational notation of \emph{consistency}. 
For example,
the literal $b=\etrue$, but not $b=\efalse$, is consistent with $\jstyping{0}{b}{\tsubset{\tbool}{\nu}}$;
A list of values $\ell = [v_1,\cdots,v_n]$ is consistent with $\jstyping{q}{\ell}{\mathsf{NatList}^{\rabs{x}{\bbN}{x}}}$, if $q \ge \sum_{i=1}^n v_i$.
%
We then show that well-typed values are \emph{consistent} with their typing judgement.

\begin{lemma}[Consistency]
  If $\jstyping{q}{v}{S}$, then $v$ satisfies the conditions indicated
  by $S$ and $q$ is greater than or equal to the potential stored in
  $v$ with respect to $S$.
\end{lemma}
\begin{proof}
  By inversion on the typing judgment we have $\jatyping{q}{v}{B}$ for some base type $B$ or $v$ is an abstraction.
  The latter case is easy as the refinement language cannot mention function values.
  For the former case, we proceed by strengthening the assumption to $\jatyping{\Gamma}{v}{B}$ where $\Gamma$ is a sequence of type variables and free potentials, then induction on $\jatyping{\Gamma}{v}{B}$.
\end{proof}

As a result of the lemmas above, we derive the following main technical theorem of this paper.

\begin{theorem}[Soundness]
  If $\jstyping{q}{e}{S}$ and $p \ge q$ then either
  \begin{itemize}
  \item $\jsteps{e}{v}{p}{p'}$ and $v$ is consistent with $\jstyping{p'}{v}{S}$ or
  \item for every $n$ there is $\tuple{e',p'}$ such that $\jstepn{e}{e'}{p}{p'}$.
  \end{itemize}
\end{theorem}

Detailed proofs are included in
\iflong 
\autoref{sec:appendixproofs}
\else
the technical report~\cite{techreport}
\fi

\section{Evaluation}
\label{sec:eval}

We have implemented the new features of liquid resource types, inductive and abstract potentials,
on top of the \resyn type checker;
we refer to the resulting implementation as \tool.
In this section, we evaluate \tool according to three metrics: 

\textbf{Expressiveness:} How well can \tool express non-linear and dependent bounds?  To what extent can \tool express bounds that systems like \resyn and \raml{} could not? 

\textbf{Automation:} Can \tool automatically verify expressive bounds which other tools cannot?  Are the verification times reasonable?

\textbf{Flexibility:}  Can we define reusable datatypes that can express a variety of resource bounds across different programs?

\subsection{Reusable Datatypes} 
\label{sec:eval:lib}
We first describe a small library of resource-annotated datatypes we created,
which we will use to specify type signatures for our benchmark functions.
The definitions of the four datatypes are listed in \autoref{fig:eval:data}.
Since potential is only specified inductively in these definitions, 
we also provide a closed form expression for the potential associated with each such data structure
(omitting the potential stored in the element type \T{a}).
The proofs of these closed forms can be found in 
\iflong
\autoref{sec:appendixclosedf}.  
\else
the technical report~\cite{techreport}.
\fi

\begin{table*}[!t]
\centering
\begin{tabular}{ >{\centering\arraybackslash}m{.5cm} >{\centering\arraybackslash}m{8.5cm}| >{\centering\arraybackslash}m{4cm}}
         & Datatype & Potential Interpretation \\ \hline \hline
      1 & 
{\begin{nanomltbl}
data List a <q::a->a->Int> where
  Nil :: List a <q>
  Cons :: x:a -> List a$^{q\, x}$ <q> -> List a <q>
\end{nanomltbl}} &
      $\sum_{i<j} q(a_i,a_j)$ \\ \hline
      2 & 
{\begin{nanomltbl}
data EList a <q::Int> where
  Nil :: EList a <q>
  Cons :: x:a$^q$ -> EList a <2*q> -> EList a <q> 
\end{nanomltbl}} &
$q\cdot (2^n - 1)$ \\ \hline 
3 &
{\begin{nanomltbl}
data LTree a <q::Int> where
  Leaf :: a -> LTree a <q>
  Node :: LTree a$^q$ <q> -> LTree a$^q$ <q> -> LTree a <q>
\end{nanomltbl}} &
$\approx q\cdot n \log_2 (n)$ \\ \hline
4 &
{\begin{nanomltbl}
data PTree a <p::a->Bool, q::Int> where
  Leaf :: PTree a <p,q>
  Node :: x:a$^q$ -> PTree a <p, $\mathsf{ite}(\mathsf{p(x), q, 0})$>
     -> PTree a <p, $\mathsf{ite}(\mathsf{p(x), 0, q})$> -> PTree a <p,q>
\end{nanomltbl}} &
$q\cdot \lvert\ell\rvert$ \\
\hline
\end{tabular}
\caption{Annotated data structures with their corresponding potential functions. 
$n$ is taken to be the number of elements in the data structure.
In \T{PTree}, $\lvert\ell\rvert$ is the length of the path specified by the predicate $p$. 
}
\label{fig:eval:data}
\end{table*}

\T{List} and \T{EList} are general purpose list data structures that contain quadratic and exponential potential, respectively. In particular, \T{List} admits dependent potential expressions, as the abstract potential parameter is a function of the list elements.  This list type can be adapted to express higher-degree polynomial potential functions via the generalized left shift operation, described in \autoref{sec:potentials}. \T{EList} can be modified to express exponential potential for any positive integer base $k$ by modifying the type of the second argument to \T{Cons}: 
\[ \T{Cons} \dblcolon \tarrow{\mathsf{x}}{\mathsf{a}^q}{\tarrow{\mathsf{xs}}{\tkadt{a}{k \cdot q}{EList}}{\tkadt{a}{q}{EList}}} \]
Such a list contains $q \cdot (k^n-1)$ units of potential; $k$ has to be fixed for annotations to remain linear. 

\T{LTree} is a binary tree with values (and thus, potential) stored in its leaves.
We show that the total potential stored in the tree is $q\cdot n \cdot h$,
where $n$ is the number of leaves in the tree and $h$ is its height.
If we additionally assume that the tree is balanced,
then $h = O(\log(n))$, and hence the amount of potential in the tree is $O(n \cdot \log(n))$.
In \autoref{sec:eval:benchmarks} we use this tree as an intermediate data structure in order to reason about logarithmic bounds.

\T{PTree} is a binary tree with elements in the nodes,
which uses dependent potential annotations to specify the exact \emph{path} through the tree that carries potential;
we refer to this data structure as \emph{pathed potential tree}. 
\T{PTree} is parameterized by a boolean-sorted \emph{abstract refinement}~\cite{VazouRoJh13}, $p$, 
which is then used in the potential annotations to conditionally allocate potential either to the left or to the right subtree,
depending on the element in the node.
Since $p$ is used to pick exactly one subtree at each step, it specifies a path from root to leaf.

These data structures showcase a variety of ways in which liquid resource types can be used to reason about a program's performance. 
Additionally, because the interpretation of abstract potentials is left entirely to the user, 
one can define custom data structures to describe other resource bounds as needed.  

\subsection{Benchmark Programs}
\label{sec:eval:benchmarks}
We evaluate the expressiveness of \tool on a suite of \numBench benchmark programs listed in \autoref{fig:eval:functions}.
The resource consumptions of these benchmarks covers a wide range of complexity classes.  
We choose functions with quadratic, exponential, logarithmic, and value-dependent resource bounds in order to showcase the breadth of bounds \tool can verify.  
We are able to express these bounds using only the datatypes from \autoref{fig:eval:data}, showing the flexibility and reusability of these datatype definitions. 
The cost model in all benchmarks is the number of recursive calls (as in \autoref{sec:overview}). 

\begin{table*}[!t]
\begin{center}
\footnotesize
\resizebox{\textwidth}{!}{
\begin{tabular}{c|c|c|c|c|c} 
Type & No. & Description & Type Signature & t (s) & Source \\ \hline \hline

\multirow{7}{*}{\begin{tabular}{c}Polynomial\\Quadratic\\Potential\end{tabular}}
&1 & All ordered pairs & $\tkaplist{a}{2}{2} \to \tklist{(Pair\ a)}$ & 0.5 & \raml\\

&2 & List Reverse & $\tkaplist{a}{2}{1} \to \tklist{a}$ & 0.4 & \synquid \\

&3 & List Remove Duplicates& $\tkaplist{a}{2}{1} \to \tklist{a}$ & 0.4 & \synquid \\

&4 & Insertion Sort (Coarse) & $\tkaplist{a}{2}{1} \to \tklist{a}$ & 0.6 & \synquid \\

&5 & Selection Sort & $\tkaplist{a}{4}{3} \to \tklist{a}$ & 0.5 & \synquid \\

&6 & Quick Sort & $\tkaplist{a}{3}{3} \to \tklist{a}$ & 1.0 & \synquid \\

&7 & Merge Sort & $\tkaplist{a}{2}{2} \to \tklist{a}$ & 0.9 & \synquid \\

\hline

\multirow{2}{*}{\begin{tabular}{c}Non-Polynomial\\Potential\end{tabular}}
&8 & Subset Sum & $\tkadt{Int}{2}{EList} \to \T{Int} \to \T{Bool}$ & 0.3 & --\\
&9 & Merge Sort Flatten & $\tkapdt{a}{1}{1}{LTree} \to \tklist{a}$ & 0.9 & --\\
\hline
\multirow{3}{*}{\begin{tabular}{c}Value-\\Dependent\\Potential\end{tabular}}
&10 & Insertion Sort (Fine) & $\tkaplist{a}{1}{\lambda x_1,x_2. \, \mathsf{ite}(x_1 < x_2,1,0)} \to \tklist{a}$& 5.4 & \relcost \\
&11 & BST Insert & $\tarrow{x}{\T{a}}{\tkadt{a}{\lambda x_1. \, x < x_1, 1}{PTree} \to \tkadt{a}{\lambda x_1. \, x < x_1, 0}{PTree}}$ & 2.4 & --\\
&12 & BST Member & $\tarrow{x}{\T{a}}{\tkadt{a}{\lambda x_1. \, x < x_1, 1}{PTree} \to \T{Bool}}$ & 6.0 & --\\ 
\end{tabular}
}
\end{center}
\caption{ Functional benchmarks. For each benchmark, we list its type signature, verification time (t), and source for the benchmark -- either \raml{} \cite{RAML10}, \synquid \cite{PolikarpovaKS16}, or \relcost \cite{Radicek18}.}
\label{fig:eval:functions}
\end{table*}

Benchmarks 1-7 require only standard quadratic bounds.
Benchmarks 2-7 are those programs from the original \synquid benchmark suite~\cite{PolikarpovaKS16}
that \resyn could not handle, because they require non-linear bounds.
Some of the analyses, such as merge sort, are overapproximate.
Benchmark 8 moves beyond polynomials, solving the well-known subset sum problem. The function runs in exponential time, so we can write a resource bound using our \T{EList} data structure to require exponential potential in the input. 
Once we have verified that $\T{subsetSum} \dblcolon \tkadt{Int}{2}{EList} \to \T{Int} \to \T{Bool}$, we can use the provided closed-form potential function to calculate total resource usage: $2(2^n - 1)$, exactly the number of recursive calls made at runtime.

Benchmark 9 illustrates how \tool can verify a more precise $O(n \log(n))$ bound for a version of merge sort.
LRT is unable to allocate logarithmic amount of potential directly to a list,
hence we specify this benchmark using \T{LTree} as an intermediate data structure.
Prior work has shown~\cite{morphisms} that merge sort can be written more explicitly as a composition of two function:
\T{build}, which converts a list into a tree 
(where each internal node represents a split of a list into halves),
and \T{flatten}, which takes a tree and recursively merges its subtrees into a single sorted list.
In the traditional implementation of merge sort, the two passes are fused, and the intermediate tree is never constructed;
however, keeping this tree explicit, enables us to specify a logarithmic bound on the \T{flatten} phase of merge sort,
which performs the actual sorting.
We accomplish this by typing its input as $\tkapdt{a}{1}{1}{LTree}$;
because \T{build} always constructs balanced trees,
this tree carries approximately $n \log(n)$ units of potential,
where $n$ is the number of leaves in the tree, which coincides with the number of list elements.
Unfortunately, LRT is unable to express a precise resource specification for the \T{build} phase of merge sort,
or for the traditional, fused version without the intermediate tree.

Benchmarks 10 through 12 show the expressiveness of value-dependent potentials. 
Benchmark 10 is the dependent version of insertion sort from \autoref{sec:overview}. 
%
Benchmarks 11 and 12 use the \T{PTree} data structure to allocate linear potential along a value-dependent path in a binary tree.
We use \T{PTree} to specify the resource consumption of inserting into and checking membership in a binary search tree. 
\T{PTree} allows us to assign potential only along the specific path taken while searching for the relevant node in the tree.
As a result, we can endow our tree with exactly the amount of potential required to execute \T{member} or \T{insert} on an arbitrary BST.  
If we have the additional guarantee that our BST is balanced, we can also conclude that these bounds are logarithmic, 
as the relevant path is the same length as the height of the tree. 
%

\subsection{Discussion and Limitations}
\label{sec:eval:discussion}

\autoref{fig:eval:functions} confirms that \tool is reasonably efficient:
verification takes under a second for simple numeric bounds (benchmarks 1-9).
The precision of value-dependent bounds (benchmarks 10-12) comes with slightly higher verification times---up to six seconds;
these benchmarks generate second-order CLIA constraints and require the use of \resyn's CEGIS solver
(as opposed to first-order constraints that can be handled by an SMT solver).

No other automated resource analysis system can verify all of our benchmarks in \autoref{sec:eval:benchmarks}.
\relcost can be used to verify all of these bounds, but provides no automation.
\resyn cannot verify any of our benchmarks, as it can only reason about linear resource consumption.
\raml{} can infer an appropriate bound for benchmarks 1-7, which all require quadratic potential.
However, \raml{} cannot reason about the other examples, as it cannot reason about program values, and only support polynomials.
\raml{} relies on a built-in definition of potential in a data structure, 
while \tool exposes allocation of potential via datatype declarations,
allowing the programmer to easily configure it to handle non-polynomial bounds. 
In particular, a \tool user can adopt \raml's treatment of polynomial resource bounds via our \T{List} type,
and can also write non-polynomial specifications with other datatypes from our library
or with a custom datatype.

The \relcost formalism presented in \cite{Radicek18} allows one to manually
verify all of the bounds in \autoref{sec:eval:benchmarks}.
\cite{CicekQBG019} presents an implementation of a subset \relcost.
This tool can be used to automatically verify non-linear bounds that are dependent
only on the length of a list.
To verify non-linear bounds, the system still generates non-linear constraints,
and thus relies on incomplete heuristics for constraint solving.
Benchmarks 10-12 in \autoref{fig:eval:functions} all consist of conditional bounds,
which are not supported by the implementation of \relcost.

Despite \tool's flexibility, it has some limitations.
Firstly, 
our resource bounds must be defined inductively over the function's input,
and hence we cannot express bounds that do not match the structure of the input type.
A prototypical example is the logarithmic bound for merge sort:  
we can specify this bound for the \T{flatten} phase, which operates over a tree (where the logarithm is ``reified'' in the tree height),
but not for merge sort as a whole that operates over a list.


Secondly, \tool cannot express multivariate resource bounds.
Consider a function that takes two lists and returns a list of every pair in the cartestian product of the two inputs. 
This function runs in $O(m\cdot n)$, where $m$ and $n$ are the lengths of the two input lists. 
There is no way to express this bound by annotating the types of input lists with terms form CLIA.

Finally, \tool can verify, but not infer resource bounds.  So while verification is automatic, finding the correct type signature must be done manually, even if the correct data structure has been selected.  
Simple modifications would allow the system to infer non-dependent resource bounds following the approach of \raml ~\cite{RAML10}, but this technique does not generalize to the dependent case.

\section{Related work}
\label{sec:related}

Verification and inference techniques for resource analysis have been
extensively studied.
Traditionally, automatic techniques for resource analysis are based on
a two-phase process: (1) extract recurrence relations from a program
and (2) solve recurrence relations to obtain a closed-form bound.
This strategy has been pioneered by Wegbreit~\cite{Wegbreit75} and has
been later been studied for imperative
programs~\cite{kn:AAG11,TACAS:AFR15} using
techniques such as abstract interpretation and symbolic
analysis~\cite{PLDI:KBB17,POPL:KCB19}.
The approach can also be used for higher-order functional
programs by extracting higher-order recurrences~\cite{ICFP:DLR15}.
Other resource analysis techniques are based on static
analysis~\cite{POPL:GMC09,GulwaniJK09,SAS:ZSG11,SinnZV14} and term
rewriting~\cite{RTA:AM13,TLCA:HM15,JAR:NEG13,TACAS:BEG14}.

Most closely related to our work are type-based approaches to resource
bound analysis.
We biuld upon type-based automated amortized resource analysis (AARA).
AARA has been introduced by Hofmann and Jost~\cite{POPL:HJ03} to
automatically derive linear bounds on the heap-space consumption of
first-order programs.
It has then been extended to higher-order programs~\cite{POPL:JHL10},
polynomial bounds~\cite{ESOP:HH10,POPL:HAH11} and user-defined
types~\cite{POPL:JHL10,POPL:HDW17}.
Most recently, AARA has been combined with refinement types~\cite{PLDI:FP91} in
the Re$^2$ type system~\cite{resyn} behind \resyn, a resource-aware program
synthesizer.
None of these works support user-defined potential functions.
As discussed in \autoref{sec:intro}, this paper extends Re$^2$
with inductive datatypes that can be annotated with custom potential functions.
The introduction of abstract potential functions allows this work
to reuse \resyn's constraint solving infrastructure when
reasoning about richer resource bounds.
This work also formalizes the technique for user-defined inductive datatypes,
while the Re$^2$ formalism admitted only reasoning about lists.

Several other works have used refinement types and dependent types for
resource bound analysis. 
Danielsson~\cite{Danielsson08} presented a dependent cost monad that
has been integrated in the proof assistant Agda.
d$\ell$PCF~\cite{LICS:LG11} introduced linear dependent types to reason
about the worst-case cost of PCF terms.
Granule~\cite{orchard2019} introduces graded modal types, combining
the indexed types of d$\ell$PCF with bounded linear logic~\cite{GirardSS92}
and other modal type systems~\cite{ghica2014,brunel2014}.
While useful for a variety of applications, such as 
enforcing stateful protocols, reasoning about privacy,
and bounding variable reuse, these techniques do not allow 
an amortized resource analysis.
{\c{C}}i{\c{c}}ek et al.~\cite{POPL:CBG17,CicekQBG019} have pioneered the use of relational
refinement type systems for verifying the bounds on the difference
of the cost of two programs.
It has been shown that linear AARA can be embedded in a generalized
relational type systems for monadic refinements~\cite{RadicekBG0Z18}.
While this article does not consider relational verification, the
presented type system allows for decidable type checking and is a
conservative extension of AARA instead of an embedding.

Similarly, TiML~\cite{OOPSLA:WWC17} implements (non-relational)
refinement types in the proof assistant Coq to aid verification of
resource usage.
A recent article also studied refinement types for a
language with lazy evaluation~\cite{HandleyVH20}.
However, these works do not directly support amortized analysis and do
not reduce type checking of non-linear bounds to linear constraints.

\begin{acks}
  We thank the anonymous reviewers and our shepherd, Richard
  Eisenberg, for their valuable and detailed feedback on earlier
  versions of this article.

  This article is based on research supported by DARPA under AA
  Contract FA8750-18-C-0092 and by the National Science Foundation
  under SaTC Award 1801369, CAREER Award 1845514, and SHF Awards
  1812876, 1814358, and 2007784.
  Any opinions, findings, and conclusions contained in this document
  are those of the authors and do not necessarily reflect the views of
  the sponsoring organizations.

\end{acks}

\bibliography{references,db}

\iflong
\newpage
\appendix
\section{Appendix}
\label{sec:appendixclosedf}
In order for our closed-form potential function to be valid, we need to check that for each data structure we require that for each constructor, the sum of potentials of the arguments $r_1,\dots,r_s$ of that constructor equals the potential of the overall data structure.  
\subsection{List: Dependent Quadratic Potential}
\begin{nanoml}
data List a <q::a -> a -> Nat> where
  Nil::List a <q>
  Cons::x:a -> xs:List a$^{q(x,_v)}$ <q> -> List a <q>
\end{nanoml}

\noindent
{\it Claim. } Let $\ell \dblcolon \tklist{a} \langle q \rangle$ be of length $n$, so that $\ell$ is $[a_1,\dots,a_n]$.  Then: 
$$\Phi_{\T{List}}(\ell) = \sum_{1\leq i < j\leq n} q(a_i,a_j) \text{ is a sound potential function.}$$

\noindent
{\it Proof. } Matching $\ell$ to \T{Nil}, we have
$$\sum_i\Phi_{\T{List}}(u^{\T{Nil}}_i) = 0 = \sum_{1\leq i < j\leq n} q(a_i,a_j) = \Phi_{\T{List}}(\ell).$$
Matching $\ell$ to \T{Cons x xs}, we have
\begin{align*}
  \sum_i\Phi_{\T{List}}(u^{\T{Cons}}_i) = \Phi_{\T{List}}(\T{xs}) = \sum_{2\leq i \leq n} q(a_1,a_i) + \sum_{2\leq i < j\leq n} q(a_i,a_j)\\
  =\sum_{1\leq i < j \leq n} q(a_i,a_j) = \Phi_{\T{List}}(\ell).
\end{align*}

\subsection{List: Exponential Potential}

\begin{nanoml}
data EList a <q::Int> where
  Nil::EList a <q>
  Cons::x:a$^q$ -> xs: EList a <q + q> -> EList a <q>
\end{nanoml}

\noindent
{\it Claim. } Let $\ell \dblcolon \tkadt{a}{q \dblcolon \T{Nat}}{EList}$ be of length $n$.  Then:
$$\Phi_{\T{EList}}(\ell) = q * (2^n - 1) \text{ is a sound potential function}.$$

\noindent
{\it Proof. } Matching $\ell$ to \T{Nil}, we have
$$\sum_i \Phi_{\T{EList}} (u^{\T{Nil}_i}) = 0 = q * (2^n - 1) = \Phi_{EList}.$$

\noindent
Matching $\ell$ to \T{Cons x xs}, we have
\begin{align*}
  \sum_i \Phi_{\T{EList}} (u^{\T{Cons}}_i) = \Phi_{\T{EList}} (x) + \Phi_{\T{EList}} = q + (q+q)\left(2^{n-1}-1\right) = q * (1 + 2*2^{n-1}-2) \\
  = q + 2q * (2^{n-1} - 1) = q*(2^n - 1).
\end{align*}

\subsection{Balanced Binary Tree: Logarithmic potential}

\begin{nanoml}
data LTree a <q::Nat> where
  Leaf::a$^q$ -> LTree a <q>
  Node::LTree a$^q$ <q> -> LTree a$^q$ <q> -> LTree a <q>
\end{nanoml}

\noindent
{\it Claim. } Let $t \dblcolon \tkadt{a}{q \dblcolon \T{Nat}}{LTree}$ be a balanced tree storing $n$ values (leaves).  Then
$$\Phi_{LTree} (t) = q * (n\log_2(n)) \text{ is a sound potential function.}$$

\noindent
{\it Proof. } Matching $t$ to \T{Leaf}, we have
$$\sum_i \Phi_{\T{LTree}} (u^{\T{Leaf}}_i) = q = q * n\log_2(n) = \Phi_{\T{LTree}} (t).$$

\noindent
Matching $t$ to \T{Node x l r}, we have
\begin{align*}
    \sum_i \Phi_{\T{LTree}} (u^{\T{Node}}_i) = \Phi_{\T{LTree}} (l) + \Phi_{\T{LTree}} (r) = \Phi_{\T{LTree}} (x) + 2 * \Phi_{\T{LTree}} (l)\\
    = 2\left[q*\left(\frac{n}{2}\right) + q*\left(\frac{n}{2}\log_2\left(\frac{n}{2}\right)\right)\right] = q*\left[n + n(\log_2(n) - 1)\right] = q*n\log_2(n) = \Phi_{\T{LTree} (t)}.
\end{align*}

\subsection{Tree: Potential on path}

\begin{nanoml}
data PTree a <p::a -> Bool, q::Int> where
  Leaf :: PTree a <p,q>
  Node :: x: a$^q$
       -> l: PTree a <p,if (p x) then q else 0>
       -> r: PTree a <p,if (p x) then 0 else q>
       -> PTree a <p,q>
\end{nanoml}

\noindent
{\it Claim. } Let $t \dblcolon \tkadt{a}{p \dblcolon \T{a -> Bool}, q \dblcolon Nat}{PTree}$.  Then
$$\Phi_{\T{PTree}} (t) = q*\lvert \ell \rvert \text{ is a sound potential function,}$$
for the path $\ell$ to leaf defined by the predicate \T{p} having length $\lvert \ell \rvert$. 

\noindent
{\it Proof. } Matching $t$ to \T{Leaf}, we have
$$\sum_i \Phi_{\T{PTree}} (u^{\T{Leaf}}_i) = 0 = q*\lvert \ell \rvert = \Phi_{\T{PTree}} (t).$$

\noindent
Matching $t$ on $\T{Node x l r}$, we have
\begin{align*}
    \sum_i \Phi_{\T{PTree} (u^{\T{Node}}_i)} = \Phi_{\T{PTree}} (x) + \Phi_{\T{PTree}} (l) + \Phi_{\T{PTree}} (r) 
    \\ = q + \text{ (if $p(x)$ then $q$ else $0$) }*(\lvert \ell\rvert - 1) + \text{ (if $p(x)$ then $0$ else $q$) }*(\lvert\ell\rvert - 1)\\
    = q + q * (\lvert \ell\rvert - 1) = q *\lvert \ell \rvert = \Phi_{\T{PTree}} (t).  
\end{align*}

\mathtoolsset{showonlyrefs,showmanualtags}
\allowdisplaybreaks

\newcommand{\propref}[1]{Prop.~\ref{prop:#1}}
\newcommand{\lemref}[1]{Lem.~\ref{lem:#1}}
\newcommand{\theoref}[1]{Thm.~\ref{the:#1}}

\renewcommand{\tinduct}[5][\mu,\lhd,\pi]{\ensuremath{\mathsf{ind}_{#1}^{#5}\p{ \many{ #2 {\,:\,} (#3,#4) } }}}

\section{Full Specification of the Type System}
\label{sec:appendixre2}

\subsection{Inductive Datatypes with Measures}

In the full type system, an inductive datatype $\tinduct{C}{T}{m}{\theta}$ consists of a sequence of constructors, each of which has a name $C$, a content type $T$ (which must be a scalar type), and a finite number $m \in \bbZ_0^+$ of child nodes.
The parameter $\mu$ specifies a \emph{measure} of the inductive datatype, which is a tuple of refinement-level functions, which is used to derive the interpretation $\calI_D$ for a datatype $D$.
Intuitively, the $j$-th component of $\mu$, written $\mu.\mathbf{j}$, should be a function of sort $\sarrow{\Delta_{T_j}}{\sarrow{\Delta_D^{m_j}}{\Delta_D}}$ where $\Delta_{T_j}$ is the sort for refinements of the content type $T_j$, \ie, a function computes the measurement of a data structure $C_j(v_0,\tuple{v_1,\cdots,v_{m_j}})$ as $\mu.\mathbf{j}(s_0)(s_1,\cdots,s_{m_j})$, where $s_0$ is the logical refinement of the content value $v_0$, and $s_1,\cdots,s_{m_j}$ are $\Delta_D$-sorted measurements of child nodes $v_1,\cdots,v_{m_j}$.

\begin{example}[Measures in the refinement language]\label{appendix:exa:refinement-measures}
  Recall the length measure $\calI_{\mathsf{NatList}}$ for natural-number lists in \autoref{exa:inductive-measures}.
  We want to redefine $\mathsf{NatList}$ as $\mathsf{ind}_\mu(\mathsf{Nil}{\,:\,}(\tunit,0),\mathsf{Cons}{\,:\,}(\tnat,1))$ with some proper $\mu$ such that $\calI_\mathsf{NatList}$ can be derived from $\mu$.
  Indeed, we can define $\mu = (\mu_{\mathsf{Nil}},\mu_{\mathsf{Cons}})$ where
  \begin{align*}
    \mu_\mathsf{Nil} & \defeq \rabs{\_}{\bbU}{ \rabs{\_}{\bbU}{0}}, & \mu_\mathsf{Cons} & \defeq \rabs{h}{\bbN}{\rabs{t}{\bbN}{t + 1}}.
  \end{align*}
  The first argument of both $\mu_\mathsf{Nil}$ and $\mu_\mathsf{Cons}$ reflects the corresponding content type.
  In addition, $\mu_\mathsf{Cons}$ has a second argument that represents the measurement of the child node \ie the length of the tail list.
  We can now define $\calI_\mathsf{NatList}$ as
  \begin{align*}
    \calI_\mathsf{NatList}(\mathsf{Nil}(\etriv,\tuple{})) & \defeq \mu_{\mathsf{Nil}}(\calI(\etriv))(\star) \equiv 0, \\
     \calI_\mathsf{NatList}(\mathsf{Cons}(v_h,\tuple{v_t})) & \defeq \mu_{\mathsf{Cons}}(\calI(v_h))(\calI_\mathsf{NatList}(v_t)) \equiv \calI_\mathsf{NatList}(v_t) + 1.
  \end{align*}
\end{example}

Inspired by \autoref{appendix:exa:refinement-measures}, for a general inductive datatype $D$ with the form $\mathsf{ind}_\mu(\many{C{\,:\,}(T,m)})$, we can inductively define a measure $\calI_D$ for values of this datatype using $\mu$:
\begin{equation}\label{eq:interpretation-from-measure}
  \calI_D(C_j(a_0,\tuple{a_1,\cdots,a_{m_j}})) \defeq \mu.\mathbf{j}(\calI(a_0))(\calI(a_1),\cdots,\calI(a_{m_j})).
\end{equation}

\subsection{Sorting: $\jsort{\Gamma}{\psi}{\Delta}$}

Refinements are classified by sorts.
The \emph{sorting} judgment $\jsort{\Gamma}{\psi}{\Delta}$ states that a refinement $\psi$ has a sort $\Delta$ under a context $\Gamma$.
The typing context is needed because refinements can reference program variables.
\autoref{fig:sorting} presents the sorting rules.
To reflect types of program variables in the refinement level, we define a relation $B \rightsquigarrow \Delta$ as follows.
The relation $\rightsquigarrow$ defines a total function from base types to sorts.

\begin{mathpar}\footnotesize
\Rule{K-Nat}
{  }
{ \jkind{\tnat}{\bbN} }
\and
\Rule{K-Bool}
{ }
{ \jkind{\tbool}{\bbB} }
\and
\Rule{K-Unit}
{ }
{ \jkind{\tunit}{\bbU} }
\and
\Rule{K-Prod}
{ \jkind{B_1}{\Delta_1} \\\\ \jkind{B_2}{\Delta_2} }
{  \jkind{\tprod{B_1}{B_2}}{\sprod{\Delta_1}{\Delta_2}} }
\and
\Rule{K-Dtype}
{ D = \tinduct{C}{T}{m}{\theta}  }
{ \jkind{ D }{\Delta_{D} } }
\and
\Rule{K-Tvar}
{ }
{ \jkind{m \cdot \alpha}{\delta_\alpha} }
\end{mathpar}

We also define a relation $\tscalar{\Delta}$ to state a sort $\Delta$ is first-order as follows.
The relation is used to define first-order quantifications.

\begin{mathpar}\footnotesize
\Rule{Sc-Bool}{ }{ \tscalar{\bbB} }
\and
\Rule{Sc-Nat}{ }{ \tscalar{\bbN} }
\and
\Rule{Sc-Unit}{ }{ \tscalar{\bbU} }
\and
\Rule{Sc-Tvar}{ }{ \tscalar{\delta_\alpha} }
\and
\Rule{Sc-Prod}{ \tscalar{\Delta_1} \\ \tscalar{\Delta_2} }{ \tscalar{\sprod{\Delta_1}{\Delta_2}} }
\end{mathpar}

\begin{figure}[h]
\begin{flushleft}
  \small\fbox{$\jsort{\Gamma}{\psi}{\Delta}$}
\end{flushleft}
\begin{mathpar}\footnotesize
	\Rule{S-Var}
	{ {\jwfctxt{\Gamma}} \\ \jkind{ \Gamma(x) }{ \Delta }}
	{ \jsort{\Gamma}{x}{\Delta} }
	\and
	\Rule{S-Nat}
	{ {\jwfctxt{\Gamma}} }
	{ \jsort{\Gamma}{n}{\bbN} }
	\and
	\Rule{S-Triv}
	{ \jwfctxt{\Gamma} }
	{ \jsort{\Gamma}{\star}{\bbU} }
	\and
	\Rule{S-Top}
	{ {\jwfctxt{\Gamma}} }
	{ \jsort{\Gamma}{\top}{\bbB} }
	\and
	\Rule{S-Neg}
	{ \jsort{\Gamma}{\psi}{\bbB} }
	{ \jsort{\Gamma}{\neg\psi}{\bbB} }
	\and
	\Rule{S-And}
	{ \jsort{\Gamma}{\psi_1}{\bbB} \\ \jsort{\Gamma}{\psi_2}{\bbB} }
	{ \jsort{\Gamma}{\psi_1 \wedge \psi_2}{\bbB} }
	\and
	\Rule{S-Rel}
	{ \jsort{\Gamma}{\phi_1}{\bbN} \\ \jsort{\Gamma}{\phi_2}{\bbN} }
	{ \jsort{\Gamma}{\phi_1 \leq \phi_2}{\bbB} }
	\and
	\Rule{S-Op}
	{ \jsort{\Gamma}{\phi_1}{\bbN} \\ \jsort{\Gamma}{\phi_2}{\bbN} }
	{ \jsort{\Gamma}{\phi_1 + \phi_2}{\bbN} }
	\and
	\Rule{S-Eq}
	{ \jsort{\Gamma}{\psi_1}{\Delta} \\ \jsort{\Gamma}{\psi_2}{\Delta} \\ \tscalar{\Delta} }
	{ \jsort{\Gamma}{\psi_1 = \psi_2}{\bbB} }
	\and
	\Rule{S-Lvar}
	{ \jwfctxt{\Gamma} \\ \Gamma(a) = \Delta }
	{ \jsort{\Gamma}{a}{\Delta} }
	\and
	\Rule{S-Abs}
	{ \jsort{\Gamma,a:\Delta}{\psi}{\Delta'} }
	{ \jsort{\Gamma}{\rabs{a}{\Delta}{\psi}}{\sarrow{\Delta}{\Delta'}} }
	\and
	\Rule{S-App}
	{ \jsort{\Gamma}{\psi_1}{\sarrow{\Delta}{\Delta'}} \\
	\jsort{\Gamma}{\psi_2}{\Delta} }
	{ \jsort{\Gamma}{\psi_1~\psi_2}{\Delta'} }
	\and
\Rule{S-Pair}
{ \jsort{\Gamma}{\psi_1}{\Delta_1} \\ \jsort{\Gamma}{\psi_2}{\Delta_2} }
{ \jsort{\Gamma}{\rpair{\psi_1}{\psi_2}}{\sprod{\Delta_1}{\Delta_2}} }
\\
\Rule{S-Proj-Left}
{ \jsort{\Gamma}{\psi}{\sprod{\Delta_1}{\Delta_2}} }
{ \jsort{\Gamma}{\projl{\psi}}{\Delta_1} }
\and
\Rule{S-Proj-Right}
{ \jsort{\Gamma}{\psi}{\sprod{\Delta_1}{\Delta_2}} }
{ \jsort{\Gamma}{\projr{\psi}}{\Delta_2} }
\and
\Rule{S-Forall}
{ \jsort{\Gamma,a:\Delta}{\psi}{\bbB} \\ \tscalar{\Delta} }
{ \jsort{\Gamma}{\rforall{a}{\Delta}{\psi}}{\bbB} }
\end{mathpar}
\caption{Sorting rules}
\label{fig:sorting}
\end{figure}

\subsection{Type Wellformedness: $\jwftype{\Gamma}{S}$}

A type $S$ is said to be \emph{wellformed} under a context $\Gamma$ if the following three properties hold:
\begin{itemize}
	\item every referenced program variables in $S$ is in the correct scope, and
	\item polymorphic types can never carry positive potential.
\end{itemize}

\autoref{fig:wftype} presents the type wellformedness rules.

\begin{figure}[h]
\begin{flushleft}
  \small\fbox{$\jwftype{\Gamma}{S}$}
\end{flushleft}
\begin{mathpar}\footnotesize
	\Rule{Wf-Nat}
	{ \jwfctxt{\Gamma} }
	{ \jwftype{\Gamma}{\tnat} }
	\and
	\Rule{Wf-Bool}
	{ \jwfctxt{\Gamma} }
	{ \jwftype{\Gamma}{\tbool} }
	\and
	\Rule{Wf-Unit}
	{ \jwfctxt{\Gamma} }
	{ \jwftype{\Gamma}{\tunit} }
	\and
	\Rule{Wf-Prod}
	{ \jwftype{\Gamma}{B_1} \\ \jwftype{\Gamma}{B_2} }
	{ \jwftype{\Gamma}{\tprod{B_1}{B_2}} }
	\and
	\Rule{Wf-Dtype}
	{ \Gamma \vdash \tinduct{C}{T}{m}{} ~\mathsf{indexed~by}~\Delta_\theta \\
	\jsort{\Gamma}{\theta}{\Delta_\theta} }
	{ \jwftype{\Gamma}{ \tinduct{C}{T}{m}{\theta} } }
	\and
	\Rule{Wf-Tvar}
	{ \jwfctxt{\Gamma} \\ \alpha \in \Gamma  }
	{ \jwftype{\Gamma}{m \cdot \alpha} }
	\and
	\Rule{Wf-Refined}
	{ \jwftype{\Gamma}{B} \\ \jsort{\Gamma,\nu:B}{\psi}{\bbB} }
	{ \jwftype{\Gamma}{\tsubset{B}{\psi}} }
	\and
	\Rule{Wf-Arrow}
	{ \jwftype{\Gamma}{T_x}  \\ \jwftype{\Gamma,x:T_x}{T} }
	{ \jwftype{\Gamma}{\tarrowm{x}{T_x }{T}{m}} }
	\and
	\Rule{Wf-Pot}
	{ \jwftype{\Gamma}{R} \\ \jsort{\Gamma,\nu:R}{\phi}{\bbN} }
	{  \jwftype{\Gamma}{\tpot{R}{\phi}} }
	\and
	\Rule{Wf-Poly}
	{ \jsharing{\Gamma,\alpha}{S}{S}{S} }
	{ \jwftype{\Gamma}{\forall\alpha.S} }
\end{mathpar}
\caption{Type well-formedness rules}
\label{fig:wftype}
\end{figure}

%

In the rule \textsc{(Wf-Dtype)}, we make use of another judgment to check the well-definedness of datatypes $D=\tinduct{C}{T}{m}{\theta}$.
Our metatheory does \emph{not} impose a specific definition of well-definedness of inductive datatypes, but rather states that these types are consistent with their subtyping and sharing relation.
For example, for subtyping, we want to ensure that if $B = \tinduct{C}{T}{m}{\theta} <: \tinduct{C}{T'}{m}{\theta'} = B'$, then for every $j$, the shifted types for children nodes of the $j$-th constructor of $B$ and $B'$ satisfy the subtyping relation accordingly.
The reasoning on sharing follows the same scheme as subtyping.
The rule \textsc{(Dtype-Index)} below formalizes the idea.

\begin{tiny}
\[
\Rule{Dtype-Index}{
	\forall j{\,:\,} \jkind{T_j}{\Delta_{T_j}} \\
	\jsort{\Gamma}{\mu}{ \textstyle\prod_{j=1}^k (\Delta_{T_j} \Rightarrow \Delta_D^{m_j} \Rightarrow \Delta_D) } \\
	\jsort{\Gamma}{\lhd}{\textstyle\prod_{j=1}^k (\Delta_{T_j} \Rightarrow \Delta_\theta \Rightarrow \Delta_\theta^{m_j} ) } \\
	\jsort{\Gamma}{\pi}{\textstyle\prod_{j=1}^k (\Delta_{T_j} \Rightarrow \Delta_\theta \Rightarrow \bbN) }  \\\\
	\Gamma_{<:} = \Gamma,\theta:\Delta_\theta, \theta':\Delta_\theta ,\tinduct{C}{T}{m}{\theta} <: \tinduct{C}{T'}{m}{\theta'} \\
	\text{for each $j$}, \jsubty{\Gamma_{<:},y:T_j }{ \tpot{( \textstyle\prod_{i=1}^{m_j} \tinduct{C}{T}{m}{\lhd_{j}(y)(\theta).\mathbf{i}} )}{ \pi.\mathbf{j}(y)(\theta) } }{ \tpot{ (\textstyle\prod_{i=1}^{m_j} \tinduct{C}{T'}{m}{\lhd_{j}(y)(\theta').\mathbf{i}})}{ { \pi.\mathbf{j}(y)(\theta') }} } \\
	\Gamma_{\sharing} = \Gamma,\theta:\Delta_\theta,\theta_1:\Delta_\theta,\theta_2:\Delta_\theta, \tinduct{C}{T}{m}{\theta} \sharing \tinduct{C}{T_1}{m}{\theta_1} \mid \tinduct{C}{T_2}{m}{\theta_2} \\\\
	\text{for each $j$,}
	\\\\
	 \jsharing{\Gamma_{\sharing}, y:T_j }{ \tpot{( \textstyle\prod_{i=1}^{m_j} \tinduct{C}{T}{m}{\lhd_{j}(y)(\theta).\mathbf{i} })}{ \pi.\mathbf{j}(y)(\theta) } }{ \tpot{( \textstyle\prod_{i=1}^{m_j} \tinduct{C}{T_1}{m}{\lhd_{j}(y)(\theta_1).\mathbf{i}} )}{ \pi.\mathbf{j}(y)(\theta_1) } }{ \tpot{( \textstyle\prod_{i=1}^{m_j} \tinduct{C}{T_2}{m}{\lhd_{j}(y)(\theta_2).\mathbf{i}} )}{ \pi.\mathbf{j}(y)(\theta_2) } }
	}
	{ \Gamma \vdash  \tinduct{C}{T}{m}{} ~\mathsf{indexed~by}~\Delta_\theta }
\]
\end{tiny}

Note that there are subtyping and sharing relations appearing in the antecedents of the judgments, which are not covered in our definition of typing contexts.
In the metatheory, we ``instantiate'' the rule above with some properly designed subtyping and sharing relations that can be encoded in the refinement language.
As shown in \autoref{fig:sharing-and-subtyping}, we achieve this by introducing the partial-order $\sqsubseteq_{\Delta}$ and sum $\oplus_\Delta$ operators as follows.

\begin{align*}
  \psi_1 \sqsubseteq_\bbB \psi_2  & \defeq \psi_1 = \psi_2 \\
  \psi_1 \sqsubseteq_\bbN \psi_2 & \defeq \psi_1 \le \psi_2 \\
  \psi_1 \sqsubseteq_\bbU \psi_2 & \defeq \psi_1 = \psi_2 \\
  \psi_1 \sqsubseteq_{\delta_\alpha} \psi_2 & \defeq \psi_1 = \psi_2 \\
  \psi_1 \sqsubseteq_{\sprod{\Delta_1}{ \Delta_2}} \psi_2 & \defeq \projl{\psi_1} \sqsubseteq_{\Delta_1} \projl{\psi_2} \wedge \projr{\psi_1} \sqsubseteq_{\Delta_2} \projr{\psi_2} \\
  \psi_1 \sqsubseteq_{\sarrow{\Delta_1 }{ \Delta_2}} \psi_2 & \defeq \rforall{a}{\Delta_1}{ \psi_1(a) \sqsubseteq_{\Delta_2} \psi_2(a)} \\
  \\
  \psi = \psi_1 \oplus_{\bbN} \psi_2 & \defeq \psi = \psi_1 + \psi_2 \\
  \psi = \psi_1 \oplus_{\sprod{\Delta_1}{\Delta_2}} \psi_2 & \defeq (\projl{\psi} = \projl{\psi_1} \oplus_{\Delta_1} \projl{\psi_2}) \wedge
  (\projr{\psi} = \projr{\psi_1} \oplus_{\Delta_2} \projr{\psi_2}) \\
  \psi = \psi_1 \oplus_{\sarrow{\Delta_1}{\Delta_2}} \psi_2 & \defeq \rforall{a}{\Delta_1}{ \psi(a) = \psi_1(a) \oplus_{\Delta_2} \psi_2(a) } \\
  \psi = \psi_1 \oplus_{\_} \psi_2 & \defeq \neg\top 
\end{align*}

\subsection{Context Wellformedness: $\jwfctxt{\Gamma}$}

A context $\Gamma$ is said to be \emph{wellformed} if every binding in $\Gamma$ is wellformed under a ``prefix'' context before it.
Recall that the context is a sequence of variable bindings, type variables, path conditions, and free potentials.
\autoref{fig:wfctxt} shows these rules.

\begin{figure}[h]
\begin{mathpar}\footnotesize
	\Rule{Wf-Empty}
	{ }
	{ \jwfctxt{\cdot} }
	\and
	\Rule{Wf-Bind-Type}
	{ \jwfctxt{\Gamma} \\ \jwftype{\Gamma}{S} }
	{ \jwfctxt{\Gamma,\bindvar{x}{ S}} }
	\and
	\Rule{Wf-Bind-Sort}
	{ \jwfctxt{\Gamma}  }
	{ \jwfctxt{\Gamma,a:\Delta} }
	\and
	\Rule{Wf-Bind-Cond}
	{ \jwfctxt{\Gamma} \\ \jsort{\Gamma}{\psi}{\bbB} }
	{ \jwfctxt{\Gamma,\psi} }
	\and
	\Rule{Wf-Bind-TVar}
	{ \jwfctxt{\Gamma} }
	{ \jwfctxt{\Gamma,\alpha} }
	\and
	\Rule{Wf-Bind-Pot}
	{ \jwfctxt{\Gamma} \\ \jsort{\Gamma}{\phi}{\bbN} }
	{ \jwfctxt{\Gamma,\phi} }
\end{mathpar}
\caption{Context wellformedness rules}
\label{fig:wfctxt}
\end{figure}

\subsection{Context Sharing: $\jctxsharing{\Gamma}{\Gamma_1}{\Gamma_2}$}

We have already presented type sharing rules.
To apportion the associated potential of $\Gamma$ properly to two contexts $\Gamma_1,\Gamma_2$ with the same sequence of bindings, we introduce \emph{context sharing} relations.
The rules are summarized in \autoref{fig:ctxsharing}.

\begin{figure}[h]
\begin{mathpar}\footnotesize
	\Rule{Share-Empty}
	{ }
	{ \jctxsharing{\cdot}{\cdot}{\cdot} }
	\and
	\Rule{Share-Bind-Type}
	{ \jctxsharing{\Gamma}{\Gamma_1}{\Gamma_2} \\ \jsharing{\Gamma}{S}{S_1}{S_2} }
	{ \jctxsharing{\Gamma,\bindvar{x}{ S}}{\Gamma_1,\bindvar{x}{ S_1}}{\Gamma_2,\bindvar{x}{ S_2}} }
	\and
	\Rule{Share-Bind-Sort}
	{ \jctxsharing{\Gamma}{\Gamma_1}{\Gamma_2} }
	{ \jctxsharing{\Gamma,a:\Delta}{\Gamma_1,a:\Delta}{\Gamma_2,a:\Delta} }
	\and
	\Rule{Share-Bind-Cond}
	{ \jctxsharing{\Gamma}{\Gamma_1}{\Gamma_2} \\ \jsort{\Gamma}{\psi}{\bbB} }
	{ \jctxsharing{\Gamma,\psi}{\Gamma_1,\psi}{\Gamma_2,\psi} }
	\and
	\Rule{Share-Bind-TVar}
	{ \jctxsharing{\Gamma}{\Gamma_1}{\Gamma_2} }
	{ \jctxsharing{\Gamma,\alpha}{\Gamma_1,\alpha}{\Gamma_2,\alpha} }
	\and
	\Rule{Share-Bind-Pot}
	{ \jctxsharing{\Gamma}{\Gamma_1}{\Gamma_2} \\ \jprop{\Gamma}{\phi=\phi_1+\phi_2} }
	{ \jctxsharing{\Gamma,\phi}{\Gamma_1,\phi_1}{\Gamma_2,\phi_2} }
\end{mathpar}
\caption{Context sharing rules}
\label{fig:ctxsharing}
\end{figure}

\subsection{Total Free Potential: $\pot{\Gamma}$}

The \emph{free potentials} of a context $\Gamma$, written $\pot{\Gamma}$, include all the potential bindings, as well as outermost annotated potentials of variable bindings.
\begin{align*}
	\pot{\cdot} & =  0  & \pot{\Gamma,\alpha} & =  \pot{\Gamma} \\
	\pot{\Gamma,x: \tpot{\tsubset{B}{\psi}}{\phi}} & =  \pot{\Gamma} + \subst{x}{\nu}{\phi} \enskip &  \pot{\Gamma,\psi} & = \pot{\Gamma} \\
	\pot{\Gamma,x: \tpot{\p{\tarrowm{y}{T_y}{T}{m}}}{\phi}} & = \pot{\Gamma} + \phi & \pot{\Gamma,\phi} & = \pot{\Gamma} + \phi \\
	\pot{\Gamma,x: \forall\alpha.S} & = \pot{\Gamma}  \Omit{& \pot{\Gamma,x: \tprod{T_l}{T_r}}  & =  \pot{\Gamma} } & \pot{\Gamma,a:\Delta} & = \pot{\Gamma}
\end{align*}

\subsection{Type Substitution: $\subst{\trefined{B}{\psi}{\phi}}{\alpha}{S}$}

In \typesys, type substitution is restricted to resource-annotated subset types.
The substitution $\subst{\trefined{B}{\psi}{\phi}}{\alpha}{S}$ should take care of logical refinements and potential annotations from both $S$ and $\trefined{B}{\psi}{\phi}$.
Following gives the definition.
\begin{align}
    \subst{U}{\alpha}{\tunit} & = \tunit \\
    \subst{U}{\alpha}{\tnat} & = \tnat \\
	\subst{U}{\alpha}{\tbool} & =  \tbool \\
	\subst{U}{\alpha}{(\tprod{B_1}{B_2})} & = \trefined{\tprod{B_1'}{B_2'}}{ [\projl{\nu}/\nu]\psi_1 \wedge [\projr{\nu}/\nu]\psi_2 }{ [\projl{\nu}/\nu]\phi_1 + [\projr{\nu}/\nu]\phi_2 } \\
	& \text{where}~\subst{U}{\alpha}{B_1} = \trefined{B_1'}{\psi_1}{\phi_1} , \subst{U}{\alpha}{B_2} = \trefined{B_2'}{\psi_2}{\phi_2} \\
	\subst{U}{\alpha}{\tinduct{C}{T}{m}{\theta}} & =  \tinduct{C}{\subst{U}{\alpha}{T}}{m}{\theta} \\
	\subst{U}{\alpha}{m \cdot \beta} & =  m \cdot \beta \\
	\subst{\tpot{\tsubset{B}{\psi}}{\phi}}{\alpha}{m \cdot \alpha} & =   \tpot{\tsubset{m \times B}{\psi}}{ m \times \phi} \\
	\subst{U}{\alpha}{\tsubset{B}{\psi}} & =  \tpot{\tsubset{B'}{\psi \wedge \psi'}}{\phi'}\\
	& \text{where}~\subst{U}{\alpha}{B} =\tpot{\tsubset{B'}{\psi'}}{\phi'} \\
	\subst{U}{\alpha}{\tarrowm{x}{T_x}{T}{m}} & =  \tarrowm{x}{\subst{U}{\alpha}{T_x}}{\subst{U}{\alpha}{T}}{m} \\
	\subst{U}{\alpha}{\tpot{R}{\phi}} & =  \tpot{R'}{\phi + \phi'}\\
	& \text{where}~ \subst{U}{\alpha}{R} = R'^{\phi'} \\
	\subst{U}{\alpha}{\forall\beta.S} & =  \forall\beta. \subst{U}{\alpha}{S}
\end{align}

Type multiplication is defined as follows.
\begin{align}
	m \times \tbool & =  \tbool \\
	m \times \tnat & = \tnat \\
	m \times \tunit & = \tunit \\
	m \times (\tprod{B_1}{B_2 }) & = \tprod{(m \times B_1)}{(m \times B_2)} \\
	m \times \tinduct{C}{T}{k}{\theta} & =  \tinduct{C}{m \times T}{k}{m \times \theta} \\
	m_1 \times (m_2 \cdot \alpha) & =  (m_1 \cdot m_2) \cdot \alpha \\
	m \times \tsubset{B}{\psi} & = \tsubset{m \times B}{\psi} \\
	m_1 \times (\tarrowm{x}{T_x}{T}{m_2}) & = \tarrowm{x}{T_x}{T}{(m_1 \cdot m_2)} \\
	m \times \tpot{R}{\phi} & = \tpot{(m \times R)}{m \times \phi}
\end{align}

\subsection{Validity Checking}
\label{sec:appendixvalidity}

In this section, we define the \emph{validity checking} judgment $\jprop{\Gamma}{\psi}$ where $\Gamma$ is a wellformed context and $\psi$ is a Boolean-sorted refinement, following the approach of $\mathrm{Re}^2$~\cite{resyn}.
Intuitively, the judgment states that the formula $\psi$ is always true under any instance of $\Gamma$.
Our approach is to define a set-based denotational semantics for refinements and then reduce the validity checking to Presburger arithmetic. 

\paragraph{Semantics of Sorts}
A sort $\Delta$ represents a set $\interps{\Delta}$ of $\Delta$-sorted refinements.
The following gives the definition of $\interps{\Delta}$.
Note that we only define the semantics for sorts that do \emph{not} contain uninterpreted sorts.
We denote such sorts by $\Delta_o$.
\begin{align}
	\interps{\bbB} & =  \{ \top, \bot \} \\
	\interps{\bbN} & =  \bbZ^+_0 \\ 
	\interps{\bbU} & = \{ \star \} \\
	\interps{\sprod{\Delta_1}{\Delta_2}} & = \{ \rpair{\psi_1}{\psi_2} : \psi_1 \in \interps{\Delta_1} \wedge \psi_2 \in \interps{\Delta_2} \} \\
	\interps{\sarrow{\Delta_1}{\Delta_2}} & = \interps{\Delta_1} \to \interps{\Delta_2}
\end{align}

\paragraph{Semantics of Types}
As we have already done in the sorting rules, scalar types are reflected in the refinement level.
To interpret a wellformed scalar type as a sort without uninterpreted sorts, we define a transformation $\calT_E(\cdot)$ from types to sorts, parametrized by an \emph{environment} that resolves uninterpreted sorts $\delta_\alpha$.
\begin{align}
    \calT_E(\tunit) & = \bbU \\
	\calT_E(\tbool) & = \bbB \\
	\calT_E(\tnat) & = \bbN \\
	\calT_E(\tprod{B_1}{B_2}) & = \sprod{\calT_E(B_1)}{\calT_E(B_2)} \\
	\calT_E(D) & = \Delta_D ~\text{where}~D=\tinduct{C}{T}{m}{\theta}\\
	\calT_E(m \cdot \alpha) & = E(\delta_\alpha) 
\end{align}

\paragraph{Semantics of Contexts}
To give a meaning to a context $\Gamma$, we need to assign an instance for each variable binding with a scalar type, as well as type variables.
Intuitively, a context $\Gamma$ represents a set of \emph{environments} that resolves both program variables and uninterpreted sorts.
Making use of semantics for sorts and types defined above, we can define $\interps{\Gamma}$ inductively as follows.
\begin{align}
	\interps{\cdot} & = \{ \emptyset \} \\
	\interps{\Gamma,\bindvar{x}{ \tpot{\tsubset{B}{\psi}}{\phi}}} & =  \{ E[x \mapsto \psi] : E \in \interps{\Gamma} \wedge \psi \in \interps{\calT_E(B)}  \} \\
	\interps{\Gamma,a:\Delta} & = \{ E[a \mapsto \psi] : E \in \interps{\Gamma} \wedge \psi \in \interps{\Delta} \} \\
	\interps{\Gamma,\bindvar{x}{ \tpot{(\tarrowm{y}{T_y}{T}{m})}{\phi}}} & =  \interps{\Gamma} \\
	\interps{\Gamma,\bindvar{x}{ \forall\alpha.S}} & = \interps{\Gamma} \\
	\interps{\Gamma,\alpha} & = \{ E[\delta_\alpha \mapsto \Delta] \mid E \in \interps{\Gamma} \wedge \Delta \in \Delta_o  \}   \\
	\interps{\Gamma,\psi} & = \interps{\Gamma} \\
	\interps{\Gamma,\phi} &=  \interps{\Gamma}
\end{align}

\paragraph{Semantics of Refinements}
The meaning of a refinement $\psi$ is defined with respect to its sorting judgment $\jsort{\Gamma}{\psi}{\Delta}$.
The following defines an \emph{evaluation} map $\interp{\psi} : \interps{\Gamma} \to \interps{\Delta}$, by induction on the derivation of the sorting judgment, or essentially structural induction on $\psi$.
\begin{align}
	\interp{x}(E) & = E(x) \\
	\interp{a}(E) & = E(a) \\
	\interp{n}(E) & = n \\
	\interp{\star}(E) & = \star \\
	\interp{\top}(E) & = \top \\
	\interp{\neg\psi}(E) & = \neg\interp{\psi}(E) \\
	\interp{\psi_1 \wedge \psi_2}(E) & = \interp{\psi_1}(E) \wedge \interp{\psi_2}(E) \\
	\interp{n}(E) & = n \\
	\interp{\phi_1 \le \phi_2}(E) & = \interp{\phi_1}(E) \le \interp{\phi_2}(E) \\
	\interp{\phi_1 + \phi_2}(E) & =\interp{\phi_1}(E) + \interp{\phi_2}(E) \\
	\interp{\psi_1 = \psi_2}(E) & = \interp{\psi_1}(E) = \interp{\psi_2}(E) \\
	\interp{\rforall{a}{\Delta}{\psi}}(E) & = \forall \phi. \phi \in \interps{\Delta} \implies \interp{\psi}(E[a \mapsto \phi]) \\
	\interp{\rabs{a}{\Delta}{\psi}}(E) & = f~\text{where}~ f(\phi) \defeq \interp{\psi}(E[a \mapsto \phi]) \\
	\interp{\psi_1~\psi_2}(E) & = \interp{\psi_1}(E) ( \interp{\psi_2}(E)) \\
	\interp{\rpair{\psi_1}{\psi_2}}(E) & = (\interp{\psi_1}(E),\interp{\psi_2}(E)) \\
	\interp{\projl{\psi}}(E) & = \mathbf{let}~\rpair{\psi_l}{\psi_r} = \interp{\psi}(E)~\mathbf{in}~\psi_l \\
	\interp{\projr{\psi}}(E) & = \mathbf{let}~\rpair{\psi_l}{\psi_r} = \interp{\psi}(E)~\mathbf{in}~\psi_r
\end{align}

\paragraph{Validity Checking}
Now we show how to assign meanings to contexts and refinements, then the last step to define $\jprop{\Gamma}{\psi}$ is to collect all the refinement constraints mentioned in $\Gamma$.

We first define how to extract constraints from a type binding.
Note that only scalar types (i.e., subset types) can carry logical refinements.
\begin{align}
	\scrB_\Gamma(\bindvar{x}{\trefined{B}{\psi}{\phi}}) & = \subst{x}{\nu}{\psi} \\
	\scrB_\Gamma(\bindvar{x}{\tpot{(\tarrowm{y}{T_y}{T}{m})}{\phi}}) & = \top \\
	\scrB_\Gamma(\bindvar{x}{\forall\alpha.S}) & = \top
\end{align}

Then we define $\scrB(\Gamma)$ to collect all the constraints from variable bindings and path conditions in $\Gamma$.
It is defined inductively on $\Gamma$.
\begin{align}
	\scrB(\cdot) & = \top \\
	\scrB(\Gamma,\bindvar{x}{S}) & = \scrB(\Gamma) \wedge \scrB_\Gamma(\bindvar{x}{S}) \\
	\scrB(\Gamma,x:\Delta) & = \scrB(\Gamma) \\
	\scrB(\Gamma,\bindvar{x}{\tpot{(\tarrowm{y}{T_y}{T}{m})}{\phi}}) & = \scrB(\Gamma) \\
	\scrB(\Gamma,\alpha) & = \scrB(\Gamma) \\
	\scrB(\Gamma,\psi) & = \scrB(\Gamma) \wedge \psi \\
	\scrB(\Gamma,\phi) & = \scrB(\Gamma)
\end{align}

Now we can define the validity checking judgment $\jprop{\Gamma}{\psi}$.
\[
\jprop{\Gamma}{\psi} \defeq \forall E \in \interps{\Gamma}\!:  \interp{\scrB(\Gamma) \implies \psi}(E)
\]
Further, we can embed our denotational semantics for refinements in Presburger arithmetic, so we can also write the validity checking as the following formula
\[
\forall E \in \interps{\Gamma}\!: E \models \scrB(\Gamma) \implies \psi,
\]
where $\models$ is interpreted in Presburger arithmetic.

\subsection{Definition of Consistency}

To describe soundness of our type system, we will need a notion of \emph{consistency}.
Basically, given a typing judgment $\jstyping{\Gamma}{v}{S}$ of a value, we want to know that under the context $\Gamma$, $v$ satisfies the logical conditions indicated by $S$, as well as $\Gamma$ has sufficient amount of potential to be stored in $v$ with respect to $S$.


We use $\calI(\cdot)$ to transform a \emph{value stack} $V$ to a \emph{refinement environment} $E$ with respect to a context $\Gamma$.
The stack $V$ maps type variables to concrete types, program variables to values, and index variables to refinements.
The environment $E$ is used to define validity checking in former sections.
The following defines the transformation $\calI_V(\Gamma)$ by induction on $\Gamma$.
\begin{align}
	\calI_V(\cdot) & = \emptyset \\
	\calI_V(\Gamma,\bindvar{x}{\tpot{\tsubset{B}{\psi}}{\phi}}) & = \calI_V(\Gamma)[x \mapsto \calI(V(x))] \\
	\calI_V(\Gamma,a:\Delta) & = \calI_V(\Gamma)[a \mapsto V(a)] \\
	\calI_V(\Gamma,\bindvar{x}{\tpot{(\tarrowm{y}{T_y}{T}{m})}{\phi}}) & = \calI_V(\Gamma) \\
	\calI_V(\Gamma,\bindvar{x}{\forall\alpha.S}) & = \calI_V(\Gamma) \\
	\calI_V(\Gamma,\alpha) & = \mathbf{let}~E=\calI_V(\Gamma)~\mathbf{in} \\
	& \quad  E[\delta_\alpha \mapsto  \calT_E(V(\alpha))] \\
	\calI_V(\Gamma,\psi) & =\calI_V(\Gamma) \\
	\calI_V(\Gamma,\phi) & = \calI_V(\Gamma)
\end{align}

Now we define how to extract constraints from a value with respect to its type.
It is similar to how we extract constraints from a typing binding in the refinement level.
The differences are that (i) we need to use the interpretation $\calI(\cdot)$ to map values to refinements, (ii) we need to take care of list elements and pair components, (iii) we need to substitute type variables with concrete types, and (iv) for polymorphic type schemas, we assert that the constraints hold for all instantiations.
\begin{align}
	\condv{V}{b}{\trefined{\tbool}{\psi}{\phi}} & = \subst{\calI(b)}{\nu}{\psi} \\
	\condv{V}{u}{\trefined{\tunit}{\psi}{\phi} } & = \subst{\calI(u)}{\nu}{\psi} \\
	\condv{V}{n}{\trefined{\tnat}{\psi}{\phi}} & = \subst{\calI(n)}{\nu}{\psi} \\
	\condv{V}{\epair{v_1}{v_2}}{\trefined{\tprod{B_1}{B_2}}{\psi}{\phi}} & = \subst{\calI(\epair{v_1}{v_2})}{\nu}{\psi} \\
	\condv{V}{C_j(v_0,\tuple{v_1,\cdots,v_{m_j}})}{ \trefined{\tinduct{C}{T}{m}{\theta}}{\psi}{\phi} } & = \subst{\calI(C_j(v_0,\tuple{v_1,\cdots,v_{m_j}}))}{\nu}{\psi} \wedge \condv{V}{v_0}{T_j}  \\
	& \wedge \bigwedge_{i=1}^{m_j} \condv{V}{v_j}{ \tinduct{C}{T}{m}{\lhd.\mathbf{j}(\calI(v_0))(\theta).\mathbf{i}} } \\
	\condv{V}{v}{\trefined{m \cdot \alpha}{\psi}{\phi}} & = \condv{V}{v}{ \subst{V(\alpha)}{\alpha}{\tsubset{m \cdot \alpha}{\psi}} } \\
 	\condv{V}{v}{\tpot{(\tarrowm{x}{T_x}{T}{m})}{\phi}} & = \top \\
 	\condv{V}{v}{\forall\alpha.S} & = \forall \trefined{B}{\psi}{\phi} \!:  \condv{V'}{v}{S} \\
 	& \text{where}~\jwftype{\Gamma}{\tpot{\tsubset{B}{\psi}}{\phi}} \\
 	& \text{and} ~V'=V[\alpha \mapsto \tpot{\tsubset{B}{\psi}}{\phi}]
\end{align}

The following defines how to collect path conditions of a stack $V$ with respect to its typing context $\Gamma$, written $\condc{V}{\Gamma}$.
\begin{align}
	\condc{V}{\cdot} & = \top \\
	\condc{V}{\Gamma,\bindvar{x }{ \tpot{\tsubset{B}{\psi}}{\phi} }} & = \condc{V}{\Gamma} \wedge \condv{V}{V(x)}{\trefined{B}{\psi}{\phi}} \\
	\condc{V}{\Gamma,a:\Delta} & = \condc{V}{\Gamma} \\
	\condc{V}{\Gamma,\bindvar{x }{ \tpot{(\tarrowm{y}{T_y}{T}{m})}{\phi}}} & = \condc{V}{\Gamma} \\
	\condc{V}{\Gamma,\bindvar{x}{ \forall\alpha.S}} & = \condc{V}{\Gamma} \\
	\condc{V}{\Gamma,\alpha} & = \condc{V}{\Gamma} \\
	\condc{V}{\Gamma,\psi} & = \condc{V}{\Gamma} \wedge \psi\\
	\condc{V}{\Gamma,\phi} & = \condc{V}{\Gamma}
\end{align}

Similar to logical refinements, we can also collect potential annotations.
The following defines $\potv{V}{v}{S}$ as the potential stored in the value $v$ with respect to the type $S$ under the stack $V$.
\begin{align}
	\potv{V}{b}{\trefined{\tbool}{\psi}{\phi}} & =  \subst{\calI(b)}{\nu}{\phi} \\
	\potv{V}{u}{\trefined{\tunit}{\psi}{\phi}} & = \subst{\calI(u)}{\nu}{\phi} \\
	\potv{V}{n}{\trefined{\tnat}{\psi}{\phi}} & = \subst{\calI(n)}{\nu}{\phi} \\
	\potv{V}{\epair{v_1}{v_2 }}{\trefined{\tprod{B_1}{B_2}}{\psi}{\phi}} & = \subst{\calI(\epair{v_1}{v_2})}{\nu}{\phi} \\
	\potv{V}{C_j(v_0,\tuple{v_1,\cdots,v_{m_j}})}{\trefined{\tinduct{C}{T}{m}{\theta}}{\psi}{\phi}} & = \subst{\calI(C_j(v_0,\tuple{v_1,\cdots,v_{m_j}}))}{\nu}{\phi} + \potv{V}{v_0}{T_j} \\
	& + \pi.\mathbf{j}(\calI(v_0))(\theta) \\
	& + \sum_{i=1}^{m_j} \potv{V}{v_j}{\tinduct{C}{T}{m}{ \lhd.\mathbf{j}(\calI(v_0)).\mathbf{i}} } \\
	\potv{V}{v}{ \trefined{m \cdot \alpha}{\psi}{\phi} } & =  \potv{V}{v }{ \subst{V(\alpha)}{\alpha}{ \tpot{(m \cdot \alpha)}{\phi}} } \\
	\potv{V}{v}{ \tpot{(\tarrowm{x}{T_x}{T}{m})}{\phi}} & =  \phi\\
	\potv{V}{v}{ \forall\alpha.S } & =  0
\end{align}

Also we have a stack version for potentials $\potc{V}{\Gamma}$.
\begin{align}
	\potc{V}{\cdot} & =  0 \\
	\potc{V}{\Gamma,\bindvar{x }{ \tpot{\tsubset{B}{\psi}}{\phi}}} & =  \potc{V}{\Gamma} + \potv{V}{V(x) }{ \trefined{B}{\psi}{\phi} } \\
	\potc{V}{\Gamma,a:\Delta} & = \potc{V}{\Gamma} \\
	\potc{V}{\Gamma,\bindvar{x}{ \tpot{(\tarrowm{y}{T_y}{T}{m})}{\phi}}}& =  \potc{V}{\Gamma}+ \phi  \\
	\potc{V}{\Gamma,\bindvar{x}{\forall\alpha.S}} & =  \potc{V}{\Gamma} \\
	\potc{V}{\Gamma,\alpha} & =  \potc{V}{\Gamma} \\
	\potc{V}{\Gamma,\psi} & = \potc{V}{\Gamma} \\
	\potc{V}{\Gamma,\phi} & =  \potc{V}{\Gamma} + \phi
\end{align}

Finally, we are able to define two notions of consistency for values and stacks, respectively.

\begin{definition}[Value consistency]\label{de:valconsistency}
 A value $\jval{v}$ is said to be \emph{consistent} with $\jstyping{\Gamma}{v}{S}$, if for all $\jctxtyping{V}{\Gamma}$, $E=\calI_V(\Gamma)$ such that $E \models \condc{V}{\Gamma}$, we have $E \models \condv{V}{v}{S} \wedge \potc{V}{\Gamma} \ge \potv{V}{v}{S}$.
\end{definition}

\begin{definition}[Stack consistency]\label{de:envconsistency}
An environment $V'$ is said to be \emph{consistent} with $\jctxtyping[\Gamma]{V'}{\Gamma'}$, if for for all $\jctxtyping{V}{\Gamma}$, $E = \calI_V(\Gamma)$ such that $E \models \condc{V}{\Gamma}$, we have $E' \models \condc{V,V'}{\Gamma'} \wedge \potc{V}{\Gamma} \ge \potc{V,V'}{\Gamma'}$ where $E' \defeq \calI_{V,V'}(\Gamma,\Gamma')$.
\end{definition}

\makeatletter
\tagsleft@true
\makeatother
\newtagform{nobrackets}[\underline]{}{}
\usetagform{nobrackets}

\section{Proofs for Soundness}
\label{sec:appendixproofs}

\subsection{Progress}

\begin{lemma}\label{lem:consistentcons}
%
	Let $\Gamma = \overline{q \mid \alpha}$.
	If $\jctxsharing{\Gamma}{\Gamma_1}{\Gamma_2}$, $v_0$ is consistent with $\jstyping{\Gamma_1}{v_0}{T_j}$, and $\tuple{v_1,\cdots,v_{m_j}}$ is consistent with $\jatyping{\Gamma_2}{  \tuple{v_1,\cdots,v_{m_j}} }{ \prod_{i=1}^{m_j} \tinduct{C}{T}{m}{ \lhd.\mathbf{j}(\calI(v_0))(\theta).\mathbf{i} }}$, then $C_j(v_0,\tuple{v_1,\cdots,v_{m_j}})$ is consistent with $\jstyping{\Gamma,\pi.\mathbf{j}(\calI(v_0))(\theta)}{C_j(v_0,\tuple{v_1,\cdots,v_{m_j}})}{\tsubset{ \tinduct{C}{T}{m}{\theta} }{ \nu = \calI(C_j(v_0,\tuple{v_1,\cdots,v_{m_j}} )) } }$.
\end{lemma}
\begin{proof}
		\begin{alignat}{2}
			& \text{Fix}~\jctxtyping{V}{\Gamma}, E=\calI_V(\Gamma)~\text{s.t.}~E\models \condc{V}{\Gamma} \\
			& \jctxsharing{\Gamma}{\Gamma_1}{\Gamma_2} & \text{[premise]} \\
			& \quad \implies \potc{V}{\Gamma} = \potc{V}{\Gamma_1}+\potc{V}{\Gamma_2} & \label{eq:foldcons:sharing} \\
			&  \jstyping{\Gamma_1}{v_0}{T_j}~\text{consistent} & \text{[premise]} \label{eq:foldcons:vhconsist}\\
			& \jstyping{\Gamma_2}{\tuple{v_1,\cdots,v_{m_j}}}{\tsubset{ \textstyle\prod_{i=1}^{m_j} \tinduct{C}{T}{m}{ \lhd.\mathbf{j}(\calI(v_0))(\theta).\mathbf{i} } }{\nu = \calI(\tuple{\cdots})}} ~\text{consistent} & \text{[premise]} \label{eq:foldcons:vtconsist} \\
			& \jatyping{\Gamma}{C_j(v_0,\tuple{v_1,\cdots,v_{m_j}})}{\tinduct{C}{T}{m}{\theta}} & \text{[typing]} \\
			& \jstyping{\Gamma}{C_j(v_0,\tuple{v_1,\cdots,v_{m_j}})}{\tsubset{\tinduct{C}{T}{m}{\theta}}{\nu = \calI(C_j(v_0,\tuple{\cdots})) }} & \text{[typing]} \\
			& \condv{V}{ C_j(v_0,\tuple{v_1,\cdots,v_{m_j}}) }{\tsubset{\tinduct{C}{T}{m}{\theta}}{\nu = \calI(C_j(v_0,\tuple{\cdots}))}}  \\
			& \quad = \subst{\calI(C_j(v_0,\tuple{\cdots}))}{\nu}{(\nu = \calI(C_j(v_0,\tuple{\cdots}))} \wedge \\
			& \qquad \condv{V}{v_0}{T_j} \wedge \bigwedge_{i=1}^{m_j} \condv{V}{v_i}{\tinduct{C}{T}{m}{\lhd.\mathbf{j}(\calI(v_0))(\theta).\mathbf{i} }} \\
			& \quad = \condv{V}{v_0}{T_j} \wedge \bigwedge_{i=1}^{m_j} \condv{V}{v_i}{\tinduct{C}{T}{m}{\lhd.\mathbf{j}(\calI(v_0))(\theta).\mathbf{i} }} \\
			& \potv{V}{C_j(v_0,\tuple{v_1,\cdots,v_{m_j}})}{\tpot{\tinduct{C}{T}{m}{\theta}}{0}} = 0 +\\
			& \quad \potv{V}{v_0}{T_j} + \pi.\mathbf{j}(\calI(v_0))(\theta) + \sum_{i=1}^{m_j} \potv{V}{v_j}{\tinduct{C}{T}{m}{\lhd.\mathbf{j}(\calI(v_0))(\theta).\mathbf{i}}} \\
			& \quad = \potv{V}{v_0}{T_j} + \pi.\mathbf{j}(\calI(v_0))(\theta) + \sum_{i=1}^{m_j} \potv{V}{v_j}{\tinduct{C}{T}{m}{\lhd.\mathbf{j}(\calI(v_0))(\theta).\mathbf{i}}} \\
			& E \models \condv{V}{v_0}{T_j} \wedge \potc{V}{\Gamma_1} \ge \potv{V}{v_0}{T_j} & \text{[\eqref{eq:foldcons:vhconsist}]} \\ 
			& E \models \bigwedge_{i=1}^{m_j} \condv{V}{v_j}{\tinduct{C}{T}{m}{\lhd.\mathbf{j}(\calI(v_0))(\theta).\mathbf{i}}} \wedge  \\
			& \quad \potc{V}{\Gamma_2} + \pi.\mathbf{j}(\calI(v_0))(\theta) \ge \pi.\mathbf{j}(\calI(v_0))(\theta) +  \sum_{i=1}^{m_j} \potv{V}{v_j}{\tpot{\tinduct{C}{T}{m}{ \lhd.\mathbf{j}(\calI(v_0))(\theta).\mathbf{i} }}{0}} & \text{[\eqref{eq:foldcons:vtconsist}]} \\
			& \text{done} & \text{[\eqref{eq:foldcons:sharing}]}
		\end{alignat}
\end{proof}

\begin{proposition}\label{prop:stepdelta}
	If $\jstep{e}{e'}{p}{p'}$ and $c \ge 0$, then $\jstep{e}{e'}{p+c}{p'+c}$.
\end{proposition}
\begin{proof}
	By induction on $\jstep{e}{e'}{p}{p'}$.
\end{proof}

\begin{proposition}\label{prop:subtyping}
	If $\jval{v}$, $\jstyping{\Gamma}{v}{T_1}$, $\jsubty{\Gamma}{T_1}{T_2}$, $\jctxtyping{V}{\Gamma}$ and $E=\calI_V(\Gamma)$ such that $E \models \condc{V}{\Gamma}$, then $E \models \condv{V}{v}{T_1} \implies (\condv{V}{v}{T_2} \wedge \potv{V}{v}{T_1} \ge \potv{V}{v}{T_2})$.
\end{proposition}
\begin{proof}
	By induction on $\jstyping{\Gamma}{v}{T_1}$, followed by an induction on atomic typing.
	\begin{itemize}
    \item \textsc{(SimpAtom-ConsD)}: Let $v = C_j(v_0,\tuple{v_1,\cdots,v_{m_j}})$ for some $j$ and $B_1 = \tinduct{C}{T}{m}{\theta}$.
    By inversion on $\jsubty{\Gamma}{B_1}{B_2}$, we know that $B_2 = \tinduct{C}{T'}{m}{\theta'}$ for some $\theta',\many{T'}$ satisfying $\jsubty{\Gamma}{\many{T}}{\many{T'}}$.
    By \textsc{(Dtype-Index)}, for all $i$, we know that
    \[
    \jsubty{ \Gamma_{<:} ,y:T_j}{ \tinduct{C}{T}{m}{\lhd_{j}(y)(\theta).\mathbf{i}} }{ \tinduct{C}{T'}{m}{\lhd_{j}(y)(\theta').\mathbf{i}}  }
    \]
    where
    \[
   \Gamma_{<:} \defeq  \Gamma,\theta:\Delta_\theta,\theta':\Delta_\theta,B_1 <: B_2,
    \]
    and $ \jprop{\Gamma_{<:},y:T_j}{ \pi.\mathbf{j}(y)(\theta) \ge \pi.\mathbf{j}(y)(\theta')}$.
    By the substitution lemma, we have $$\jsubty{\Gamma}{ \tinduct{C}{T}{m}{\lhd.\mathbf{j}(\calI(v_0))(\theta).\mathbf{i}} }{ \tinduct{C}{T'}{m}{\lhd.\mathbf{j}(\calI(v_0))(\theta').\mathbf{i}} },$$ and $\jprop{\Gamma}{\pi.\mathbf{j}(\calI(v_0))(\theta) \ge \pi.\mathbf{j}(\calI(v_0))(\theta')}$.
    Thus, by induction hypothesis, for all $i$, we know that $E \models \potv{V}{v_i}{ \tinduct{C}{T}{m}{ \lhd.\mathbf{j}(\calI(v_0))(\theta).\mathbf{i} } } \ge \potv{V}{v_i}{ \tinduct{C}{T'}{m}{\lhd.\mathbf{j}(\calI(v_0))(\theta').\mathbf{i}} }$.
    By an inner induction on $\jstyping{\Gamma_1}{v_0}{T_j}$, we obtain that $E \models \potv{V}{v_0}{T_j} \ge \potv{V}{v_0}{T_j'}$.
    
    By the definition of potential, we have
    \begin{align*}
      \potv{V}{v}{B_1} & = \potv{V}{v_0}{T_j} + \pi.\mathbf{j}(\calI(v_0))(\theta) + \sum_{i=1}^{m_j} \potv{V}{v_i}{ \tinduct{C}{T}{m}{\lhd.\mathbf{j}(\calI(v_0))(\theta).\mathbf{i}} }, \\
      \potv{V}{v}{B_2} & = \potv{V}{v_0}{T_j'} + \pi.\mathbf{j}(\calI(v_0))(\theta') + \sum_{i=1}^{m_j} \potv{V}{v_i}{ \tinduct{C}{T'}{m}{\lhd.\mathbf{j}(\calI(v_0))(\theta').\mathbf{i}} },
    \end{align*}
    and conclude $E \models \potv{V}{v}{B_1} \ge \potv{V}{v}{B_2}$ by the inequalities derived above.
  \end{itemize}
\end{proof}

\begin{proposition}\label{prop:sharingsplit}
	If $\jval{v}$, $\jstyping{\Gamma}{v}{S}$, $\jsharing{\Gamma}{S}{S_1}{S_2}$, $\jctxtyping{V}{\Gamma}$ and $E = \calI_V(\Gamma)$ such that $E \models \condc{V}{\Gamma}$, then $E \models \potv{V}{v}{S} = \potv{V}{v}{S_1} + \potv{V}{v}{S_2}$.
\end{proposition}
\begin{proof}
	By induction on $\jstyping{\Gamma}{v}{T_1}$, followed by an induction on atomic typing.
	\begin{itemize}
    \item \textsc{(SimpAtom-ConsD)}: Let $v = C_j(v_0,\tuple{v_1,\cdots,v_{m_j}})$ for some $j$ and $B = \tinduct{C}{T}{m}{\theta} $.
    By inversion on $\jsharing{\Gamma}{B}{B_1}{B_2}$, we know that $B_1 = \tinduct{C}{T_1}{m}{\theta_1}$, $B_2=\tinduct{C}{T_2}{m}{\theta_2}$ for some $\theta_1,\theta_2,\many{T_1},\many{T_2}$ satisfying $\jsharing{\Gamma}{\many{T}}{\many{T_1}}{\many{T_2}}$.
    By \textsc{(Dtype-Index)}, for all $i$, we know that
    \[
    \jsharing{\Gamma_{\sharing} , y : T_j }{\tinduct{C}{T}{m}{ \lhd.\mathbf{j}(y)(\theta).\mathbf{i}}}{ \tinduct{C}{T_1}{m}{\lhd.\mathbf{j}(y)(\theta_1).\mathbf{i}} }{ \tinduct{C}{T_2}{m}{\lhd.\mathbf{j}(y)(\theta_2).\mathbf{i}} }
    \]
    where
    \[
    \Gamma_{\sharing} \defeq \Gamma,\theta:\Delta_\theta,\theta_1:\Delta_\theta,\theta_2:\Delta_\theta, B \sharing B_1 \mid B_2,
    \]
    and $ \jprop{\Gamma_{\sharing},y:T_j}{ \pi.\mathbf{j}(y)(\theta) = \pi.\mathbf{j}(y)(\theta_1) + \pi.\mathbf{j}(y)(\theta_2)}$.
    By the substitution lemma, we have $\jsharing{\Gamma}{\tinduct{C}{T}{m}{\lhd_{j}(\calI(v_0))(\theta).\mathbf{i}}}{\tinduct{C}{T_1}{m}{\lhd_{j}(\calI(v_0))(\theta_1).\mathbf{i}}}{\tinduct{C}{T_2}{m}{\lhd_{j}(\calI(v_0))(\theta_2).\mathbf{i}}}$ and $\jprop{\Gamma}{ \pi.\mathbf{j}(\calI(v_0))(\theta) = \pi.\mathbf{j}(\calI(v_0))(\theta_1) + \pi.\mathbf{j}(\calI(v_0))(\theta_2)}$.
    Thus, by induction hypothesis, for all $i$, we know that $E \models \potv{V}{v_i}{\tinduct{C}{T}{m}{\lhd_{j}(\calI(v_0))(\theta).\mathbf{i}}} = \potv{V}{v_i}{\tinduct{C}{T_1}{m}{\lhd_{j}(\calI(v_0))(\theta_1).\mathbf{i}}} + \potv{V}{v_i}{\tinduct{C}{T_2}{m}{\lhd_{j}(\calI(v_0))(\theta_2).\mathbf{i}}} $.
    By an inner induction on $\jstyping{\Gamma_1}{v_0}{T_j}$, we obtain that $\potv{V}{v_0}{T_j} = \potv{V}{v_0}{T_{1j}} + \potv{V}{v_0}{T_{2j}}$.
    
    By the definition of potential, we have
    \begin{align*}
      \potv{V}{v}{B} & = \potv{V}{v_0}{T_j} + \pi.\mathbf{j}(\calI(v_0))(\theta) + \sum_{i=1}^{m_j} \potv{V}{v_i}{\tinduct{C}{T}{m}{\lhd_{j}(\calI(v_0))(\theta).\mathbf{i}}}, \\
      \potv{V}{v}{B_1} & = \potv{V}{v_0}{T_{1j}} + \pi.\mathbf{j}(\calI(v_0))(\theta_1) + \sum_{i=1}^{m_j} \potv{V}{v_i}{\tinduct{C}{T_1}{m}{\lhd_{j}(\calI(v_0))(\theta_1).\mathbf{i}}}, \\
      \potv{V}{v}{B_2} & = \potv{V}{v_0}{T_{2j}} + \pi.\mathbf{j}(\calI(v_0))(\theta_2) + \sum_{i=1}^{m_j} \potv{V}{v_i}{\tinduct{C}{T_2}{m}{\lhd_{j}(\calI(v_0))(\theta_2).\mathbf{i}}}, \\
    \end{align*}
    and conclude $E \models \potv{V}{v}{B} = \potv{V}{v}{B_1} + \potv{V}{v}{B_2}$ by the equalities derived above.
  \end{itemize}
\end{proof}

\begin{lemma}\label{lem:progressatom}
  If $\Gamma = \overline{q \mid \alpha}$, $\jatyping{\Gamma}{a}{B}$, $\jctxtyping{V}{\Gamma}$ and $p \ge \potc{V}{\Gamma}$, then $\jval{a}$ and $a$ is consistent with $\jstyping{\Gamma}{a}{\tsubset{B}{\nu = \calI(a)}}$.
\end{lemma}
\begin{proof}
  By induction on $\jatyping{\Gamma}{a}{B}$:
  \begin{alignat}{2}\footnotesize
    \shortintertext{\bf\textsc{(SimpAtom-True)}}
    & \text{SPS}~a=\etrue,B=\tbool \\
    & \jval{\etrue} & \text{[value]} \\
    & \condv{V}{\etrue}{\tsubset{\tbool}{\nu = \calI(\etrue)}} \\
    & \quad = \subst{\calI(\etrue)}{\nu}{(\nu = \calI(\etrue))} = \top \\
    & \potv{V}{\etrue}{\tpot{\tbool}{0}} = 0 \le \potc{V}{\Gamma}
		\shortintertext{\bf\textsc{(SimpAtom-False)}}
		& \text{SPS}~a=\efalse, B=\tbool \\
		& \jval{\efalse} & \text{[value]} \\
		& \condv{V}{\efalse}{\tsubset{\tbool}{\nu = \calI(\efalse)}} \\
		& \quad = \subst{\calI(\efalse)}{\nu}{(\nu = \calI(\efalse))} = \top \\
		& \potv{V}{\efalse}{\tpot{\tbool}{0}} = 0 \le \potc{V}{\Gamma}
	\shortintertext{\bf\textsc{(SimpAtom-Pair)}}
	& \text{SPS}~a=\epair{a_1}{a_2}, B=\tprod{B_1}{B_2} \\
	& \jctxsharing{\Gamma}{\Gamma_1}{\Gamma_2} & \text{[premise]} \label{eq:proga:pairshare} \\
	& \jatyping{\Gamma_1}{a_1}{B_1} & \text{[premise]} \label{eq:proga:pairprem1} \\
	& \jatyping{\Gamma_2}{a_2}{B_2} & \text{[premise]} \label{eq:proga:pairprem2} \\
	& \jval{a_1}, a_1~\text{consistent} & \text{[ind. hyp., \eqref{eq:proga:pairprem1}]} \label{eq:proga:pairval1} \\
	& \jval{a_2}, a_2~\text{consistent} & \text{[ind. hyp., \eqref{eq:proga:pairprem2}]} \label{eq:proga:pairval2} \\
	& \jval{\epair{a_1}{a_2}} & \text{[value]} \\
	& \condv{V}{\epair{a_1}{a_2}}{ \tsubset{ \tprod{B_1}{B_2}}{ \nu = \calI(\epair{a_1}{a_2}}} \\
	& \quad = \subst{\calI(\epair{a_1}{a_2})}{\nu}{(\nu = \calI(\epair{a_1}{a_2}))} = \top \\
	& \potv{V}{\epair{a_1}{a_2}}{\tpot{(\tprod{B_1}{B_2})}{0}}  \\
	& \quad = \potv{V}{a_1}{B_1} +\potv{V}{a_2}{B_2} \le \potc{V}{\Gamma_1} + \potc{V}{\Gamma_2} = \potc{V}{\Gamma} & \text{[\eqref{eq:proga:pairshare},\eqref{eq:proga:pairval1},\eqref{eq:proga:pairval2}]}
	\shortintertext{\bf\textsc{(SimpAtom-ConsD)}}
	& \text{SPS}~a=C_j(\hat{a}_0,\tuple{a_1,\cdots,a_{m_j}}), B = \tinduct{C}{T}{m}{\theta} \\
	& \Gamma~\text{contains no variables} \implies \jval{a_0} \label{eq:proga:conshead} \\
	& \Gamma=\Gamma',\pi.\mathbf{j}(\calI(a_0))(\theta), \jctxsharing{\Gamma'}{\Gamma_1}{\Gamma_2} & \text{[premise]} \label{eq:proga:consshare} \\
	& \jstyping{\Gamma_1}{a_0}{T_j} & \text{[premise]} \label{eq:proga:conspremh} \\
	& \jatyping{\Gamma_2}{\tuple{a_1,\cdots,a_{m_j}}}{ \textstyle\prod_{i=1}^{m_j} \tinduct{C}{T}{m}{\lhd.\mathbf{j}(\calI(a_0))(\theta).\mathbf{i} }  } & \text{[premise]} \label{eq:proga:conspremt} \\
	& a_0~\text{consistent} & \text{[\theoref{progress}, \eqref{eq:proga:conshead}, \eqref{eq:proga:conspremh}]} \\
	& \forall j :  \jval{a_j}, a_j~\text{consistent} & \text{[ind. hyp., \eqref{eq:proga:conspremt}]} \\
	& \jval{C_j(a_0,\tuple{a_1,\cdots,a_{m_j}})} & \text{[value]} \\
	& C_j(a_0,\tuple{a_1,\cdots,a_{m_j}})~\text{consistent} & \text{[\lemref{consistentcons}]}
  \end{alignat}  
\end{proof}

\begin{theorem}[Progress]\label{the:progress}
	If  $\Gamma=\overline{q \mid \alpha}$, $\jstyping{\Gamma}{e}{S}$, $\jctxtyping{V}{\Gamma}$ and $p \ge \potc{V}{\Gamma}$, then either $\jval{e}$ and $e$ is consistent with $\jstyping{\Gamma}{e}{S}$, or there exist $e'$ and $p'$ such that $\jstep{e}{e'}{p}{p'}$.
\end{theorem}
\begin{proof}
	By induction on $\jstyping{\Gamma}{e}{S}$:
	\begin{alignat}{2}\footnotesize
	   \shortintertext{\bf\textsc{(T-SimpAtom)}}
	   & \text{SPS}~e=a, S = \tsubset{B}{\nu = \calI(a)} \\
	   & \jval{a}, a~\text{consistent} & \text{[\lemref{progressatom}]}
	   %
		%
		\shortintertext{\bf\textsc{(T-Imp)}}
		& \text{SPS}~e=\eimp,S=T \\
		& \jprop{\Gamma}{\bot} & \text{[premise]} \\
		&  \top \implies \bot\\
		& \text{exfalso} 
		\shortintertext{\bf\textsc{(T-Consume-P)}}
		& \text{SPS}~\Gamma=(\Gamma',c), e=\econsume{c}{e_0}, c \ge 0 \\
		& p \ge \potc{V}{\Gamma}= \potc{V}{\Gamma'} + c \ge c \\
		& \jstep{e}{e_0}{p}{p-c} & \text{[eval.]}
		\shortintertext{\bf\textsc{(T-Consume-N)}}
		& \text{SPS}~e=\econsume{c}{e_0},c<0 \\
		& \jstep{e}{e_0}{p}{p-c} & \text{[eval.]}
		\shortintertext{\bf\textsc{(T-Cond)}}
		& \text{SPS}~e=\econd{a_0}{e_1}{e_2},S=T \\
        & \jatyping{\Gamma}{a_0}{\tbool} & \text{[premise]} \label{eq:prog:conde0} \\
        & \jval{a_0} & \text{[\lemref{progressatom}]} \label{eq:prog:conde0val} \\
        & \text{inv. on \eqref{eq:prog:conde0} with \eqref{eq:prog:conde0val}} \\
        & \textbf{case}~a_0=\etrue \\
        & \enskip \jstep{e}{e_1}{p}{p} & \text{[eval.]} \\
        & \textbf{case}~a_0=\efalse \\
        & \enskip \jstep{e}{e_2}{p}{p} & \text{[eval.]}
        \shortintertext{\bf\textsc{(T-MatP)}}
		& \text{SPS}~e=\ematp{a_0}{x_1}{x_2}{e_1}, S=T \\
		& \jctxsharing{\Gamma}{\Gamma_1}{\Gamma_2} & \text{[premise]} \\
		& \jatyping{\Gamma_1}{a_0}{\tprod{B_1}{B_2}} & \text{[premise]} \label{eq:prog:matpe0} \\
		& \jval{a_0} & \text{[\lemref{progressatom}]} \label{eq:prog:matpe0val} \\
		& \text{inv. on \eqref{eq:prog:matpe0} with \eqref{eq:prog:matpe0val}} \\
		& a_0 = \epair{v_1}{v_2}, \jatyping{\Gamma_{11}}{v_1}{B_1},\jatyping{\Gamma_{12}}{v_2}{B_2}, \jctxsharing{\Gamma_1}{\Gamma_{11}}{\Gamma_{12}} \\
		& \jstep{e}{\subst{v_1,v_2}{x_1,x_2}{e_1}}{p}{p} & \text{[eval.]}
		\shortintertext{\bf\textsc{(T-MatD)}}
		& \text{SPS}~e=\mathsf{matd}(a_0,\many{C_j(x_0,\tuple{x_1,\cdots,x_{m_j}}).e_j}),S=T \\
		& \jctxsharing{\Gamma}{\Gamma_1}{\Gamma_2}  & \text{[premise]} \\
		& \jatyping{\Gamma_1}{a_0}{\tinduct{C}{T}{m}{\theta}} & \text{[premise]} \label{eq:prog:matle0} \\
		& \jval{a_0} & \text{[\lemref{progressatom}]} \label{eq:prog:matle0val} \\
        & \text{inv. on \eqref{eq:prog:matle0} with \eqref{eq:prog:matle0val}} \\
        & \textbf{case}~a_0=C_j(v_0,\tuple{v_1,\cdots,v_{m_j}}) \\
        & \enskip \jstep{e}{ \subst{v_0,v_1,\cdots,v_{m_j}}{x_0,x_1,\cdots,x_{m_j}}{e_j} }{p}{p} & \text{[eval.]} \\
		\shortintertext{\bf\textsc{(T-Let)}}
		& \text{SPS}~e=\elet{e_1}{x}{e_2},S=T_2 \\
		& \jctxsharing{\Gamma}{\Gamma_1}{\Gamma_2}  & \text{[premise]}  \\
		& \quad \implies \potc{V}{\Gamma}=\potc{V}{\Gamma_1}+\potc{V}{\Gamma_2} \label{eq:prog:letsplit} \\
		& \jstyping{\Gamma_1}{e_1}{S_1} & \text{[premise]} \label{eq:prog:lete1} \\
		& p \ge \potc{V}{\Gamma_1} & \text{[asm., \eqref{eq:prog:letsplit}]} \label{eq:prog:letindhyp} \\
		& \text{ind. hyp. on \eqref{eq:prog:lete1} with \eqref{eq:prog:letindhyp}} \\
		& \textbf{case}~\jstep{e_1}{e_1'}{p}{p'}  \\
		& \enskip \jstep{e}{\elet{e_1'}{x}{e_2}}{p}{p'} & \text{[eval.]} \\
		& \textbf{case}~\jval{e_1}  \\
		& \enskip \jstep{e}{\subst{e_1}{x}{e_2}}{p}{p} & \text{[eval.]}
		\shortintertext{\bf\textsc{(T-App)}}
		& \text{SPS}~e=\eapp{\hat{a}_1}{\hat{a}_2}, S=T \\
		& \jctxsharing{\Gamma}{\Gamma_1}{\Gamma_2} & \text{[premise]} \\
		& \jstyping{\Gamma_1}{\hat{a}_1}{\tarrowm{x}{T_x}{T}{1}} & \text{[premise]} \label{eq:prog:appe1} \\
		& \jstyping{\Gamma_2}{\hat{a}_2}{T_x} & \text{[premise]} \label{eq:prog:appe2} \\
		&\Gamma~\text{contains no variables} \\
		& \quad \implies \jval{\hat{a}_1,\hat{a}_2} \label{eq:prog:appevals} \\
        & \text{inv. on \eqref{eq:prog:appe1} with \eqref{eq:prog:appevals}} \\
        & \textbf{case}~e_1=\eabs{x}{e_0} \\
        & \enskip \jstep{e}{\subst{\hat{a}_2}{x}{e_0}}{p}{p} & \text{[eval.]} \\
        & \textbf{case}~e_2 =\efix{f}{x}{e_0} \\
        & \enskip \jstep{e}{\subst{\efix{f}{x}{e_0},\hat{a}_2}{f,x}{e_0}}{p}{p} & \text{[eval.]}
        \shortintertext{\bf\textsc{(T-App-SimpAtom)}}
        & \text{SPS}~e=\eapp{\hat{a}_1}{a_2},S=\subst{\calI(a_2)}{x}{T} \\
        & \jctxsharing{\Gamma}{\Gamma_1}{\Gamma_2} & \text{[premise]} \\
        & \jstyping{\Gamma_1}{\hat{a}_1}{\tarrowm{x}{\trefined{B}{\psi}{\phi}}{T}{1}} & \text{[premise]} \label{eq:prog:appatome1} \\
        & \Gamma~\text{contains no variables} \\
        & \quad \implies \jval{\hat{a}_1} \label{eq:prog:appatome1val} \\
        & \jstyping{\Gamma_2}{\hat{a}_2}{\trefined{B}{\psi}{\phi}} & \text{[premise]} \\
        & \jval{a_2} & \text{[\lemref{progressatom}]} \label{eq:prog:appatome2val} \\
        & \text{inv. on \eqref{eq:prog:appatome1} with \eqref{eq:prog:appatome1val}} \\
        & \textbf{case}~e_1=\eabs{x}{e_0} \\
        & \enskip \jstep{e}{\subst{a_2}{x}{e_0}}{p}{p} & \text{[eval.]} \\
        & \textbf{case}~e_2 =\efix{f}{x}{e_0} \\
        & \enskip \jstep{e}{\subst{\efix{f}{x}{e_0},a_2}{f,x}{e_0}}{p}{p} & \text{[eval.]}
		\shortintertext{\bf\textsc{(T-Abs)}}
		& \text{SPS}~e =\eabs{x}{e_0},S = \tarrow{x}{T_x}{T} \\
		& \jval{\eabs{x}{e_0}} & \text{[value]} \\
		& \condv{V}{\eabs{x}{e_0}}{\tarrow{x}{T_x}{T}} = \top \\
		& \potv{V}{\eabs{x}{e_0}}{\tpot{(\tarrow{x}{T_x}{T})}{0}} = 0 \le \potc{V}{\Gamma}
		\shortintertext{\bf\textsc{(T-Abs-Lin)}}
		& \text{SPS}~\Gamma=m\cdot \Gamma', e =\eabs{x}{e_0},S = \tarrowm{x}{T_x}{T}{m} \\
		& \jval{\eabs{x}{e_0}} & \text{[value]} \\
		& \condv{V}{\eabs{x}{e_0}}{\tarrowm{x}{T_x}{T}{m}} = \top \\
		& \potv{V}{\eabs{x}{e_0}}{\tpot{(\tarrowm{x}{T_x}{T}{m})}{0}} = 0 \le \potc{V}{\Gamma}
		\shortintertext{\bf\textsc{(T-Fix)}}
		& \text{SPS}~e=\efix{f}{x}{e_0},S = \forall\many{\alpha}. \tarrow{x}{T_x}{T} \\
		& \jstyping{\Gamma,\bindvar{f}{ S},\many{\alpha},\bindvar{x}{T_x} }{e_0}{T} & \text{[premise]} \label{eq:prog:fixlambda} \\
		& \jval{\efix{f}{x}{e_0}} & \text{[value]} \\
		& \condv{V}{\efix{f}{x}{e_0}}{S} = \top \\
		& \potv{V}{\efix{f}{x}{e_0}}{S} = 0 \le \potc{V}{\Gamma}
		%
		%
		\shortintertext{\bf\textsc{(S-Gen)}}
		& \text{SPS}~e=v,S=\forall\beta.S' \\
		& \jstyping{\Gamma,\beta}{v}{S'} & \text{[premise]} \label{eq:prog:genprem} \\
		& \jval{v} & \text{[premise]} \\
		& \potv{V}{v}{\forall\beta.S'} = 0 \le \potc{V}{\Gamma} \\
		& \textbf{for all $\jwftype{\Gamma}{\tpot{\tsubset{B}{\psi}}{\phi}}$} \\
		& \enskip \textbf{let}~V'=V[\beta \mapsto \tpot{\tsubset{B}{\psi}}{\phi}] \\
		& \enskip \potc{V'}{\Gamma,\beta} = \potc{V}{\Gamma} \\
		& \enskip \text{ind. hyp. on \eqref{eq:prog:genprem} with $p \ge \potc{V'}{\Gamma,\beta}$} \\
		& \enskip \textbf{case}~\jstep{v}{e'}{p}{p'} \\
		& \quad \text{contradict}~\jval{v} \\
		& \enskip \textbf{case}~\jval{v} \\
		& \quad \condv{V'}{v}{S'} = \top & \text{[ind. hyp.]} \\
		& \implies \condv{V}{v}{\forall\beta.S'} = \top
		\shortintertext{\bf\textsc{(S-Inst)}}
		& \text{SPS}~S=\subst{\tpot{\tsubset{B}{\psi}}{\phi}}{\alpha'}{S'} \\
		& \jstyping{\Gamma}{e}{\forall\alpha'.S'} & \text{[premise]} \label{eq:prog:instprem} \\
		& \text{ind.~hyp. on \eqref{eq:prog:instprem} with $p \ge \potc{V}{\Gamma}$} \\
		& \textbf{case}~\jstep{e}{e'}{p}{p'} \\
		& \enskip \text{done} \\
		& \textbf{case}~\jval{e} \\
		& \enskip \condv{V}{e}{\forall\alpha'.S'} = \top & \text{[ind. hyp.]} \\
		& \enskip  \condv{V[\alpha' \mapsto \tpot{\tsubset{B}{\psi}}{\phi}]}{e}{S'} = \top \\
		& \enskip  \condv{V}{e}{\subst{\tpot{\tsubset{B}{\psi}}{\phi}}{\alpha'}{S'}} = \top \\
		& \enskip \jsharing{\Gamma,\alpha'}{S'}{S'}{S'} & \text{[wellformed.]} \\
		& \enskip  \potv{V[\alpha' \mapsto \tpot{\tsubset{B}{\psi}}{\phi}]}{e}{S'} = 0 & \text{[Prop.~\ref{prop:sharingsplit}]} \\
		& \enskip  \potv{V}{e}{\subst{\tpot{\tsubset{B}{\psi}}{\phi}}{\alpha'}{S'}} = 0 \\
		& \enskip  \potv{V}{e}{\subst{\tpot{\tsubset{B}{\psi}}{\phi}}{\alpha'}{S'}} \le \potc{V}{\Gamma}
		\shortintertext{\bf\textsc{(S-Subtype)}}
		& \text{SPS}~S=T_2 \\
		& \jstyping{\Gamma}{e}{T_1} & \text{[premise]} \label{eq:prog:subtypeprem} \\
		& \jsubty{\Gamma}{T_1}{T_2} & \text{[premise]} \label{eq:prog:subtyperel} \\
		& \text{ind. hyp. on \eqref{eq:prog:subtypeprem} with $p \ge \potc{V}{\Gamma}$} \\
		& \textbf{case}~\jstep{e}{e'}{p}{p'} \\
		& \enskip \text{done} \\
		& \textbf{case}~\jval{e} \\
		& \enskip \condv{V}{e}{T_1} = \top & \text{[ind. hyp.]} \\
		& \enskip \condv{V}{e}{T_1} \implies \condv{V}{e}{T_2} & \text{[Prop.~\ref{prop:subtyping}, \eqref{eq:prog:subtyperel}]} \\
		& \enskip \condv{V}{e}{T_2} = \top \\
		& \enskip \potv{V}{e}{T_1} \le \potc{V}{\Gamma} & \text{[ind. hyp.]} \\ 
		& \enskip \condv{V}{e}{T_1} \implies ( \potv{V}{e}{T_1} \ge \potv{V}{e}{T_2}) & \text{[Prop.~\ref{prop:subtyping}, \eqref{eq:prog:subtyperel}]} \\
		& \enskip \potv{V}{e}{T_2} \le \potc{V}{\Gamma}
		\shortintertext{\bf\textsc{(S-Transfer)}}
		& \jstyping{\Gamma'}{e}{S} & \text{[premise]} \label{eq:prog:transe} \\
		& \jprop{\Gamma}{\pot{\Gamma}=\pot{\Gamma'}} & \text{[premise]} \\
		& \Gamma'=\overline{q' \mid \alpha} \wedge \potc{V}{\Gamma}=\potc{V}{\Gamma'} \label{eq:prog:transequal} \\
		& p \ge \potc{V}{\Gamma'} \label{eq:prog:transindhyp} \\
		& \text{ind. hyp. on \eqref{eq:prog:transe} with \eqref{eq:prog:transindhyp}} \\
		& \textbf{case}~\jstep{e}{e'}{p}{p'} \\
		& \enskip \text{done} \\
		& \textbf{case}~\jval{e} \\
		& \enskip \condv{V}{e}{S} = \top & \text{[ind. hyp.]} \\
		& \enskip \potv{V}{e}{S} \le \potc{V}{\Gamma'} & \text{[ind. hyp.]} \\
		& \enskip \potv{V}{e}{S} \le \potc{V}{\Gamma} & \text{[\eqref{eq:prog:transequal}]}
		\shortintertext{\bf\textsc{(S-Relax)}}
		& \text{SPS}~\Gamma=(\Gamma',\phi'), S=\tpot{R}{\phi+\phi'} \\
		& \jstyping{\Gamma'}{e}{\tpot{R}{\phi}} & \text{[premise]}  \label{eq:prog:relaxprem} \\
		& p \ge \potc{V}{\Gamma',\phi'} = \potc{V}{\Gamma'}+\phi' \label{eq:prog:relaxindhyp} \\
		& \text{ind. hyp. on \eqref{eq:prog:relaxprem} with \eqref{eq:prog:relaxindhyp}} \\
		& \textbf{case}~\jstep{e}{e'}{p}{p'} \\
		& \enskip \text{done} \\
		& \textbf{case}~\jval{e} \\
		& \enskip \condv{V}{e}{R} = \top & \text{[ind. hyp.]} \\
		& \enskip \potv{V}{e}{\tpot{R}{\phi}} \le \potc{V}{\Gamma'} & \text{[ind. hyp.]} \\
		& \enskip \potv{V}{e}{\tpot{R}{\phi+\phi'}} \le \potc{V}{\Gamma',\phi'} & \text{[\eqref{eq:prog:relaxindhyp}]}
	\end{alignat}
\end{proof}

\subsection{Substitution}

\begin{proposition}\label{prop:ctxweaken}
	If $\jstyping{\Gamma}{e}{S}$ and $\jwfctxt{\Gamma,\Gamma'}$, then $\jstyping{\Gamma,\Gamma'}{e}{S}$.
\end{proposition}
\begin{proof}
	By induction on $\jstyping{\Gamma}{e}{S}$.
\end{proof}

\begin{proposition}\label{prop:ctxrelax}
	If $\jstyping{\Gamma_1}{e}{S}$ and $\jctxsharing{\Gamma}{\Gamma_1}{\Gamma_2}$, then $\jstyping{\Gamma}{e}{S}$.
\end{proposition}
\begin{proof}
	By induction on $\jstyping{\Gamma_1}{e}{S}$.
\end{proof}

\begin{proposition}\label{prop:typingvalinterp}
	If $\jstyping{\Gamma}{v}{\tpot{\tsubset{B}{\psi}}{\phi}}$ and $\jval{v}$, then $\jstyping{\Gamma}{v}{\tpot{\tsubset{B}{\nu = \calI(v)}}{\phi}}$.
\end{proposition}
\begin{proof}
	By induction on $\jstyping{\Gamma}{v}{\tpot{\tsubset{B}{\psi}}{\phi}}$.
\end{proof}

\begin{proposition}\label{prop:freetofree}
	If $\jstyping{\Gamma}{v}{\tpot{R}{\phi}}$ and $\jval{v}$, then $\jprop{\Gamma}{\pot{\Gamma} \ge \subst{\calI(v)}{\nu}{\phi}}$.
\end{proposition}
\begin{proof}
	By induction on $\jstyping{\Gamma}{v}{\tpot{R}{\phi}}$.
\end{proof}

\begin{proposition}\label{prop:typingvalsharingsplit}
	If $\jstyping{\Gamma}{v}{S}$, $\jsharing{\Gamma}{S}{S_1}{S_2}$ and $\jval{v}$, then there exist $\Gamma_1$ and $\Gamma_2$ such that $\jctxsharing{\Gamma}{\Gamma_1}{\Gamma_2}$, and $\jstyping{\Gamma_1}{v}{S_1}$, $\jstyping{\Gamma_2}{v}{S_2}$.
\end{proposition}
\begin{proof}
	By induction on $\jstyping{\Gamma}{v}{S}$.
\end{proof}

\begin{proposition}\label{prop:typingvalzero}
	If $\jstyping{\Gamma}{v}{S}$, $\jsharing{\Gamma}{S}{S}{S}$ and $\jval{v}$, then there exists $\Gamma'$ such that $\jctxsharing{\Gamma}{\Gamma}{\Gamma'}$ (so $\jctxsharing{\Gamma'}{\Gamma'}{\Gamma'}$), and $\jstyping{\Gamma'}{v}{S}$.
\end{proposition}
\begin{proof}
	By induction on $\jstyping{\Gamma}{v}{S}$.
\end{proof}

\begin{lemma}\label{lem:substprop}
	If $\Gamma,\psi,\Gamma' \vdash \calJ$ and $\jprop{\Gamma}{\psi}$, then $\Gamma,\Gamma' \vdash \calJ$.
\end{lemma}
\begin{proof}
	By induction on $\Gamma,\psi,\Gamma' \vdash \calJ$.
\end{proof}

\begin{lemma}\label{lem:substintoref}
	Suppose $\calJ$ is a judgment other than typing.
	\begin{enumerate}
		\item If $\Gamma_1,\bindvar{x}{\tpot{\tsubset{B}{\psi}}{\phi}},\Gamma' \vdash \calJ$, $\jstyping{\Gamma_2}{t}{\tpot{\tsubset{B}{\psi}}{\phi}}$, $\jval{t}$ and $\jctxsharing{\Gamma}{\Gamma_1}{\Gamma_2}$, then $\Gamma,\subst{\calI(t)}{x}{\Gamma'} \vdash \subst{\calI(t)}{x}{\calJ}$.
		\item If $\Gamma_1,\bindvar{x}{S_x},\Gamma' \vdash \calJ$, $S_x$ is non-scalar/poly, $\jstyping{\Gamma_2}{t}{S_x}$, $\jval{t}$ and $\jctxsharing{\Gamma}{\Gamma_1}{\Gamma_2}$, then $\Gamma,\Gamma' \vdash \calJ$.
	\end{enumerate}
\end{lemma}
\begin{proof}
	By induction on $\Gamma,\bindvar{x}{S_x},\Gamma' \vdash \calJ$.
\end{proof}

\begin{lemma}\label{lem:substatom}\
  \begin{enumerate}
    \item If $\jatyping{\Gamma_1,\bindvar{x}{\trefined{B_x}{\psi}{\phi}},\Gamma'}{e}{B}$, $\jstyping{\Gamma_2}{t}{\trefined{B_x}{\psi}{\phi}}$, $\jval{t}$ and $\jctxsharing{\Gamma}{\Gamma_1}{\Gamma_2}$, then $\jatyping{\Gamma,\subst{\calI(t)}{x}{\Gamma'}}{\subst{t}{x}{a}}{\subst{\calI(t)}{x}{B}}$. 
    \item If $\jatyping{\Gamma_1,\bindvar{x}{S_x},\Gamma'}{a}{B}$, $S_x$ is non-scalar/poly, $\jstyping{\Gamma_2}{t}{S_x}$, $\jval{t}$ and $\jctxsharing{\Gamma}{\Gamma_1}{\Gamma_2}$, then $\jatyping{\Gamma,\Gamma'}{\subst{t}{x}{a}}{B}$.
  \end{enumerate}
\end{lemma}
\begin{proof}[Proof of (1)]
  By induction on $\jatyping{\Gamma_1,\bindvar{x}{\trefined{B_x}{\psi}{\phi}},\Gamma'}{a}{B}$:
  \begin{alignat}{2}
    \shortintertext{\bf{\textsc{(SimpAtom-Var)}$=$}}
    & \text{SPS}~a=x, B=B_x\\
		& \subst{t}{x}{a} = t, \subst{\calI(t)}{x}{B} = B_x \\
		& \jstyping{\Gamma}{t}{\tpot{\tsubset{B_x}{\psi}}{\phi}} & \text{[Prop.~\ref{prop:ctxrelax}]} \\
		& \jstyping{\Gamma}{t}{\tpot{\tsubset{B_x}{\nu = \calI(t)}}{\phi}} & \text{[Prop.~\ref{prop:typingvalinterp}]} \\
		& \jstyping{\Gamma,\subst{\calI(t)}{x}{\Gamma'}}{t}{\tpot{\tsubset{B_x}{\nu = \calI(t)}}{\phi}} & \text{[Prop.~\ref{prop:ctxweaken}]} \\
		& \jatyping{\Gamma,\subst{\calI(t)}{x}{\Gamma'}}{t}{B_x} & \text{[typing]}
		\shortintertext{\bf{\textsc{(SimpAtom-Var)}$\neq$}}
		& \text{SPS}~a=y \\
		& \subst{t}{x}{a} = y \\
		& \textbf{case}~y \in \Gamma \\
		& \enskip B = \text{base of}~ \Gamma_1(y) \\
		& \enskip \jsharing{\Gamma}{\Gamma(y)}{\Gamma_1(y)}{\Gamma_2(y)} \\
		& \enskip \Gamma(y) = \trefined{B}{\psi'}{\phi'} & \\
		& \enskip \jatyping{\Gamma,\subst{\calI(t)}{x}{\Gamma'}}{y}{B} & \text{[typing]} \\
		& \textbf{case}~y \in \Gamma' \\
		& \enskip B = \text{base of}~\Gamma'(y), \Gamma'(y) = \trefined{B}{\psi'}{\phi'} \\
		& \enskip (\subst{\calI(t)}{x}{\Gamma'})(y) =\\
		& \enskip \quad \trefined{\subst{\calI(t)}{x}{B}}{\subst{\calI(t)}{x}{\psi'}}{\subst{\calI(t)}{x}{\phi'}} \\
		& \enskip \jatyping{\Gamma,\subst{\calI(t)}{x}{\Gamma'}}{y}{\subst{\calI(t)}{x}{B}} & \text{[typing]}
    %
    \shortintertext{\bf\textsc{(SimpAtom-ConsD)}}
    & \text{SPS}~a=C_j(a_0,\tuple{a_1,\cdots,a_{m_j}}),B=\tinduct{C}{T}{m}{\theta} \\
    & \vdash \Gamma_1,\bindvar{x}{\trefined{B_x}{\psi}{\phi}},\Gamma'',\pi.\mathbf{j}(\calI(a_0))(\theta) \sharing \\
    & \quad \Gamma_{11},\bindvar{x}{\trefined{B_1}{\psi}{\phi_1},\Gamma_1'},\phi_1' \mid \\
    & \quad \Gamma_{12},\bindvar{x}{\trefined{B_2}{\psi}{\phi_2},\Gamma_2'},\phi_2' & \text{[premise]} \\
    & \jstyping{\Gamma_{11},\bindvar{x}{\trefined{B_1}{\psi}{\phi_1}},\Gamma_1',\phi_1'}{a_0}{T_j} & \text{[premise]}  \label{eq:substa:conspremh}\\
    & \jatyping{\Gamma_{12},\bindvar{x}{\trefined{B_2}{\psi}{\phi_2}},\Gamma_2',\phi_2'}{ \tuple{a_1,\cdots,a_{m_j}} }{ \textstyle\prod_{i=1}^{m_j} \tinduct{C}{T}{m}{ \lhd.\mathbf{j}(\calI(a_0))(\theta).\mathbf{i} } } & \text{[premise]} \label{eq:substa:conspremt} \\
    & \text{exist}~\Gamma_{21},\Gamma_{22}~\text{s.t.}~\jctxsharing{\Gamma_2}{\Gamma_{21}}{\Gamma_{22}}, \\
    & \quad \jstyping{\Gamma_{21}}{t}{\trefined{B_1}{\psi}{\phi_1}},\jstyping{\Gamma_{22}}{t}{\trefined{B_2}{\psi}{\phi_2}} & \text{[\propref{typingvalsharingsplit}]} \\
    & \sharing(\Gamma_{11},\Gamma_{21}), \subst{\calI(t)}{x}{(\Gamma_1',\phi_1')} \vdash \\
    & \quad \subst{t}{x}{a_0} \dblcolon \subst{\calI(t)}{x}{T_j} & \text{[\theoref{substitution}, \eqref{eq:substa:conspremh}]} \\
    & \text{ind. hyp. on \eqref{eq:substa:conspremt}} \\
    & \sharing(\Gamma_{12},\Gamma_{22}), \subst{\calI(t)}{x}{(\Gamma_2',\phi_2')} \vdash \\
    & \quad \subst{t}{x}{\tuple{a_1,\cdots,a_{m_j}}} : \tinduct{C}{\subst{\calI(t)}{x}{T}}{m}{\subst{\calI(t)}{x}{\theta}} \\
    & \jctxsharing{\Gamma}{\Gamma_1}{\Gamma_2} \implies & \\
    & \quad \jctxsharing{\Gamma}{(\sharing(\Gamma_{11},\Gamma_{21})}{\sharing(\Gamma_{12},\Gamma_{22})} \\
    & \Gamma,\bindvar{x}{\trefined{B_x}{\psi}{\phi}} \vdash \Gamma',\pi.\mathbf{j}(\calI(a_0))(\theta) \sharing \Gamma_1',\phi_1'\mid \Gamma_2',\phi_2' \implies  \\
    & \quad \Gamma \vdash \subst{\calI(t)}{x}{(\Gamma'',\pi.\mathbf{j}(\calI(a_0))(\theta))} \sharing \subst{\calI(t)}{x}{(\Gamma_1',\phi_1')} \mid \subst{\calI(t)}{x}{(\Gamma_2',\phi_2')} & \text{[\lemref{substintoref}]} \\
    & {\Gamma,\subst{\calI(t)}{x}{\Gamma''},{\pi.\mathbf{j}(\calI([t/x]a_0))(\subst{\calI(t)}{x}{\theta})}} \vdash [t/x]a : \\
    & \quad \tinduct{C}{\subst{\calI(t)}{x}{T}}{m}{\lhd.\mathbf{j}(\calI([t/x]a_0))(\subst{\calI(t)}{x}{\theta}).\mathbf{i} }  & \text{[typing]}
  \end{alignat}
\end{proof}

\begin{theorem}[Substitution]\label{the:substitution}\
	\begin{enumerate}
		\item If $\jstyping{\Gamma_1,\bindvar{x}{\tpot{\tsubset{B}{\psi}}{\phi}},\Gamma'}{e}{S}$, $\jstyping{\Gamma_2}{t}{\tpot{\tsubset{B}{\psi}}{\phi}}$, $\jval{t}$ and $\jctxsharing{\Gamma}{\Gamma_1}{\Gamma_2}$, then $\jstyping{\Gamma,\subst{\calI(t)}{x}{\Gamma'}}{\subst{t}{x}{e}}{\subst{\calI(t)}{x}{S}}$.
		\item If $\jstyping{\Gamma_1,\bindvar{x}{S_x},\Gamma'}{e}{S}$, $S_x$ is non-scalar/poly, $\jstyping{\Gamma_2}{t}{S_x}$, $\jval{t}$ and $\jctxsharing{\Gamma}{\Gamma_1}{\Gamma_2}$, then $\jstyping{\Gamma,\Gamma'}{\subst{t}{x}{e}}{S}$.
	\end{enumerate}
\end{theorem}
\begin{proof}[Proof of (1)]
	By induction on $\jstyping{\Gamma_1,\bindvar{x}{\tpot{\tsubset{B}{\psi}}{\phi}},\Gamma'}{e}{S}$:
	\begin{alignat}{2}
	     \shortintertext{\bf\textsc{(T-SimpAtom)}}
	     & \text{SPS}~e=a,S=\tsubset{B'}{\nu = \calI(a)} \\
	     & \jatyping{\Gamma_1,\bindvar{x}{\trefined{B}{\psi}{\phi}},\Gamma'}{a}{B'} & \text{[premise]} \\
	     & \jatyping{\Gamma,\subst{\calI(t)}{x}{\Gamma'}}{\subst{t}{x}{a}}{\subst{\calI(t)}{x}{B'}} & \text{[\lemref{substatom}]} \\
	     &  {\Gamma,\subst{\calI(t)}{x}{\Gamma'}} \vdash {\subst{t}{x}{a}} \dblcolon \\
	     & \quad {\tsubset{\subst{\calI(t)}{x}{B'}}{\nu = \calI(\subst{t}{x}{a})}} & \text{[typing]} \\
	     & \tsubset{\subst{\calI(t)}{x}{B'}}{\nu = \calI(\subst{t}{x}{a})} = \\
	     & \quad \subst{\calI(t)}{x}{(\tsubset{B'}{\nu = \calI(a)})}
		\shortintertext{\bf{\textsc{(T-Var)}$=$}}
		& \text{SPS}~e=x,S= \tpot{\tsubset{B}{\psi}}{\phi} \\
		& \subst{t}{x}{e} = t, \subst{\calI(t)}{x}{S} = \tpot{\tsubset{B}{\psi}}{\phi} \\
		& \jstyping{\Gamma}{t}{\tpot{\tsubset{B}{\psi}}{\phi}} & \text{[Prop.~\ref{prop:ctxrelax}]} \\
		& \jstyping{\Gamma,\subst{\calI(t)}{x}{\Gamma'}}{t}{\tpot{\tsubset{B}{\psi}}{\phi}} & \text{[Prop.~\ref{prop:ctxweaken}]}
		\shortintertext{\bf{\textsc{(T-Var)}$\ne$}}
		& \text{SPS}~e=y,S=\Gamma(y) \\
		& \subst{t}{x}{e} = y \\
		& \textbf{case}~y \in \Gamma \\
		& \enskip \text{WLOG}~\Gamma(y) = \tpot{\tsubset{B'}{\psi'}}{\phi'} \\
		& \enskip \jsharing{\Gamma}{\Gamma(y)}{\Gamma_1(y)}{\Gamma_2(y)} \\
		& \enskip \textbf{let}~\Gamma_1(y) = \tpot{\tsubset{B_1'}{\psi_1'}}{\phi_1'} \\
		& \enskip \subst{\calI(t)}{x}{S} = S = \tpot{\tsubset{B_1'}{\psi_1'}}{\phi_1'} \\
		& \enskip \jstyping{\Gamma,\subst{\calI(t)}{x}{\Gamma'}}{y}{\tpot{\tsubset{B'}{\psi'}}{\phi'}} & \text{[typing]} \\
		& \textbf{case}~y \in \Gamma' \\
		& \enskip \text{WLOG}~\Gamma'(y) = \tpot{\tsubset{B'}{\psi'}}{\phi'} \\
		& \enskip S = \tpot{\tsubset{B'}{\psi'}}{\phi'} \\
		& \enskip \subst{\calI(t)}{x}{S} = \\
		& \enskip \quad \tpot{\tsubset{\subst{\calI(t)}{x}{B'}}{\subst{\calI(t)}{x}{\psi'}}}{\subst{\calI(t)}{x}{\phi'}} \\
		& \enskip (\subst{\calI(t)}{x}{\Gamma'})(y) = \\
		& \enskip \quad \tpot{\tsubset{\subst{\calI(t)}{x}{B'}}{\subst{\calI(t)}{x}{\psi'}}}{\subst{\calI(t)}{x}{\phi'}} \\
		& \enskip \Gamma,\subst{\calI(t)}{x}{\Gamma'} \vdash y \dblcolon \\
		& \enskip \quad  \tpot{\tsubset{\subst{\calI(t)}{x}{B'}}{ \subst{\calI(t)}{x}{\psi'} }}{\subst{\calI(t)}{x}{\phi'}} & \text{[typing]}
		\shortintertext{\bf\textsc{(T-Imp)}}
		& \text{SPS}~e=\eimp, S=T \\
		& \subst{t}{x}{e} = \eimp \\
		& \subst{\calI(t)}{x}{S} = \subst{\calI(t)}{x}{T} \\
		& \jprop{\Gamma_1,\bindvar{x}{\tpot{\tsubset{B}{\psi}}{\phi}},\Gamma'}{\bot} & \text{[premise]} \label{eq:subst:impbot} \\
		& \jwftype{\Gamma_1,\bindvar{x}{\tpot{\tsubset{B}{\psi}}{\phi}},\Gamma'}{T} & \text{[premise]} \label{eq:subst:impwftype} \\
		& \jprop{\Gamma,\subst{\calI(t)}{x}{\Gamma'}}{\bot} & \text{[Lem.~\ref{lem:substintoref}, \eqref{eq:subst:impbot}]} \\
		& \jwftype{\Gamma,\subst{\calI(t)}{x}{\Gamma'}}{\subst{\calI(t)}{x}{T}} & \text{[Lem.~\ref{lem:substintoref}, \eqref{eq:subst:impwftype}]} \\
		& \jstyping{\Gamma,\subst{\calI(t)}{x}{\Gamma'}}{\eimp}{\subst{\calI(t)}{x}{T}} & \text{[typing]}
		\shortintertext{\bf\textsc{(T-Consume-P)}}
		& \text{SPS}~e=\econsume{c}{e_0}, c\ge 0, S=T \\
		& \text{SPS}~\Gamma'= \Gamma'',c & \text{[premise]} \\
		& \subst{t}{x}{e} = \econsume{c}{\subst{t}{x}{e_0}} \\
		& \subst{\calI(t)}{x}{S} = \subst{\calI(t)}{x}{T} \\
		& \jstyping{\Gamma_1,\bindvar{x}{\tpot{\tsubset{B}{\psi}}{\phi}},\Gamma''}{e_0}{T} & \text{[premise]} \label{eq:subst:consumepe0} \\
		& \text{ind. hyp. on \eqref{eq:subst:consumepe0}} \\
		& \jstyping{\Gamma,\subst{\calI(t)}{x}{\Gamma''}}{\subst{t}{x}{e_0}}{\subst{\calI(t)}{x}{T}} \\
		& \jstyping{\Gamma,\subst{\calI(t)}{x}{\Gamma''},c}{\econsume{c}{\subst{t}{x}{e_0}}}{\subst{\calI(t)}{x}{T}} & \text{[typing]} \\
		& \jstyping{\Gamma,\subst{\calI(t)}{x}{\Gamma',c}}{\econsume{c}{\subst{t}{x}{e_0}}}{\subst{\calI(t)}{x}{T}}
		\shortintertext{\bf\textsc{(T-Consume-N)}}
		& \text{SPS}~e=\econsume{c}{e_0},c<0,S=T \\
		& \subst{t}{x}{e} = \econsume{c}{\subst{t}{x}{e_0}} \\
		& \subst{\calI(t)}{x}{S} = \subst{\calI(t)}{x}{T} \\
		& \jstyping{\Gamma_1,\bindvar{x}{\tpot{\tsubset{B}{\psi}}{\phi}},\Gamma',-c}{e_0}{T} & \text{[premise]} \label{eq:subst:consumene0} \\
		& \text{ind. hyp. on \eqref{eq:subst:consumene0}} \\
		& \jstyping{\Gamma,\subst{\calI(t)}{x}{\Gamma'},-c}{\subst{t}{x}{e_0}}{\subst{\calI(t)}{x}{T}} \\
		& \jstyping{\Gamma,\subst{\calI(t)}{x}{\Gamma'}}{\econsume{c}{\subst{t}{x}{e_0}}}{\subst{\calI(t)}{x}{T}} & \text{[typing]}
		\shortintertext{\bf\textsc{(T-Cond)}}
		& \text{SPS}~e=\econd{a_0}{e_1}{e_2},S=T \\
		& \subst{t}{x}{e} = \econd{\subst{t}{x}{a_0}}{\subst{t}{x}{e_1}}{\subst{t}{x}{e_2}} \\
		& \subst{\calI(t)}{x}{S} = \subst{\calI(t)}{x}{T} \\
		& \jatyping{\Gamma_{1},\bindvar{x}{\tpot{\tsubset{B}{\psi}}{\phi}},\Gamma'}{a_0}{\tbool} & \text{[premise]} \label{eq:subst:conde0} \\
		& \Gamma_{1},\bindvar{x}{\tpot{\tsubset{B}{\psi}}{\phi}},\Gamma',\calI(a_0) \vdash \\
		& \quad 	e_1 \dblcolon T & \text{[premise]} \label{eq:subst:conde1} \\
		& \Gamma_{1},\bindvar{x}{\tpot{\tsubset{B}{\psi}}{\phi}},\Gamma', \neg\calI(a_0) \vdash \\
		& \quad e_2 \dblcolon T & \text{[premise]} \label{eq:subst:conde2} \\
       & \sharing(\Gamma_{1},\Gamma_{2}), \subst{\calI(t)}{x}{\Gamma'} \vdash \subst{t}{x}{a_0} : \tbool & \text{[\lemref{substatom}]} \label{eq:subst:conde0subst}  \\
		& \text{ind. hyp. on \eqref{eq:subst:conde1}} \\
		& \sharing(\Gamma_{1},\Gamma_{2}), \subst{\calI(t)}{x}{\Gamma'}, \subst{\calI(t)}{x}{\calI(a_0)} \vdash \\
		& \quad \subst{t}{x}{e_1} \dblcolon \subst{\calI(t)}{x}{T} \label{eq:subst:conde1subst} \\
		& \text{ind. hyp. on \eqref{eq:subst:conde2}} \\
		& \sharing(\Gamma_{1},\Gamma_{2}), \subst{\calI(t)}{x}{\Gamma'}, \subst{\calI(t)}{x}{\neg\calI(a_0)} \vdash \\
		& \quad \subst{t}{x}{e_2} \dblcolon \subst{\calI(t)}{x}{T} \label{eq:subst:conde2subst} \\
		& \text{typing on \eqref{eq:subst:conde0subst}, \eqref{eq:subst:conde1subst}, \eqref{eq:subst:conde2subst}} \\
		& \Gamma,\subst{\calI(t)}{x}{\Gamma'} \vdash \\
		& \quad \econd{\subst{t}{x}{e_0}}{\subst{t}{x}{e_1}}{\subst{t}{x}{e_2}} \dblcolon \subst{\calI(t)}{x}{T}
		\shortintertext{\bf\textsc{(T-MatP)}}
		& \text{SPS}~e=\ematp{a_0}{x_1}{x_2}{e_1}, S=T \\
		& [t/x]e = \ematp{[t/x]a_0}{x_1}{x_2}{[t/x]e_1} \\
		& [\calI(t)/x]S = [\calI(t)/x] T \\
		& \vdash \Gamma_1,\bindvar{x}{\tpot{\tsubset{B}{\psi}}{\phi}},\Gamma' \sharing \\
		& \quad \Gamma_{11},\bindvar{x}{\tpot{\tsubset{B_1}{\psi}}{\phi_1}},\Gamma_1' \mid \\
		& \quad \Gamma_{12},\bindvar{x}{\tpot{\tsubset{B_2}{\psi}}{\phi_2}},\Gamma_2' & \text{[premise]} \\
		& \jatyping{\Gamma_{11},x:\trefined{B_1}{\psi}{\phi_1},\Gamma'_1}{a_0}{\tprod{A_1}{A_2}} & \text{[premise]} \\
		& \jatyping{\Gamma_{12},x:\trefined{B_2}{\psi}{\phi_1},\Gamma'_2, x_1:A_1,x_2:A_2,\calI(a_0)=(x_1,x_2)}{e_1}{T} & \text{[premise]} \label{eq:subst:matpe1} \\
		& \text{exist $\Gamma_{21},\Gamma_{22}$ s.t. $\jctxsharing{\Gamma_2}{\Gamma_{21}}{\Gamma_{22}}$,} \\
		& \quad \jstyping{\Gamma_{21}}{t}{\trefined{B_1}{\psi}{\phi_1}}, \jstyping{\Gamma_{22}}{t}{\trefined{B_2}{\psi}{\phi_2}} & \text{[\propref{typingvalsharingsplit}]}  \\
		& \sharing(\Gamma_{11},\Gamma_{21}),[\calI(t)/x]\Gamma_1' \vdash [t/x]a_0 : [\calI(t)/x]A_1 \times [\calI(t)/x]A_2 & \text{[\lemref{substatom}]} \label{eq:subst:matpe0subst} \\
		& \text{ind. hyp. on \eqref{eq:subst:matpe1} with \eqref{eq:subst:matpe0subst}} \\
		& \sharing(\Gamma_{12},\Gamma_{22}),[\calI(t)/x]\Gamma_2',x_1:[\calI(t)/x]A_1,\\
		& \quad x_2:[\calI(t)/x]A_2,\calI([t/x]a_0)=(x_1,x_2) \vdash \\
		& \quad [t/x]e_1 \dblcolon [\calI(t)/x]T \label{eq:subst:matpe1subst} \\
		& \jctxsharing{\Gamma}{\Gamma_1}{\Gamma_2} \implies \jctxsharing{\Gamma}{( \sharing(\Gamma_{11},\Gamma_{21}) ) }{( \sharing(\Gamma_{12},\Gamma_{22}) )} \\
		& \Gamma,\bindvar{x}{\tpot{\tsubset{B}{\psi}}{\phi}} \vdash \Gamma' \sharing \Gamma'_1 \mid \Gamma'_2 \implies \\
		& \quad \Gamma \vdash \subst{\calI(t)}{x}{\Gamma'} \sharing \subst{\calI(t)}{x}{\Gamma'_1} \mid \subst{\calI(t)}{x}{\Gamma'_2} & \text{[Lem.~\ref{lem:substintoref}]} \\
		& \text{typing on \eqref{eq:subst:matpe0subst}, \eqref{eq:subst:matpe1subst}} \\
		& \Gamma,[\calI(t)/x]\Gamma' \vdash \ematp{[t/x]a_0}{x_1}{x_2}{[t/x]e_1} \dblcolon [\calI(t)/x]T
		\shortintertext{\bf\textsc{(T-MatD)}}
		& \text{SPS}~e=\mathsf{matd}(a_0,\many{C_j(x_0,\tuple{x_1,\cdots,x_{m_j}}).e_j}), ,S=T' \\
		& \subst{t}{x}{e} = \mathsf{matd}(\subst{t}{x}{a_0}, \many{C_j(x_0,\tuple{x_1,\cdots,x_{m_j}}).\subst{t}{x}{e_j} }) \\
		& \subst{\calI(t)}{x}{S} = \subst{\calI(t)}{x}{T'} \\
		& \vdash \Gamma_1,\bindvar{x}{\tpot{\tsubset{B}{\psi}}{\phi}},\Gamma' \sharing \\
		& \quad \Gamma_{11},\bindvar{x}{\tpot{\tsubset{B_1}{\psi}}{\phi_1}},\Gamma_1' \mid \\
		& \quad \Gamma_{12},\bindvar{x}{\tpot{\tsubset{B_2}{\psi}}{\phi_2}},\Gamma_2' & \text{[premise]} \\
		& \jatyping{\Gamma_{11},\bindvar{x}{\tpot{\tsubset{B_1}{\psi}}{\phi_1}},\Gamma'_1}{a_0}{\tinduct{C}{T}{m}{\theta}} & \text{[premise]} \label{eq:subst:matle0} \\
		& \forall j: \Gamma_{12},x:\trefined{B_2}{\psi}{\phi_2} , \Gamma_2', \bindvar{x_0}{T_j},\many{x_i:\tinduct{C}{T}{m}{ \lhd.\mathbf{j}(x_0)(\theta).\mathbf{i} } }, \\
		& \quad \calI(a_0)=\mu.\mathbf{j}(x_0)(\cdots),\pi.\mathbf{j}(x_0)(\theta)  \vdash  {e_j} \dblcolon {T'} & \text{[premise]} \label{eq:subst:matle1} \\
		& \text{exist}~\Gamma_{21},\Gamma_{22}~\text{s.t.}~\jctxsharing{\Gamma_2}{\Gamma_{21}}{\Gamma_{22}}, \\
		& \quad \jstyping{\Gamma_{21}}{t}{\tpot{\tsubset{B_1}{\psi}}{\phi_1}}, \jstyping{\Gamma_{22}}{t}{\tpot{\tsubset{B_2}{\psi}}{\phi_2}} & \text{[Prop.~\ref{prop:typingvalsharingsplit}]} \label{eq:subst:matlsplit} \\
		& \sharing(\Gamma_{11},\Gamma_{21}), \subst{\calI(t)}{x}{\Gamma'_1} \vdash \\
		& \quad \subst{t}{x}{a_0} : \tinduct{C}{\subst{\calI(t)}{x}{T}}{m}{\subst{\calI(t)}{x}{\theta}} & \text{[\lemref{substatom}]} \label{eq:subst:matle0subst} \\
		& \text{ind. hyp. on \eqref{eq:subst:matle1} with \eqref{eq:subst:matlsplit}} \\
		& \forall j: \sharing(\Gamma_{12},\Gamma_{22}), \subst{\calI(t)}{x}{\Gamma'_2}, x_0:\subst{\calI(t)}{x}{T_j}, \\
		& \quad \many{x_i : \tinduct{C}{[\calI(t)/x]T}{m}{ \lhd.\mathbf{j}(x_0)([\calI(t)/x]\theta).\mathbf{i} } }, \\
		& \quad \calI([t/x]a_0)=\mu.\mathbf{j}(x_0)(\cdots), \pi.\mathbf{j}(x_0)([\calI(t)/x]\theta) \vdash \\
		& \quad \subst{t}{x}{e_j} \dblcolon \subst{\calI(t)}{x}{T'} \label{eq:subst:matle1subst} \\
		& \jctxsharing{\Gamma}{\Gamma_1}{\Gamma_2} \implies \\
		& \quad \jctxsharing{\Gamma}{(\sharing(\Gamma_{11},\Gamma_{21}))}{(\sharing(\Gamma_{12},\Gamma_{22}))} \\
		& \Gamma,\bindvar{x}{\tpot{\tsubset{B}{\psi}}{\phi}} \vdash \Gamma' \sharing \Gamma'_1 \mid \Gamma'_2 \implies \\
		& \quad \Gamma \vdash \subst{\calI(t)}{x}{\Gamma'} \sharing \subst{\calI(t)}{x}{\Gamma'_1} \mid \subst{\calI(t)}{x}{\Gamma'_2} & \text{[Lem.~\ref{lem:substintoref}]} \\
		& \text{typing on \eqref{eq:subst:matle0subst}, \eqref{eq:subst:matle1subst}} \\ 
		& \Gamma,\subst{\calI(t)}{x}{\Gamma'} \vdash \\
		& \quad \mathsf{matd}({\subst{t}{x}{a_0}}, \many{C_j(x_0,\tuple{x_1,\cdots,x_{m_j}}).[t/x]e_j } )\\
		& \quad \dblcolon \subst{\calI(t)}{x}{T'}
		\shortintertext{\bf\textsc{(T-Let)}}
		& \text{SPS}~e=\elet{e_1}{y}{e_2},S=T_2 \\
		& \subst{t}{x}{e}=\elet{\subst{t}{x}{e_1}}{y}{\subst{t}{x}{e_2}} \\
		& \subst{\calI(t)}{x}{S} = \subst{\calI(t)}{x}{T_2} \\
		& \vdash \Gamma_1,\bindvar{x}{\tpot{\tsubset{B}{\psi}}{\phi}},\Gamma' \sharing \\
		& \quad \Gamma_{11},\bindvar{x}{\tpot{\tsubset{B_1}{\psi}}{\phi_1}},\Gamma_1' \mid \\
		& \quad \Gamma_{12},\bindvar{x}{\tpot{\tsubset{B_2}{\psi}}{\phi_2}},\Gamma_2' & \text{[premise]} \\
		& \jstyping{\Gamma_{11},\bindvar{x}{\tpot{\tsubset{B_1}{\psi}}{\phi_1}},\Gamma_1'}{e_1}{S_1} & \text{[premise]} \label{eq:subst:lete1} \\
		& \jstyping{\Gamma_{12},\bindvar{x}{\tpot{\tsubset{B_2}{\psi_2}}{\phi_2}},\Gamma_2',\bindvar{y}{S_1}}{e_2}{T_2} & \text{[premise]} \label{eq:subst:lete2} \\
		& \text{exist}~\Gamma_{21},\Gamma_{22}~\text{s.t.}~\jctxsharing{\Gamma_2}{\Gamma_{21}}{\Gamma_{22}}, \\
		& \quad \jstyping{\Gamma_{21}}{t}{\tpot{\tsubset{B_1}{\psi}}{\phi_1}}, \jstyping{\Gamma_{22}}{t}{\tpot{\tsubset{B_2}{\psi}}{\phi_2}} & \text{[Prop.~\ref{prop:typingvalsharingsplit}]} \label{eq:subst:letsplitval} \\
		& \text{ind. hyp. on \eqref{eq:subst:lete1} with \eqref{eq:subst:letsplitval}} \\
		& \jstyping{\sharing(\Gamma_{11},\Gamma_{21}),\subst{\calI(t)}{x}{\Gamma_1'} }{\subst{t}{x}{e_1}}{\subst{\calI(t)}{x}{S_1}} \label{eq:subst:lete1subst} \\
		& \text{ind. hyp. on \eqref{eq:subst:lete2} with \eqref{eq:subst:letsplitval}} \\
		& \sharing(\Gamma_{12},\Gamma_{22}), \subst{\calI(t)}{x}{\Gamma_2'}, \bindvar{y}{\subst{\calI(t)}{x}{S_1}} \vdash \\
		& \quad  \subst{t}{x}{e_2} \dblcolon \subst{\calI(x)}{t}{T_2} \label{eq:subst:lete2subst} \\
		& \jctxsharing{\Gamma}{\Gamma_1}{\Gamma_2} \implies \\
		& \quad \jctxsharing{\Gamma}{(\sharing(\Gamma_{11},\Gamma_{21}))}{(\sharing(\Gamma_{12},\Gamma_{22}))} \\
		& \Gamma,\bindvar{x}{\tpot{\tsubset{B}{\psi}}{\phi	}} \vdash \Gamma' \sharing \Gamma'_1 \mid \Gamma'_2 \implies \\
		& \quad \Gamma \vdash \subst{\calI(t)}{x}{\Gamma'} \sharing \subst{\calI(t)}{x}{\Gamma'_1} \mid \subst{\calI(t)}{x}{\Gamma'_2} & \text{[Lem.~\ref{lem:substintoref}]} \\
		& \text{typing on \eqref{eq:subst:lete1subst}, \eqref{eq:subst:lete2subst}} \\
		& \Gamma,\subst{\calI(t)}{x}{\Gamma'} \vdash \\
		 & \quad \elet{\subst{t}{x}{e_1}}{y}{\subst{t}{x}{e_2}} \dblcolon \subst{\calI(t)}{x}{T_2}
		\shortintertext{\bf\textsc{(T-Abs)}}
		& \text{SPS}~e=\eabs{y}{e_0},S=\tpot{(\tarrow{y}{T_y}{T})}{0} \\
		& \subst{t}{x}{e} = \eabs{y}{\subst{t}{x}{e_0}} \\
		&  \subst{\calI(t)}{x}{S} = \tpot{(\tarrow{y}{\subst{\calI(t)}{x}{T_y}}{\subst{\calI(t)}{x}{T}})}{0} \\
		& \jstyping{\Gamma_1,\bindvar{x}{\tpot{\tsubset{B}{\psi}}{\phi}},\Gamma',\bindvar{y}{T_y}}{e_0}{T} & \text{[premise]} \label{eq:subst:abse0} \\
		& \vdash {\Gamma_1,\bindvar{x}{\tpot{\tsubset{B}{\psi}}{\phi}},\Gamma'} \sharing \\
		& \quad {\Gamma_1,\bindvar{x}{\tpot{\tsubset{B}{\psi}}{\phi}},\Gamma'} \mid \\
		& \quad {\Gamma_1,\bindvar{x}{\tpot{\tsubset{B}{\psi}}{\phi}},\Gamma'} & \text{[premise]} \\
		& \jsharing{\Gamma}{\tpot{\tsubset{B}{\psi}}{\phi}}{\tpot{\tsubset{B}{\psi}}{\phi}}{\tpot{\tsubset{B}{\psi}}{\phi}} \\
		& \text{exist}~\Gamma_2'~\text{s.t.}~\jstyping{\Gamma_2'}{t}{\tpot{\tsubset{B}{\psi}}{\phi}}, \jctxsharing{\Gamma_2}{\Gamma_2}{\Gamma_2'} & \text{[Prop.~\ref{prop:typingvalzero}]} \\
		& \text{ind. hyp. on \eqref{eq:subst:abse0}} \\
		& \sharing(\Gamma_1,\Gamma_2'), \subst{\calI(t)}{x}{\Gamma'}, \bindvar{y}{\subst{\calI(t)}{x}{T_y}} \vdash \\
		& \quad \subst{t}{x}{e_0} \dblcolon \subst{\calI(t)}{x}{T} \label{eq:subst:abse0subst} \\
		& \jctxsharing{(\sharing(\Gamma_1,\Gamma_2'))}{(\sharing(\Gamma_1,\Gamma_2'))}{(\sharing(\Gamma_1,\Gamma_2'))} \\
		& \Gamma, x:\tpot{\tsubset{B}{\psi}}{\phi} \vdash \Gamma' \sharing \Gamma' \mid \Gamma' \implies \\
		& \quad \Gamma \vdash \subst{\calI(t)}{x}{\Gamma'} \sharing \subst{\calI(t)}{x}{\Gamma'} \mid \subst{\calI(t)}{x}{\Gamma'} & \text{[Lem.~\ref{lem:substintoref}]} \\
		& \vdash \sharing(\Gamma_1,\Gamma_2'),\subst{\calI(t)}{x}{\Gamma'} \sharing \\
		& \quad \sharing(\Gamma_1,\Gamma_2'),\subst{\calI(t)}{x}{\Gamma'} \mid \\
		& \quad \sharing(\Gamma_1,\Gamma_2'),\subst{\calI(t)}{x}{\Gamma'} \\
		& \text{typing on \eqref{eq:subst:abse0subst}} \\
		& \sharing(\Gamma_1,\Gamma_2') ,\subst{\calI(v)}{x}{\Gamma'} \vdash \\
		& \quad \eabs{y}{\subst{t}{x}{e_0}} \dblcolon \tarrow{y}{\subst{\calI(t)}{x}{T_y}}{\subst{\calI(t)}{x}{T}} \\
		& \Gamma,\subst{\calI(v)}{x}{\Gamma'} \vdash \\
		& \quad \eabs{y}{\subst{t}{x}{e_0}} \dblcolon \tarrow{y}{\subst{\calI(t)}{x}{T_y}}{\subst{\calI(t)}{x}{T}} & \text{[Prop.~\ref{prop:ctxrelax}]}
		\shortintertext{\bf\textsc{(T-Abs-Lin)}}
		& \text{SPS}~ e=\eabs{y}{e_0},S=\tarrowm{y}{T_y}{T}{m} \\
		& \text{SPS}~\Gamma_1,\bindvar{x}{\trefined{B}{\psi}{\phi}},\Gamma' = \\
		& \quad m \cdot (\Gamma_1'',\bindvar{x}{\trefined{B''}{\psi}{\phi''}},\Gamma''' ) \\
		& \subst{t}{x}{e} = \eabs{y}{\subst{t}{x}{e_0}} \\
		&  \subst{\calI(t)}{x}{S} = \tarrowm{y}{\subst{\calI(t)}{x}{T_y}}{\subst{\calI(t)}{x}{T}}{m} \\
		& \jstyping{\Gamma_1'',\bindvar{x}{\tpot{\tsubset{B''}{\psi}}{\phi''}},\Gamma''',\bindvar{y}{T_y}}{e_0}{T} & \text{[premise]} \label{eq:subst:absline0} \\
		& \text{exist}~\Gamma_2''~\text{s.t.}~\Gamma_2=\sharing(m\cdot\Gamma_2'',\_),~\text{and} & \text{[\propref{typingvalsharingsplit},} \\
		& \quad \jstyping{\Gamma_2''}{t}{\tpot{\tsubset{B''}{\psi}}{\phi''}} & \text{ \ref{prop:typingvalzero}]} \\
		& \text{ind. hyp. on \eqref{eq:subst:absline0}} \\
		& \sharing(\Gamma_1'',\Gamma_2''), \subst{\calI(t)}{x}{\Gamma'''}, \bindvar{y}{\subst{\calI(t)}{x}{T_y}} \vdash \\
		& \quad \subst{t}{x}{e_0} \dblcolon \subst{\calI(t)}{x}{T} \label{eq:subst:absline0subst} \\
		& \text{typing on \eqref{eq:subst:absline0subst}} \\
		& m \cdot (\sharing(\Gamma_1'',\Gamma_2'') ,\subst{\calI(v)}{x}{\Gamma'''}) \vdash \\
		& \quad \eabs{y}{\subst{t}{x}{e_0}} \dblcolon \\
		& \quad \tarrowm{y}{\subst{\calI(t)}{x}{T_y}}{\subst{\calI(t)}{x}{T}}{m} \\
		& \Gamma,\subst{\calI(v)}{x}{\Gamma'} \vdash \\
		& \quad \eabs{y}{\subst{t}{x}{e_0}} \\
		& \quad \dblcolon \tarrowm{y}{\subst{\calI(t)}{x}{T_y}}{\subst{\calI(t)}{x}{T}}{m} 
		\shortintertext{\bf\textsc{(T-Fix)}}
		& \text{SPS}~e=\efix{f}{y}{e_0},S=\forall\many{\alpha}.\tarrow{y}{T_y}{T} \\
		& \subst{t}{x}{e} = \efix{f}{y}{\subst{t}{x}{e_0}} \\
		& \vdash \Gamma_1,\bindvar{x}{\trefined{B}{\psi}{\phi}},\Gamma' \sharing \\
		& \quad \Gamma_1,\bindvar{x}{\trefined{B}{\psi}{\phi}},\Gamma' \mid \\
		& \quad \Gamma_1,\bindvar{x}{\trefined{B}{\psi}{\phi}},\Gamma' & \text{[premise]} \\
		& \jstyping{\Gamma_1,\bindvar{x}{\tpot{\tsubset{B}{\psi}}{\phi}},\Gamma',f:S,\many{\alpha},y:T_y}{e_0}{T} \\
		& \text{ind. hyp.} \\
		& \Gamma,\subst{\calI(t)}{x}{\Gamma'},f:\subst{\calI(t)}{x}{S},y:T_y \vdash \\
		& \quad {\subst{\calI(t)}{x}{e_0}} \dblcolon \subst{\calI(t)}{x}{T} \\
		& \jstyping{\Gamma,\subst{\calI(t)}{x}{\Gamma'}}{\efix{f}{y}{\subst{t}{x}{e_0}}}{{\subst{\calI(t)}{x}{S}}}
		\shortintertext{\bf\textsc{(T-App-SimpAtom)}}
		& \text{SPS}~e=\eapp{\hat{a}_1}{a_2},S=\subst{\calI(a_2)}{y}{T} \\
		& \subst{t}{x}{e} = \eapp{\subst{t}{x}{\hat{a}_1}}{\subst{t}{x}{a_2}} \\
		& \subst{\calI(t)}{x}{S} = \subst{\calI(t)}{x}{\subst{\calI(a_2)}{y}{T}} \\
		& \vdash \Gamma_1,\bindvar{x}{\tpot{\tsubset{B}{\psi}}{\phi}},\Gamma' \sharing \\
		& \quad \Gamma_{11},\bindvar{x}{\tpot{\tsubset{B_1}{\psi}}{\phi_1}},\Gamma_1' \mid \\
		& \quad \Gamma_{12},\bindvar{x}{\tpot{\tsubset{B_2}{\psi}}{\phi_2}},\Gamma_2' & \text{[premise]} \\
		& {\Gamma_{11},\bindvar{x}{\tpot{\tsubset{B_1}{\psi}}{\phi_1}},\Gamma_1'} \vdash {\hat{a}_1} \\
		& \quad \dblcolon {\tarrowm{y}{\trefined{B_y}{\psi_y}{\phi_y}}{T}{1}} & \text{[premise]} \label{eq:subst:appae1} \\ 
		& \jstyping{\Gamma_{12},\bindvar{x}{\tpot{\tsubset{B_2}{\psi}}{\phi_2}},\Gamma_2'}{a_2}{\trefined{B_y}{\psi_y}{\phi_y}} & \text{[premise]} \label{eq:subst:appae2} \\
		& \text{exist}~\Gamma_{21},\Gamma_{22}~\text{s.t.}~\jctxsharing{\Gamma_2}{\Gamma_{21}}{\Gamma_{22}}, \\
		& \quad \jstyping{\Gamma_{21}}{t}{\tpot{\tsubset{B_1}{\psi}}{\phi_1}}, \jstyping{\Gamma_{22}}{t}{\tpot{\tsubset{B_2}{\psi}}{\phi_2}} & \text{[Prop.~\ref{prop:typingvalsharingsplit}]} \label{eq:subst:appasplitval} \\
		& \text{ind. hyp. on \eqref{eq:subst:appae1} with \eqref{eq:subst:appasplitval}} \\
		& \sharing(\Gamma_{11},\Gamma_{21}), \subst{\calI(t)}{x}{\Gamma'_1} \vdash \subst{t}{x}{\hat{a}_1} \dblcolon \\
		& \quad  {\tarrowm{y}{\subst{\calI(t)}{x}{ \trefined{B_y}{\psi_y}{\phi_y} }}{\subst{\calI(t)}{x}{T}}{1}}\label{eq:subst:appae1subst} \\
		& \text{ind. hyp. on \eqref{eq:subst:appae2} with \eqref{eq:subst:appasplitval}} \\
		& \sharing(\Gamma_{12},\Gamma_{22}), \subst{\calI(t)}{x}{\Gamma'_2} \vdash \\
		& \quad \subst{t}{x}{a_2} \dblcolon \subst{\calI(t)}{x}{ \trefined{B_y}{\psi_y}{\phi_y} } \label{eq:subst:appae2subst} \\
		& \jctxsharing{\Gamma}{\Gamma_1}{\Gamma_2} \implies \\
		& \quad \jctxsharing{\Gamma}{(\sharing(\Gamma_{11},\Gamma_{21}))}{(\sharing(\Gamma_{12},\Gamma_{22}))} \\
		& \Gamma,\bindvar{x}{\tpot{\tsubset{B}{\psi}}{\phi	}} \vdash \Gamma' \sharing \Gamma'_1 \mid \Gamma'_2 \implies \\
		& \quad \Gamma \vdash \subst{\calI(t)}{x}{\Gamma'} \sharing \subst{\calI(t)}{x}{\Gamma'_1} \mid \subst{\calI(t)}{x}{\Gamma'_2} & \text{[Lem.~\ref{lem:substintoref}]} \\
		& \text{typing on \eqref{eq:subst:appae1subst}, \eqref{eq:subst:appae2subst}} \\
		& \Gamma,\subst{\calI(t)}{x}{\Gamma'} \vdash \\
		& \quad \eapp{\subst{t}{x}{\hat{a}_1}}{\subst{t}{x}{a_2}} \dblcolon \subst{\calI(t),\calI(a_2)}{x,y}{T}
		\shortintertext{\bf\textsc{(T-App)}}
		& \text{SPS}~e=\eapp{e_1}{e_2},S=T \\
		& \subst{t}{x}{e} = \eapp{\subst{t}{x}{\hat{a}_1}}{\subst{t}{x}{\hat{a}_2}} \\
		& \subst{\calI(t)}{x}{S} = \subst{\calI(t)}{x}{T} \\
		& \vdash \Gamma_1,\bindvar{x}{\tpot{\tsubset{B}{\psi}}{\phi}},\Gamma' \sharing \\
		& \quad \Gamma_{11},\bindvar{x}{\tpot{\tsubset{B_1}{\psi}}{\phi_1}},\Gamma_1' \mid \\
		& \quad \Gamma_{12},\bindvar{x}{\tpot{\tsubset{B_2}{\psi}}{\phi_2}},\Gamma_2' & \text{[premise]} \\
		& \jstyping{\Gamma_{11},\bindvar{x}{\tpot{\tsubset{B_1}{\psi}}{\phi_1}},\Gamma_1'}{\hat{a}_1}{{\tarrowm{y}{T_y}{T}{1}}} & \text{[premise]} \label{eq:subst:appe1} \\ 
		& \jstyping{\Gamma_{12},\bindvar{x}{\tpot{\tsubset{B_2}{\psi}}{\phi_2}},\Gamma_2'}{\hat{a}_2}{T_y} & \text{[premise]} \label{eq:subst:appe2} \\
		& \text{exist}~\Gamma_{21},\Gamma_{22}~\text{s.t.}~\jctxsharing{\Gamma_2}{\Gamma_{21}}{\Gamma_{22}}, \\
		& \quad \jstyping{\Gamma_{21}}{t}{\tpot{\tsubset{B_1}{\psi}}{\phi_1}}, \jstyping{\Gamma_{22}}{t}{\tpot{\tsubset{B_2}{\psi}}{\phi_2}} & \text{[Prop.~\ref{prop:typingvalsharingsplit}]} \label{eq:subst:appsplitval} \\
		& \text{ind. hyp. on \eqref{eq:subst:appe1} with \eqref{eq:subst:appsplitval}} \\
		& \sharing(\Gamma_{11},\Gamma_{21}), \subst{\calI(t)}{x}{\Gamma'_1} \vdash \subst{t}{x}{\hat{a}_1} \dblcolon \\
		& \quad  {\tarrowm{y}{\subst{\calI(t)}{x}{T_y}}{\subst{\calI(t)}{x}{T}}{1}} \label{eq:subst:appe1subst} \\
		& \text{ind. hyp. on \eqref{eq:subst:appe2} with \eqref{eq:subst:appsplitval}} \\
		& \sharing(\Gamma_{12},\Gamma_{22}), \subst{\calI(t)}{x}{\Gamma'_2}\vdash \\
		& \quad \subst{t}{x}{\hat{a}_2} \dblcolon \subst{\calI(t)}{x}{T_y} \label{eq:subst:appe2subst} \\
		& \jctxsharing{\Gamma}{\Gamma_1}{\Gamma_2} \implies \\
		& \quad \jctxsharing{\Gamma}{(\sharing(\Gamma_{11},\Gamma_{21}))}{(\sharing(\Gamma_{12},\Gamma_{22}))} \\
		& \Gamma,\bindvar{x}{\tpot{\tsubset{B}{\psi}}{\phi	}} \vdash \Gamma' \sharing \Gamma'_1 \mid \Gamma'_2 \implies \\
		& \quad \Gamma \vdash \subst{\calI(t)}{x}{\Gamma'} \sharing \subst{\calI(t)}{x}{\Gamma'_1} \mid \subst{\calI(t)}{x}{\Gamma'_2} & \text{[Lem.~\ref{lem:substintoref}]} \\
		& \text{typing on \eqref{eq:subst:appe1subst}, \eqref{eq:subst:appe2subst}} \\
		& \Gamma,\subst{\calI(t)}{x}{\Gamma'} \vdash \\
		& \quad \eapp{\subst{t}{x}{\hat{a}_1}}{\subst{t}{x}{\hat{a}_2}} \dblcolon \subst{\calI(t)}{x}{T}
		%
		%
		\shortintertext{\bf\textsc{(S-Gen)}}
		& \text{SPS}~e=v, S=\forall\alpha.S' \\
		& \subst{t}{x}{e} = \subst{t}{x}{v} \\
		& \subst{\calI(t)}{x}{S} = \forall\alpha. \subst{\calI(t)}{x}{S'} \\
		& \jstyping{\Gamma_1,\bindvar{x}{\trefined{B}{\psi}{\phi}},\Gamma',\alpha}{v}{S'} & \text{[premise]} \\
		& \text{ind. hyp.} \\
		& \jstyping{\Gamma,\subst{\calI(t)}{x}{\Gamma'},\alpha}{\subst{t}{x}{v}}{\subst{\calI(t)}{x}{S'}} \\
		& \jstyping{\Gamma,\subst{\calI(t)}{x}{\Gamma'}}{\subst{t}{x}{v}}{\forall\alpha.\subst{\calI(t)}{x}{S'}} & \text{[typing]}
		\shortintertext{\bf\textsc{(S-Inst)}}
		& \text{SPS}~S=\subst{\trefined{B'}{\psi'}{\phi'}}{\alpha}{S'} \\
		& \subst{\calI(t)}{x}{S} = \\
		& \quad \subst{\subst{\calI(t)}{x}{\trefined{B'}{\psi'}{\phi'}}}{\alpha}{\subst{\calI(t)}{x}{S'}} \\
		& \jstyping{\Gamma_1,\bindvar{x}{\trefined{B}{\psi}{\phi}},\Gamma'}{e}{\forall\alpha.S'} & \text{[premise]} \\
		& \text{ind. hyp.} \\
		& \jstyping{\Gamma,\subst{\calI(t)}{x}{\Gamma'}}{\subst{t}{x}{e}}{\forall\alpha. \subst{\calI(t)}{x}{S'}} \\
		& \jwftype{\Gamma_1,\bindvar{x}{\trefined{B}{\psi}{\phi}},\Gamma'}{\trefined{B'}{\psi'}{\phi'}} & \text{[premise]} \\
		& \jwftype{\Gamma,\subst{\calI(t)}{x}{\Gamma'}}{\subst{\calI(t)}{x}{\trefined{B'}{\psi'}{\phi'}}} & \text{[Lem.~\ref{lem:substintoref}]} \\
		& \Gamma,\subst{\calI(t)}{x}{\Gamma'} \vdash \subst{t}{x}{e} \dblcolon \\
		& \quad  \subst{\subst{\calI(t)}{x}{\trefined{B'}{\psi'}{\phi'}}}{\alpha}{\subst{\calI(t)}{x}{S'}} & \text{[typing]}
		\shortintertext{\bf\textsc{(S-Subtype)}}
		& \text{SPS}~S=T_2 \\
		& \subst{\calI(t)}{x}{S} = \subst{\calI(t)}{x}{T_2} \\
		& \jstyping{\Gamma_1,\bindvar{x}{\trefined{B}{\psi}{\phi}},\Gamma'}{e}{T_1} & \text{[premise]} \\
		& \text{ind. hyp.} \\
		& \jstyping{\Gamma,\subst{\calI(t)}{x}{\Gamma'}}{\subst{t}{x}{e}}{\subst{\calI(t)}{x}{T_1}} \\
		& \jsubty{\Gamma_1,\bindvar{x}{\trefined{B}{\psi}{\phi}},\Gamma'}{T_1}{T_2} & \text{[premise]} \\
		& \jsubty{\Gamma,\subst{\calI(t)}{x}{\Gamma'}}{\subst{\calI(t)}{x}{T_1}}{\subst{\calI(t)}{x}{T_2}} & \text{[Lem.~\ref{lem:substintoref}]} \\
		& \jstyping{\Gamma,\subst{\calI(t)}{x}{\Gamma'}}{\subst{t}{x}{e}}{\subst{\calI(t)}{x}{T_2}} & \text{[typing]}
		\shortintertext{\bf\textsc{(S-Transfer)}}
		& \text{SPS}~\Gamma_o = \Gamma_1',\bindvar{x}{\trefined{B}{\psi}{\phi'}},\Gamma'' \\
		& \textbf{let}~\tilde{\Gamma} = \Gamma_1,\bindvar{x}{\trefined{B}{\psi}{\phi}},\Gamma' \\
		& \jstyping{\Gamma_o}{e}{S} & \text{[premise]} \label{eq:subst:transprem} \\
		& \jprop{\tilde{\Gamma}}{\pot{\tilde{\Gamma}} = \pot{\Gamma_o}} & \text{[premise]} \label{eq:subst:transequal} \\
		& \jsharing{\Gamma}{\trefined{B}{\psi}{\phi}}{\trefined{B_0}{\psi}{\phi}}{\trefined{B}{\psi}{0}} \\
		& \text{Lem.~\ref{prop:sharingsplit}, exist $\Gamma_2'$ and $\Gamma_2''$ s.t.} \\
		& \quad \jctxsharing{\Gamma_2}{\Gamma_2'}{\Gamma_2''} \\
		& \quad \jstyping{\Gamma_2'}{t}{\trefined{B_0}{\psi}{\phi}} \\
		& \qquad \implies \jprop{\Gamma_2}{\pot{\Gamma_2'} \ge \subst{\calI(t)}{\nu}{\phi}}  & \text{[Prop.~\ref{prop:freetofree}]} \\
		& \quad \jstyping{\Gamma_2''}{t}{\trefined{B}{\psi}{0}} \\
		& \jstyping{\Gamma_2'',\subst{\calI(t)}{\nu}{\phi'}}{t}{\trefined{B}{\psi}{\phi'}} & \text{[relax]} \\
		& \text{ind. hyp. on \eqref{eq:subst:transprem}} \\
		& \sharing(\Gamma_1',\Gamma_2'',\subst{\calI(t)}{\nu}{\phi'}) ,\subst{\calI(t)}{x}{\Gamma''} \vdash \\
		& \quad \subst{t}{x}{e} \dblcolon \subst{\calI(t)}{x}{S} \label{eq:subst:transind} \\
		& \text{Lem.~\ref{lem:substintoref} on \eqref{eq:subst:transequal}} \\
		& \Gamma,\subst{\calI(t)}{x}{\Gamma'} \models \\
		& \quad \subst{\calI(t)}{x}{\pot{\tilde{\Gamma}}} = \subst{\calI(t)}{x}{\pot{\Gamma_o}} \\
		& \subst{\calI(t)}{x}{ \pot{\tilde{\Gamma}}  } = \pot{\Gamma_1} +\\
		& \quad  \subst{\calI(t)}{\nu}{\phi} + \pot{\subst{\calI(t)}{x}{\Gamma'}} & \text{[def.]} \label{eq:subst:transsubst1} \\
		& \subst{\calI(t)}{x}{\pot{\Gamma_o}} = \pot{\Gamma_1'} + \\
		& \quad \subst{\calI(t)}{\nu}{\phi'} + \pot{\subst{\calI(t)}{x}{\Gamma''}} & \text{[def.]} \label{eq:subst:transsubst2} \\
		& \pot{\Gamma,\subst{\calI(t)}{x}{\Gamma'}} = \\
		& \quad \pot{\Gamma_1}+\pot{\Gamma_2} + \pot{\subst{\calI(t)}{\nu}{\Gamma'}} = \\
		& \quad \pot{\Gamma_1'}+\pot{\Gamma_2'} + \pot{\Gamma_2''} +\pot{\subst{\calI(t)}{\nu}{\Gamma''}} + \\
		& \qquad \subst{\calI(t)}{\nu}{(\phi' - \phi)} \ge & \text{[\eqref{eq:subst:transsubst1}, \eqref{eq:subst:transsubst2}]} \\
		& \quad \pot{\Gamma_1'}+\pot{\Gamma_2''}+\pot{\subst{\calI(t)}{x}{\Gamma''}} + \\
		& \qquad \subst{\calI(t)}{x}{\phi'} = \\
		& \quad \pot{\sharing(\Gamma_1',\Gamma_2'',\subst{\calI(t)}{\nu}{\phi'} ), \subst{\calI(t)}{x}{\Gamma''}} \\
		& \text{recall \eqref{eq:subst:transind}, and then typing, relax} \\
		& \jstyping{\Gamma,\subst{\calI(t)}{x}{\Gamma'}}{\subst{\calI(t)}{x}{e}}{\subst{\calI(t)}{x}{S}}
		\shortintertext{\bf\textsc{(S-Relax)}}
		& \text{SPS}~S=\tpot{R}{\phi+\phi'} \\
		& \subst{\calI(t)}{x}{S} = \tpot{\subst{\calI(t)}{x}{R}}{\subst{\calI(t)}{x}{\phi} +\subst{\calI(t)}{x}{\phi'}} \\
		& \jstyping{\Gamma_1,\bindvar{x}{\trefined{B}{\psi}{\phi}},\Gamma'}{e}{\tpot{R}{\phi}} & \text{[premise]} \\
		& \text{ind. hyp.} \\
		& \jstyping{\Gamma,\subst{\calI(t)}{x}{\Gamma'}}{\subst{t}{x}{e}}{\tpot{\subst{\calI(t)}{x}{R}}{\subst{\calI(t)}{x}{\phi}}} \\
		& \jsort{\Gamma_1,\bindvar{x}{\trefined{B}{\psi}{\phi}},\Gamma'}{\phi'}{\bbN} & \text{[premise]} \\
		& \jsort{\Gamma,\subst{\calI(t)}{x}{\Gamma'}}{\subst{\calI(t)}{x}{\phi'}}{\bbN} & \text{[Lem.~\ref{lem:substintoref}]} \\
		& \Gamma,\subst{\calI(t)}{x}{\Gamma'},\subst{\calI(t)}{x}{\phi'} \vdash \subst{t}{x}{e} \dblcolon \\
		& \quad \tpot{\subst{\calI(t)}{x}{R}}{\subst{\calI(t)}{x}{\phi}+\subst{\calI(t)}{x}{\phi'}} & \text{[typing]}
	\end{alignat}
\end{proof}

\subsection{Preservation}

\begin{proposition}\label{prop:evaldeter}
	If $\jstep{e}{e'}{p}{p'}$ and $\jstep{e}{e''}{q}{q'}$, then $e' = e''$ and $q - p = q' - p'$.
\end{proposition}
\begin{proof}
	By induction on $\jstep{e}{e'}{p}{p'}$ and then inversion on $\jstep{e}{e''}{q}{q'}$.
\end{proof}

\begin{theorem}[Preservation]\label{the:preservation}
	If $\Gamma=\overline{q}$, $\jstyping{\Gamma}{e}{S}$, $p \ge \potc{\emptyset}{\Gamma}$ and $\jstep{e}{e'}{p}{p'}$, then $\jstyping{p'}{e'}{S}$.
\end{theorem}
\begin{proof}
	By induction on $\jstyping{\Gamma}{e}{S}$:
	\begin{alignat}{2}\footnotesize
		\shortintertext{\bf\textsc{(T-Consume-P)}}
		& \text{SPS}~\Gamma=(\Gamma',c),e=\econsume{c}{e_0},c\ge 0 \\
		& \text{SPS}~S=T \\
		& \jstyping{\Gamma'}{e_0}{T} & \text{[premise]} \label{eq:pres:consumepe0} \\
		& \text{inv. on}~\jstep{e}{e'}{p}{p'} \\
		& e' = e_0, p' = p-c \ge \potc{\emptyset}{\Gamma} - c = \potc{\emptyset}{\Gamma'} \\
		& \jstyping{p'}{e_0}{T} & \text{[relax, \eqref{eq:pres:consumepe0}]}
		\shortintertext{\bf\textsc{(T-Consume-N)}}
		& \text{SPS}~e=\econsume{c}{e_0},c<0,S=T \\
		& \jstyping{\Gamma,-c}{e_0}{T} & \text{[premise]} \label{eq:pres:consumene0} \\
		& \text{inv. on}~\jstep{e}{e'}{p}{p'} \\
		& e' = e_0, p' = p -c \ge \potc{\emptyset}{\Gamma}-c \\
		& \jstyping{p'}{e_0}{T} & \text{[relax, \eqref{eq:pres:consumene0}]}
		\shortintertext{\bf\textsc{(T-Cond)}}
		& \text{SPS}~e=\econd{a_0}{e_1}{e_2},S=T \\
		& \jatyping{\Gamma}{a_0}{ \tbool } & \text{[premise]} \label{eq:pres:conde0} \\
		& \jstyping{\Gamma, \calI(a_0) }{e_1}{T} & \text{[premise]} \label{eq:pres:conde1} \\
		& \jstyping{\Gamma, \neg\calI(a_0)}{e_2}{T} & \text{[premise]} \label{eq:pres:conde2} \\
		& \text{inv. on}~\jstep{e}{e'}{p}{p'} \\
		& \textbf{case}~\jstep{e}{e_1}{p}{p} \\
		& \enskip a_0 = \etrue & \text{[premise]} \\
		& \enskip \calI(a_0) = \top \\
		& \enskip  \jprop{\Gamma}{\top} \\
		& \enskip \jstyping{\Gamma}{e_1}{T} & \text{[\lemref{substprop}, \eqref{eq:pres:conde1}]} \\
		& \enskip p \ge \potc{\emptyset}{\Gamma} & \text{[asm.]} \\
		& \enskip \jstyping{p}{e_1}{T} & \text{[relax]} \\
		& \textbf{case}~\jstep{e}{e_2}{p}{p} \\
		& \enskip a_0 = \efalse & \text{[premise]} \\
		& \enskip \text{similar to $a_0=\etrue$}
		\shortintertext{\bf\textsc{(T-MatP)}}
		& \text{SPS}~e=\ematp{a_0}{x_1}{x_2}{e_1},S=T \\
		& \jctxsharing{\Gamma}{\Gamma_1}{\Gamma_2} & \text{[premise]} \\
		& \quad \implies \potc{\emptyset}{\Gamma} =\potc{\emptyset}{\Gamma_1} +\potc{\emptyset}{\Gamma_2}\label{eq:pres:matpsplit} \\
		& \jatyping{\Gamma_1}{a_0}{\tprod{B_1}{B_2}} & \text{[premise]} \label{eq:pres:matpe0} \\
		& \jstyping{\Gamma_2,x_1:B_1,x_2:B_2,\calI(a_0)=(x_1,x_2)}{e_1}{T} & \text{[premise]} \label{eq:pres:matpe1} \\
		& \text{inv. on $\jstep{e}{e'}{p}{p'}$} \\
		& a_0 = \epair{v_1}{v_2}, e'=\subst{v_1,v_2}{x_1,x_2}{e},p'=p \\
		& \calI(a_0) = (\calI(v_1),\calI(v_2)) \\
		& \jatyping{\Gamma_{11}}{v_1}{B_1}, \jatyping{\Gamma_{12}}{v_2}{B_2}, \jctxsharing{\Gamma_1}{\Gamma_{11}}{\Gamma_{12}} & \text{[inv.]} \label{eq:pres:matpinv} \\
		& \sharing(\Gamma_2,\Gamma_1),\calI(a_0)=(\calI(v_1),\calI(v_2)) \vdash [v_1,v_2/x_1,x_2]e_1 \dblcolon T & \text{[Thm.~\ref{the:substitution}, \eqref{eq:pres:matpe1}, \eqref{eq:pres:matpinv} ]} \\
		& \Gamma \models \calI(a_0) = (\calI(v_1),\calI(v_2)) \\
		& \Gamma \vdash [v_1,v_2/x_1,x_2]e_1 \dblcolon T & \text{[\lemref{substprop}]} \\
		& p \ge \potc{\emptyset}{\Gamma} & \text{[asm., \eqref{eq:pres:matpsplit}]} \\
		& \jstyping{p}{e'}{T'} & \text{[relax]} 
		\shortintertext{\bf\textsc{(T-MatD)}}
		& \text{SPS}~e=\mathsf{matd}(a_0, \many{ C_j(x_0,\tuple{x_1,\cdots,x_{m_j}}).e_j }),S=T' \\
		& \jctxsharing{\Gamma}{\Gamma_1}{\Gamma_2} & \text{[premise]} \\
		& \quad \implies \potc{\emptyset}{\Gamma}=\potc{\emptyset}{\Gamma_1}+\potc{\emptyset}{\Gamma_2} \label{eq:pres:matlsplit} \\
		& \jatyping{\Gamma_1}{a_0}{\tinduct{C}{T}{m}{\theta}} & \text{[premise]} \label{eq:pres:matle0} \\
		& \forall j : \Gamma_2,x_0:T_j,\many{x_i: \tinduct{C}{T}{m}{\lhd.\mathbf{j}(x_0)(\theta).\mathbf{i}}}, \\
		& \quad \calI(a_0)=\mu(C_j(x_0,\tuple{\cdots}))   \vdash {e_j} \dblcolon {T'} & \text{[premise]} \label{eq:pres:matle2} \\
		& \text{inv. on}~\jstep{e}{e'}{p}{p'} \\
& \exists j : \jstep{e}{\subst{v_0,v_1,\cdots,v_{m_j}}{x_0,x_1,\cdots,x_{m_j}}{e_j} }{p}{p} \\
		& a_0 = C_j(v_0,\tuple{v_1,\cdots,v_{m_j}}) & \text{[premise]} \\
          & \calI(a_0) = \mu.\mathbf{j}(\calI(v_0))(\calI(v_1),\cdots,\calI(v_{m_j})) \\
          &  \jstyping{\Gamma_{11}}{v_0}{T}, \\
          & \quad \jatyping{\Gamma_{12}}{\tuple{v_1,\cdots,v_{m_j} }}{ \textstyle\prod_{i=1}^{m_j} \tinduct{C}{T}{m}{ \lhd.\mathbf{j}(\calI(v_0))(\theta).\mathbf{i} } }, \\
          & \quad \Gamma_1=\Gamma_1',\pi.\mathbf{j}(\calI(v_0))(\theta), \jctxsharing{\Gamma_1'}{\Gamma_{11}}{\Gamma_{12}} & \text{[inv.]} \label{eq:pres:matlinv} \\
          &  \sharing(\Gamma_2,\Gamma_1') ,\calI(a_0) = \mu.\mathbf{j}(\calI(v_0))(\calI(v_1),\cdots,\calI(v_{m_j})), \pi.\mathbf{j}(\calI(v_0))(\theta) \vdash & \text{[Thm.~\ref{the:substitution},} \\
          & \quad \subst{v_0,v_1,\cdots,v_{m_j}}{x_0,x_1,\cdots,x_{m_j}}{e_j} \dblcolon T' & \text{\eqref{eq:pres:matle2}, \eqref{eq:pres:matlinv}]}  \\
&  \jprop{\Gamma}{\calI(a_0)=\mu.\mathbf{j}(\calI(v_0))(\calI(v_1),\cdots,\calI(v_{m_j}))} \\
&  \Gamma \vdash \subst{v_0,v_1,\cdots,v_{m_j}}{x_0,x_1,\cdots,x_{m_j}}{e_j} \dblcolon T' & \text{[\lemref{substprop}]} \\
		&  p \ge \potc{\emptyset}{\Gamma} & \text{[asm., \eqref{eq:pres:matlsplit}]} \\
		&  \jstyping{p}{e'}{T'} & \text{[relax]}
		\shortintertext{\bf\textsc{(T-Let)}}
		& \text{SPS}~e=\elet{e_1}{x}{e_2},S=T_2 \\
		& \jctxsharing{\Gamma}{\Gamma_1}{\Gamma_2} \\
		& \quad \implies \potc{\emptyset}{\Gamma} = \potc{\emptyset}{\Gamma_1}+\potc{\emptyset}{\Gamma_2} & \text{[premise]} \label{eq:pres:letsplit} \\
		& \jstyping{\Gamma_1}{e_1}{S_1} & \text{[premise]} \label{eq:pres:lete1} \\
		& \jstyping{\Gamma_2,\bindvar{x}{S_1}}{e_2}{T_2} & \text{[premise]} \label{eq:pres:lete2} \\
		& \text{inv. on}~\jstep{e}{e'}{p}{p'} \\
		& \textbf{case}~\jstep{e}{\elet{e_1'}{x}{e_2}}{p}{p'} \\
		& \enskip \jstep{e_1}{e_1'}{p}{p'} & \text{[premise]} \label{eq:pres:lete1step} \\
		& \enskip p-\potc{\emptyset}{\Gamma_2} \ge \potc{\emptyset}{\Gamma_1} & \text{[asm., \eqref{eq:pres:letsplit}]} \label{eq:pres:letindhyp} \\
		& \enskip \text{Thm.~\ref{the:progress} on \eqref{eq:pres:lete1} with \eqref{eq:pres:letindhyp}} \\
		& \enskip \jstep{e_1}{e_1'}{p-\potc{\emptyset}{\Gamma_2}}{p'-\potc{\emptyset}{\Gamma_2}} & \text{[Prop.~\ref{prop:evaldeter}, \eqref{eq:pres:lete1step}]} \label{eq:pres:lete1step2} \\
		& \enskip \text{ind. hyp. on \eqref{eq:pres:lete1} with \eqref{eq:pres:lete1step2}, \eqref{eq:pres:letindhyp}} \\
		& \enskip \jstyping{p'-\potc{\emptyset}{\Gamma_2}}{e_1'}{S_1} \\
		& \enskip \jstyping{\sharing(p'-\potc{\emptyset}{\Gamma_2},\Gamma_2)}{\elet{e_1'}{x}{e_2}}{T_2} & \text{[typing]} \\
		& \enskip \jstyping{p'}{e'}{T_2} & \text{[transfer]} \\
		& \textbf{case}~\jstep{e}{\subst{e_1}{x}{e_2}}{p}{p} \\
		& \enskip \jval{e_1} & \text{[premise]} \label{eq:pres:lete1val} \\
		& \enskip \jstyping{\sharing(\Gamma_1,\Gamma_2)}{\subst{e_1}{x}{e_2}}{T_2} & \text{[Thm.~\ref{the:substitution}, \eqref{eq:pres:lete2}]} \\
		& \enskip \jstyping{\potc{\emptyset}{\Gamma_1}+\potc{\emptyset}{\Gamma_2}}{e'}{T_2} & \text{[transfer]} \\
		& \enskip \jstyping{p}{e'}{T_2 } & \text{[relax]}
		\shortintertext{\bf\textsc{(T-App-SimpAtom)}}
		& \text{SPS}~e=\eapp{\hat{a}_1}{a_2},S=T \\
		& \jctxsharing{\Gamma}{\Gamma_1}{\Gamma_2} \\
		& \quad \implies \potc{\emptyset}{\Gamma} = \potc{\emptyset}{\Gamma_1}+\potc{\emptyset}{\Gamma_2} & \text{[premise]} \label{eq:pres:appasplit} \\
		& \jstyping{\Gamma_1}{\hat{a}_1}{\tarrowm{x}{\trefined{B_x}{\psi_x}{\phi_x}}{T}{1}} & \text{[premise]} \label{eq:pres:appae1} \\
		& \jstyping{\Gamma_2}{a_2}{\trefined{B_x}{\psi_x}{\phi_x}} & \text{[premise]} \label{eq:pres:appae2} \\
		& \text{inv. on}~\jstep{e}{e'}{p}{p'} \\
		& \textbf{case}~\jstep{e}{\subst{a_2}{x}{e_0}}{p}{p} \\
		& \enskip \hat{a}_1=\eabs{x}{e_0}, \jval{a_2} & \text{[premise]} \\
		& \enskip \text{inv. on}~\eqref{eq:pres:appae1} \\
		& \enskip \jstyping{\Gamma_1, \bindvar{x}{\trefined{B_x}{\psi_x}{\phi_x}}}{e_0}{T} \Omit{\potc{\emptyset}{\Gamma_1} \ge \phi_1} \label{eq:pres:appainvlambda} \\
		& \enskip \jstyping{\Gamma}{\subst{a_2}{x}{e_0}}{\subst{\calI(a_2)}{x}{T}} & \text{[Thm.~\ref{the:substitution}, \eqref{eq:pres:appainvlambda}]} \\
        & \enskip p \ge \potc{\emptyset}{\Gamma} & \text{[asm.]} \\
		& \enskip \jstyping{p}{e'}{T} & \text{[relax]} \\
		& \textbf{case}~\jstep{e}{\subst{e_1,a_2}{f,x}{e_0}}{p}{p} \\
		& \enskip e_1=\efix{f}{x}{e_0}, \jval{a_2} & \text{[premise]} \\
		& \enskip \text{similar to $e_1=\eabs{x}{e_0}$}
		\shortintertext{\bf\textsc{(T-App)}}
		& \text{SPS}~e=\eapp{\hat{a}_1}{\hat{a}_2},S=T \\
		& \jctxsharing{\Gamma}{\Gamma_1}{\Gamma_2} \\
		& \quad \implies \potc{\emptyset}{\Gamma} = \potc{\emptyset}{\Gamma_1}+\potc{\emptyset}{\Gamma_2} & \text{[premise]} \label{eq:pres:appsplit} \\
		& \jstyping{\Gamma_1}{\hat{a}_1}{\tarrowm{x}{T_x}{T}{1}} & \text{[premise]} \label{eq:pres:appe1} \\
		& \jstyping{\Gamma_2}{\hat{a}_2}{T_x} & \text{[premise]} \label{eq:pres:appe2} \\
		& \text{inv. on}~\jstep{e}{e'}{p}{p'} \\
		& \textbf{case}~\jstep{e}{\subst{\hat{a}_2}{x}{e_0}}{p}{p} \\
		& \enskip \hat{a}_1=\eabs{x}{e_0}, \jval{\hat{a}_2} & \text{[premise]} \\
		& \enskip \text{inv. on}~\eqref{eq:pres:appe1} \\
		& \enskip \jstyping{\Gamma_1, \bindvar{x}{T_x}}{e_0}{T} \Omit{\potc{\emptyset}{\Gamma_1} \ge \phi_1} \label{eq:pres:appinvlambda} \\
		& \enskip \jstyping{\Gamma}{\subst{\hat{a}_2}{x}{e_0}}{T} & \text{[Thm.~\ref{the:substitution}, \eqref{eq:pres:appinvlambda}]} \\
        & \enskip p \ge \potc{\emptyset}{\Gamma} & \text{[asm.]} \\
		& \enskip \jstyping{p}{e'}{T} & \text{[relax]} \\
		& \textbf{case}~\jstep{e}{\subst{e_1,\hat{a}_2}{f,x}{e_0}}{p}{p} \\
		& \enskip e_1=\efix{f}{x}{e_0}, \jval{\hat{a}_2} & \text{[premise]} \\
		& \enskip \text{similar to $e_1=\eabs{x}{e_0}$}
		%
		%
		\shortintertext{\bf\textsc{(S-Inst)}}
		& \text{SPS}~S=\subst{\tpot{\tsubset{B}{\psi}}{\phi}}{\alpha}{S'} \\
		& \jstyping{\Gamma}{e}{\forall\alpha.S' } & \text{[premise]} \label{eq:pres:instprem} \\
		& \text{ind. hyp. on \eqref{eq:pres:instprem}} \\
		& \jstyping{p'}{e'}{\forall\alpha.S'} \\
		& \jstyping{p'}{e'}{\subst{\tpot{\tsubset{B}{\psi}}{\phi}}{\alpha}{S'}} & \text{[typing]}
		\shortintertext{\bf\textsc{(S-Subtype)}}
		& \text{SPS}~S=T_2 \\
		& \jstyping{\Gamma}{e}{T_1} & \text{[premise]} \label{eq:pres:subtypeprem} \\
		& \jsubty{\Gamma}{T_1}{T_2} & \text{[premise]} \label{eq:pres:subtyperel} \\
		& \text{ind. hyp. on \eqref{eq:pres:subtypeprem}} \\
		& \jstyping{p'}{e'}{T_1} \\
		& \jstyping{p'}{e'}{T_2} & \text{[typing]}
		\shortintertext{\bf\textsc{(S-Transfer)}}
		& \jstyping{\Gamma'}{e}{S}, \jprop{\Gamma}{\pot{\Gamma}=\pot{\Gamma'}} & \text{[premise]} \label{eq:pres:transprem} \\
		& \Gamma' = \overline{q'} \wedge \potc{\emptyset}{\Gamma}=\potc{\emptyset}{\Gamma'} \\
		& p \ge \potc{\emptyset}{\Gamma'} \label{eq:pres:transindhyp} \\
		& \text{ind. hyp. on \eqref{eq:pres:transprem} with \eqref{eq:pres:transindhyp}} \\
		& \jstyping{p'}{e'}{S}
		\shortintertext{\bf\textsc{(S-Relax)}}
		& \text{SPS}~\Gamma=(\Gamma',\phi'), S=\tpot{R}{\phi+\phi'} \\
		& \jstyping{\Gamma'}{e}{\tpot{R}{\phi}}  & \text{[premise]} \label{eq:pres:relaxprem} \\
		& p -\phi' \ge \potc{\emptyset}{\Gamma'} & \text{[asm.]} \label{eq:pres:relaxindhyp} \\
		& \text{Thm.~\ref{the:progress} on \eqref{eq:pres:relaxprem} with \eqref{eq:pres:relaxindhyp}} \\
		& \jstep{e}{e'}{p-\phi'}{p'-\phi'} & \text{[Prop.~\ref{prop:evaldeter}, asm.]} \label{eq:pres:relaxstep} \\
		& \text{ind. hyp. on \eqref{eq:pres:relaxprem} with \eqref{eq:pres:relaxstep}, \eqref{eq:pres:relaxindhyp}} \\
		& \jstyping{p'-\phi'}{e'}{\tpot{R}{\phi}} \\
		& \jstyping{p'-\phi',\phi'}{e'}{\tpot{R}{\phi+\phi'}} & \text{[relax]} \\
		& \jstyping{p'}{e'}{\tpot{R}{\phi+\phi'}} & \text{[transfer]}
	\end{alignat}
\end{proof}

\fi

\end{document}